%% file: 0-paper.tex
\documentclass[11pt]{article}
\usepackage{fullpage}
\usepackage[hyphens]{url}
\usepackage{hyperref}
\hypersetup{
breaklinks=true,
    colorlinks=true, 
    linkcolor=black, 
    citecolor=black, 
    filecolor=black,
    urlcolor=black 
}
\usepackage[utf8]{inputenc} 
\usepackage[T1]{fontenc}    
\usepackage{booktabs}       
\usepackage{amsfonts}       
\usepackage{nicefrac}       
\usepackage{microtype}      
\usepackage{diagbox}
\usepackage{enumerate}
\usepackage[shortlabels]{enumitem}
\usepackage{subfigure}
\usepackage{tabularx}
\usepackage{verbatim}
\usepackage{afterpage}
\usepackage{float}
\usepackage{wrapfig}
\newcommand{\JRSSBparagraphvspace}[1]{}
\newcommand{\JRSSBfigureaftervspace}[1]{}
\long\def\JRSSBFigureTwo#1{\begin{figure}[t]#1\end{figure}}

\providecommand{\doparttoc}{}
\providecommand{\faketableofcontents}{}
\providecommand{\parttoc}{}
\usepackage{empheq}
\usepackage{cancel}
\usepackage[normalem]{ulem}
\usepackage{tcolorbox}
\tcbuselibrary{skins, breakable}

\usepackage{graphicx} 
\usepackage{caption}

\usepackage{mathrsfs}
\usepackage{amsmath}
\usepackage{amsthm}
\usepackage{amssymb}
\usepackage{tablefootnote}
\usepackage{multirow}
\usepackage{enumerate}
\usepackage{color}
\usepackage{xcolor}
\usepackage{tikz}
\usetikzlibrary{arrows.meta,fit,calc,positioning}

\usepackage[numbers]{natbib}

\usepackage{xparse}

\usepackage{etoolbox}
\DeclareCaptionFont{normalcaption}{\color{black}\hypersetup{linkcolor=black,urlcolor=black}}
\NewDocumentCommand{\new}{+m}{%
  {\begingroup\color{black}#1\endgroup}
}

\newcommand\blfootnote[1]{%
  \begingroup
  \renewcommand\thefootnote{}%
  \footnote{\hspace{-1.2em}#1}%
  \addtocounter{footnote}{-1}%
  \endgroup
}

\usepackage{algorithm}
\usepackage{algpseudocode}
\allowdisplaybreaks[4]
\usepackage{bm}
\allowdisplaybreaks

\newtheorem{thm}{Theorem}[section]
\newtheorem{lem}{Lemma}[section]

\newtheorem{asmp}{Assumption}[section]
\newtheorem{defn}{Definition}[section]

\newcounter{subassumption}[asu]

\makeatletter
\renewcommand{\p@subassumption}{\theasu}
\makeatother

\newtheoremstyle{remarkstyle}
  {}                    
  {}                    
  {\normalfont}         
  {}                    
  {\itshape}            
  {.}                   
  { }                   
  {}                    
\theoremstyle{remarkstyle}
\newtheorem{rem}{Remark}[section]

\definecolor{wjs}{RGB}{200,0,50}

\hypersetup{
colorlinks=true,
filecolor=black,
citecolor=blue,
urlcolor=black,
}

\input{tex/math_commands}

\title{Optimal Watermark Localization in Mixed-Source \\
Large Language Model Texts}

\author{
  {\small
  Jose H.\ Blanchet\textsuperscript{1}\quad
  T.\ Tony Cai\textsuperscript{2}\quad
  Xiang Li\textsuperscript{2}\quad
  Hao Liu\textsuperscript{1}\quad
  Qi Long\textsuperscript{2}\quad
  Weijie J.\ Su\textsuperscript{2}}
  \\[1ex]
  \textsuperscript{1}Stanford University\\
  \textsuperscript{2}University of Pennsylvania\\[2ex]
}

\date{}

\begin{document}

\maketitle

\doparttoc

\begin{abstract}

Watermarking provides a principled way to authenticate text generated by large language models (LLMs). In practice, however, the final text may be mixed-source, with watermark evidence surviving at only a subset of token positions after rewriting, insertion, deletion, or paraphrasing. Although prior work has studied global detection of watermark signals, when such signals can be localized remains unclear. We formulate watermark localization as a token-level multiple-testing problem based on pivotal statistics, with a latent indicator recording whether watermark dependence survives at each position. Under an asymptotic regime indexed by exponents for signal sparsity, next-token concentration, and effective-vocabulary growth, we derive a sharp boundary for global detection and phase transitions for discovery and classification within the class of coordinatewise pivot-based localization rules. We show that discovery is strictly harder than detection and that consistent classification is impossible across the parameter regime within this class. We then develop an adaptive thresholding method that does not require knowledge of the exponents or time-varying next-token distributions, but uses a data-driven estimate of the surviving watermark fraction. The method attains the optimal discovery boundary and near-optimal discovery power relative to homogeneous pivot-based rules. 
Simulations support the theoretical phase transitions, while
experiments on model-generated texts demonstrate practical localization
performance under common edit mechanisms.

\blfootnote{Emails:\;
\texttt{\{jose.blanchet, haoliu20\}@stanford.edu},\;
\texttt{\{tcai, suw\}@wharton.upenn.edu},\;
\texttt{\{lx10077, qlong\}@upenn.edu}.}
\blfootnote{Author names are listed in alphabetical order.}

\end{abstract}

\doparttoc 
\faketableofcontents 

\input{1-intro}

\input{2-prelim}

\input{3-results}

\input{4-theory}

\input{5-experiments}

\input{6-discuss}


\section*{Acknowledgments}
We thank Hossein Moradi Rekabdarkolaee for helpful comments on the exposition of an earlier version of this work during the NISS Writing Workshop.
The authors used generative AI tools, including OpenAI Codex, to assist with language polishing and code development. All AI-assisted text and code were reviewed, edited, and validated by the authors, who take full responsibility for the manuscript and its results.
This work was supported in part by NIH grants R01MH143267, U01CA274576, and R01EB036016, NSF grant DMS-2310679, a Meta Faculty Research Award, and Wharton AI for Business. J. Blanchet gratefully acknowledges support from the Department of Defense through ONR award 1398311 and from the National Science Foundation through grants 2312204 and 2403007. The content is solely the responsibility of the authors and does not necessarily represent the official views of the NIH.

\bibliographystyle{plainnat}
\bibliography{bib/chatgpt,bib/privacy,bib/stat}

\clearpage
\input{7-appendix}

\end{document}

%% file: tex/math_commands.tex
\usepackage{amsmath,amsfonts,bm}

\def\1{\bm{1}}

\def\eps{{\varepsilon}}

\def\rd{{\textnormal{d}}}

\DeclareMathAlphabet{\mathsfit}{\encodingdefault}{\sfdefault}{m}{sl}
\SetMathAlphabet{\mathsfit}{bold}{\encodingdefault}{\sfdefault}{bx}{n}

\renewcommand{\xi}{\zeta}
\newcommand{\Decision}{{\mathcal D}_n}

\def\Bdelta{\boldsymbol{\delta}}

\def\0{{\bf 0}}
\def\1{{\bf 1}}

\def\AM{{\mathcal A}}

\def\FM{{\mathcal F}}

\def\SM{{\mathcal S}}

\def\WM{{\mathcal W}}

\def\RB{{\mathbb R}}

\DeclareMathOperator{\EB}{\mathbb{E}}

\def\PB{{\mathbb P}}

\def\rd{{\mathrm{d}}}

\def\Algo{{\texttt{SPOT}}}
\def\Wrate{\alpha}

\def\Key{{\mathtt{Key}}}

\def\token{{w}}

\def\Voca{{\WM}}
\newcommand\bP{\bm{P}}

\def\SMmax{\SM^{\mathrm{gum}}}

\def\Yars{Y}



%% file: 1-intro.tex
\section{Introduction}\label{sec:intro}

Large language models (LLMs) have recently become a powerful technology for generating human-like text and other media~\citep{touvron2023llama,singh2025openai,yang2025qwen3}. They are now widely used in writing, education, programming, scientific discovery, and many other aspects of daily life. At the same time, their ability to generate fluent text at scale raises serious concerns about misuse, including misinformation~\citep{weidinger2021ethical,starbird2019disinformation,pan2023risk}, academic integrity~\citep{stokel2022ai,milano2023large}, and data pollution for training future models~\citep{radford2023robust,shumailov2023curse}. These risks make it increasingly important to reliably authenticate the origin of text, especially when such information supports decisions about authorship attribution, academic assessment, and accountability~\citep{wu2025survey}.

To address this problem, watermarking has been proposed as a principled approach for verifying the origin of LLM-generated text~\citep{kirchenbauer2023watermark,scott2023watermarking,li2024statistical}. Its main idea is to embed a hidden statistical signal into the generation process through controllable and recoverable pseudorandomness. Since LLMs generate text sequentially through sampling each token from next-token prediction (NTP) distributions, a watermarking scheme can privately modify this sampling rule so that the generated token still marginally follows the same NTP distribution, while being coupled with a pseudorandom variable known to the verifier. This coupling remains largely invisible to human users but provides valid statistical evidence for verifying whether the text was produced by a watermarked model. Following this principle, many watermarking schemes have been proposed since 2023~\citep{zhao2025sok,ji2026overview}. These efforts have made watermarking one of the most promising approaches for providing provable evidence about the origin of LLM-generated text.

The practical relevance of text watermarking is no longer purely prospective. Google has deployed SynthID to watermark text generated by Gemini~\citep{dathathri2024scalable}, and Anthropic has recently begun deploying embedded text watermarks in newly launched Claude models~\citep{anthropic2026marks}. These industry deployments highlight the growing practical importance of understanding whether watermark evidence remains reliable and informative after generated text is edited.

In real-world applications, however, LLM-generated text is rarely used without modification. Users often paraphrase, rewrite, insert, or delete portions of generated outputs before using or submitting them. Such edits weaken watermark signals because modified tokens may no longer preserve the original dependence between the text and the pseudorandomness, while the verifier observes only the final text. This challenge has motivated work on robust watermark design against human edits~\citep{kuditipudi2023robust,yoo2023robust,zhu2024duwak,hou2024semstamp,ren2023robust,christ2024pseudorandom,golowich2024edit}. From a statistical perspective, \citet{li2025robust} model edited text as a mixture in which only a fraction of tokens still carry watermark signals and studies global detection of partially watermarked text. Under the same mixture model, \citet{li2025optimal} estimate how much watermark signal remains after common human editing. These studies address only aggregate questions: whether a mixed-source text still contains watermark evidence in aggregate, and how strong the remaining signal is.

A natural next question is more fine-grained: for a mixed-source text, can we identify which parts still preserve watermark evidence? We refer to this task as watermark localization, or watermark discovery, following the terminology of signal discovery in the multiple testing literature~\citep{CaiSun2017OSD}; see Figure~\ref{fig:localization-schematic} for an illustration. Unlike global detection, which only determines whether watermark evidence is present in the text as a whole, localization aims to provide token-level information about which positions still provide watermark evidence. This refinement is important when a document is only partially generated, collaboratively written, or substantially revised: a document-level conclusion may be too coarse to distinguish a short AI-generated passage from a largely AI-generated document, or to separate machine-generated content from substantial human revisions. Localization, therefore, provides a more informative basis for attribution, credit assignment, and targeted review.

\begin{figure}[t]
	\centering
	\resizebox{1\textwidth}{!}{%
		\begin{tikzpicture}[font=\small]
			
			\definecolor{wmblue}{HTML}{4C78A8}
			\definecolor{editorange}{HTML}{F58518}
			\definecolor{discgreen}{HTML}{54A24B}
			\definecolor{missred}{HTML}{E45756}
			\definecolor{lightgray}{HTML}{F2F2F2}
			
			\tikzset{
				token/.style={
					draw=black!25,
					fill=lightgray,
					rounded corners=2pt,
					minimum height=5.2mm,
					minimum width=1.05cm,
					inner xsep=2pt,
					font=\scriptsize
				},
				signal/.style={
					token,
					draw=wmblue!80,
					fill=wmblue!18
				},
				edited/.style={
					token,
					draw=editorange!90,
					fill=editorange!20
				},
				selected/.style={
					draw=discgreen!90!black,
					line width=0.85pt,
					rounded corners=2pt
				},
				missed/.style={
					draw=missred!90!black,
					line width=0.85pt,
					rounded corners=2pt,
					densely dashed
				},
				panelbox/.style={
					draw=black!18,
					rounded corners=5pt,
					inner sep=4pt
				},
				flowarrow/.style={
					-{Latex[length=2.0mm]},
					line width=0.65pt,
					draw=black!55
				},
				smallnote/.style={
					font=\scriptsize,
					text=black!70,
					align=center
				}
			}
			
			
			\node[signal] (orig1) at (0.00,0) {The};
			\node[signal] (orig2) at (1.25,0) {model};
			\node[signal]  (orig3) at (2.50,0) {writes};
			\node[signal] (orig4) at (3.75,0) {a};
			\node[signal] (orig5) at (5.00,0) {draft};
			\node[signal]  (orig6) at (6.25,0) {with};
			\node[signal] (orig7) at (7.50,0) {hidden};
			\node[signal]  (orig8) at (8.75,0) {subtle};
			\node[signal] (orig9) at (10.00,0) {evidence};
			
			\node[panelbox, fit=(orig1)(orig9)] (boxorig) {};
			\node[smallnote, anchor=west] at (11.25,0)
			{blue tokens are watermarked};
			
			
			\node[signal] (edit1) at (0.00,-1.60) {The};
			\node[edited] (edit2) at (1.25,-1.60) {author};
			\node[signal]  (edit3) at (2.50,-1.60) {writes};
			\node[signal] (edit4) at (3.75,-1.60) {a};
			\node[edited] (edit5) at (5.00,-1.60) {revised};
			\node[edited]  (edit6) at (6.25,-1.60) {text};
			\node[edited] (edit7) at (7.50,-1.60) {with};
			\node[signal] (edit8) at (8.75,-1.60) {subtle};
			\node[signal] (edit9) at (10.00,-1.60) {evidence};
			
			\node[panelbox, fit=(edit1)(edit9)] (boxedit) {};
			\node[smallnote, anchor=west] at (11.25,-1.60)
			{orange tokens are edited};
			
			
			\node[signal] (out1) at (0.00,-3.20) {The};
			\node[token] (out2) at (1.25,-3.20) {author};
			\node[token] (out3) at (2.50,-3.20) {writes};
			\node[signal] (out4) at (3.75,-3.20) {a};
			\node[token] (out5) at (5.00,-3.20) {revised};
			\node[token] (out6) at (6.25,-3.20) {text};
			\node[token] (out7) at (7.50,-3.20) {with};
			\node[token] (out8) at (8.75,-3.20) {subtle};
			\node[signal] (out9) at (10.00,-3.20) {evidence};
			
			\node[panelbox, fit=(out1)(out9)] (boxout) {};
			
			\node[selected, fit=(out1)] {};
			\node[selected, fit=(out4)] {};
			\node[missed, fit=(out7)] {};
			\node[selected, fit=(out9)] {};
			
			\newcommand{\legenditem}[2]{%
				\tikz[baseline=-0.5ex]{\node[#1, minimum width=4mm, minimum height=3mm] {};}\; #2%
			}
			
			\node[anchor=west, font=\scriptsize] at (11.25,-2.60) {\legenditem{signal}{surviving signal, $\theta_t=1$}};
			\node[anchor=west, font=\scriptsize] at (11.25,-2.96) {\legenditem{token}{null signal, $\theta_t=0$}};
			\node[anchor=west, font=\scriptsize] at (11.25,-3.30) {\legenditem{selected}{true discovery, $(\theta_t,\delta_t)=(1, 1)$}};
			\node[anchor=west, font=\scriptsize] at (11.25,-3.65) {\legenditem{missed}{false discovery, $(\theta_t,\delta_t)=(0, 1)$}};
			
			\draw[flowarrow]
			(boxorig.south) --
			node[right, font=\scriptsize, text=black!55] {human edit}
			(boxedit.north);
			
			\draw[flowarrow]
			(boxedit.south) --
			node[right, font=\scriptsize, text=black!55] {localization output: \{The, a, with, evidence\}}
			(boxout.north);
			
		\end{tikzpicture}%
	}
	\vspace{-20pt}
	\caption{
		Watermark localization under common human edits. 
		For each final-text position $t$, $\theta_t=1$ indicates a surviving watermark signal, and $\delta_t=1$ indicates selection by the localization method. 
		Selected positions with $\theta_t=1$ are true discoveries, while those with $\theta_t=0$ are false discoveries.
	}
	\label{fig:localization-schematic}
	\JRSSBfigureaftervspace{-10pt}
\end{figure}
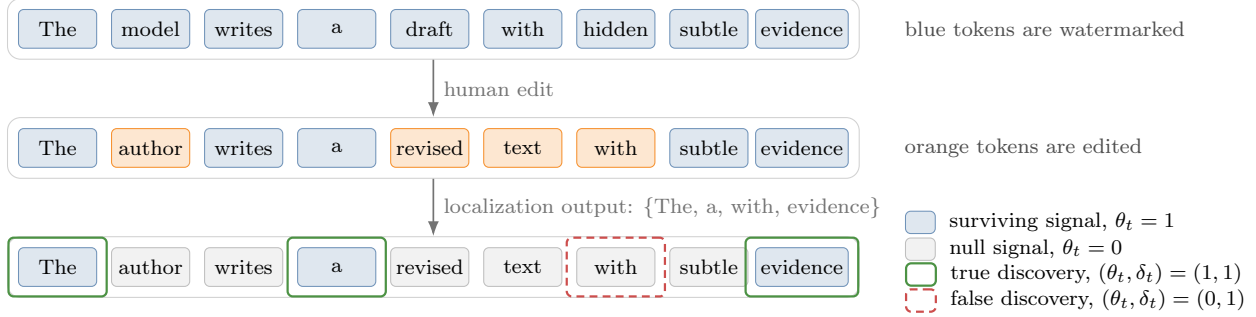

Despite its practical importance, watermark localization remains much less understood from a statistical viewpoint. Recent works have explored this fine-grained problem using change-point detection~\citep{li2024segmenting,bonnerjee2025fast} or online learning ideas~\citep{zhao2025efficiently}, showing that localization is achievable in practice. However, these approaches do not characterize the statistical limits of token-level localization for mixed-source data when surviving watermark evidence is sparse and fragmented. This gap is nontrivial: mixed-source text creates a heterogeneous sequence in which watermark-preserving locations may be sparse, non-contiguous, and distributionally heterogeneous, while the watermark signal itself varies across positions due to autoregressive LLM generation. Motivated by this gap, we ask the following questions: \textit{for mixed-source text, when is watermark localization statistically possible, and can it be achieved adaptively whenever localization is information-theoretically possible, without requiring prior knowledge of the source-mixing process?}

\subsection{Our Contributions}

\paragraph{A robust multiple-testing framework for watermark localization.}
We study these questions by developing a statistical theory and adaptive methodology for watermark localization in mixed-source LLM text.
Our first contribution is a statistical formulation of this problem.
For a text of length $n$, we use a scalar pivotal statistic $Y_t$ from~\citet{li2024statistical} to quantify the watermark signal at each token position $t$.
The key property is that, when no watermark signal survives at position $t$, $Y_t$ follows a known null distribution $\mu_0$, regardless of the marginal token distribution. This distribution-free null property has been central in prior statistical analyses of LLM watermarks and allows us to handle unknown and time-varying NTP distributions~\citep{li2024statistical,li2025robust,li2025optimal}.

To model the mixed-source text induced by human edits, we refine the mixture model of~\citet{li2025robust,li2025optimal} to the token level.
For each position $t$, we introduce a latent survival indicator $\theta_t\in\{0,1\}$: $\theta_t=1$ means that the watermark dependence at position $t$ survives editing, while $\theta_t=0$ means that this dependence is erased.
Conditional on $\theta_t$, the pivotal statistic follows either the null law $\mu_0$ or a watermark-induced alternative law $\mu_{1,\bP_t}$, determined by the local NTP distribution $\bP_t$ of token $w_t$.
Thus, given the observed sequence $Y_{1:n}:=(Y_1,\ldots,Y_n)$, watermark localization can be viewed as the task of inferring the latent survival indicators $\theta_{1:n}:=(\theta_1,\ldots,\theta_n)$.

Since exact recovery may be statistically impossible~\citep{CaiSun2017OSD}, we study this localization task through three inference goals of increasing strength.
\textit{Global detection} is the coarsest goal: it asks whether the verifier can reliably determine from the observed sequence whether at least one position satisfies $\theta_t=1$.
\textit{Discovery} asks whether the verifier can identify a non-trivial set of watermark-preserving positions while ensuring a vanishing false discovery rate.
Here, a false discovery is a selected position with no surviving watermark signal ($\theta_t=0$), while a missed discovery is an unselected position with surviving watermark signal ($\theta_t=1$).
\textit{Classification} is the strongest goal: it asks whether the verifier can asymptotically separate null positions from watermark-preserving positions across the entire text, with both false discoveries and missed discoveries vanishing.
This hierarchy turns localization into a sequence of increasingly demanding inference goals. 
We next characterize, for each goal, when it is information-theoretically achievable.


\JRSSBparagraphvspace{-1.1em}
\paragraph{Phase transitions for three inference goals.}
To characterize the fundamental limits of these three inference goals, we consider an asymptotic regime indexed by three exponents $(p,q,\alpha)$.
The surviving watermark fraction satisfies $\varepsilon_n\asymp n^{-p}$, so that the text contains approximately $n\varepsilon_n\asymp n^{1-p}$ watermark-preserving positions; thus, larger $p$ corresponds to sparser surviving signals.
Each NTP distribution $\bP_t := (P_{t,w})_w$ satisfies
$1-\max_{\token\in\Voca}P_{t,\token}\asymp n^{-q}$, so that larger $q$ corresponds to a more concentrated next-token distribution and hence weaker token-level watermark evidence.
Finally, the effective low-probability vocabulary tail grows at rate $n^\alpha$, with larger $\alpha$ providing more rare-token opportunities for distinctive local evidence.
These exponents characterize the statistical difficulty of the problem and are not required as inputs to our procedure.
Under this regime, we derive an explicit boundary for global detection and sharp phase transitions for discovery and classification within the class of coordinatewise pivot-based localization rules.

\begin{itemize}
	\vspace{-0.6em}
	\item \textbf{Detection boundary.}
	We first show that global detection is possible if and only if
	\begin{equation}
		\label{eq:new-detection-boundary}
		\max\{p+q,\;2p+q-\alpha\}<1,
	\end{equation}
	up to boundary cases.
	This generalizes the fixed-vocabulary detection boundary of~\citet{li2025robust} to the growing-vocabulary regime.
	We also show that the truncated goodness-of-fit detection method of~\citet{li2025robust} continues to attain this enlarged boundary adaptively.
	
	\vspace{-0.6em}
	\item \textbf{Discovery boundary.}
	For the main localization goal, we prove that discovery is possible if and only if
	\begin{equation}
		\label{eq:discovery-boundary}
		p<\alpha
		\quad\text{and}\quad
		p+q<1,
	\end{equation}
	up to boundary cases.
	This region \eqref{eq:discovery-boundary} is strictly smaller than the detectable region in~\eqref{eq:new-detection-boundary}, implying that discovery is fundamentally harder than global detection, because it requires locating a watermark-preserving position with vanishing false discovery rate.
	In particular, when $\alpha=0$, discovery is impossible though global detection is still possible in some regimes. 
	
	\vspace{-0.6em}
	\item \textbf{Impossible classification.}
	We further prove that consistent classification is impossible throughout the
	entire $(p,q,\alpha)$ regime within the class of coordinatewise localization rules.
	The reason is that the two requirements of classification are not compatible with each other.
	To make missed discoveries vanish, one must select almost all watermark-preserving positions, including those with weak evidence; but doing so inevitably selects too many null positions, so the false discovery rate cannot vanish.
	In this way, discovery can still be possible because it only needs a few strong watermark-preserving positions, whereas classification requires reliable recovery of all of them.
\end{itemize}

\JRSSBFigureTwo{%
\centering
\includegraphics[width=\textwidth]{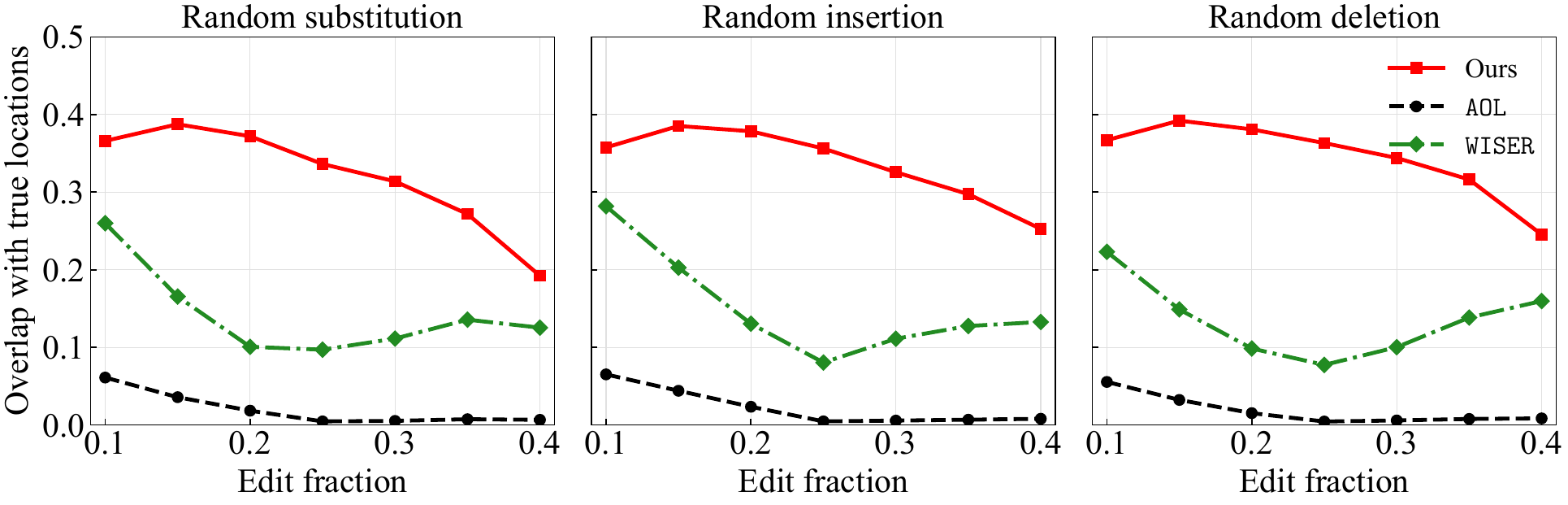}
\caption{LLM watermark localization under random edits. See Section~\ref{sec:LLM-experiments} for details. For substitution, insertion, and deletion edits, we report intersection-over-union scores for our method (Ours) and two baselines (\texttt{AOL} and \texttt{WISER}) at a target false positive rate of $0.05$. Higher values indicate better localization performance.}
\label{fig:real-llm-motivating-comparison}
}

\JRSSBparagraphvspace{-1.1em} 
\paragraph{Adaptive method and optimality.}
Our third contribution is an adaptive localization method, \Algo~(\textbf{S}canning \textbf{P}ivots \textbf{O}ver \textbf{T}hresholds), presented in
Algorithm~\ref{alg:adaptive-discovery-u}, for watermark discovery.
The method is designed to achieve discovery throughout the statistically
discoverable region without prior knowledge of the problem-dependent quantities
that determine the boundary, including $p$, $q$, $\Wrate$, and the time-varying
NTP distributions, although it uses a data-driven estimate of the overall
surviving watermark fraction \citep{li2025optimal}.

At a high level, \Algo\ searches for positions whose pivotal statistics fall in an extreme upper tail relative to the null law $\mu_0$. 
Thus, the main challenge is to choose a data-dependent threshold that is high enough to control false discoveries but not so high that useful watermark-preserving positions are lost. 
To do so, \Algo\ scans a grid of tail thresholds, estimates a surrogate false discovery level from empirical tail counts, and selects the largest discovery set whose estimated false discovery level is controlled. 
This tail-based calibration avoids estimating the heterogeneous alternative distributions, which is often impractical due to the inaccessibility of time-varying NTP distributions.

We establish two optimality guarantees.
First, \Algo\ is boundary-optimal: it achieves discovery throughout the
discoverable region characterized in \eqref{eq:discovery-boundary}.
Second, it has near-optimal discovery power, recovering, up to a constant factor,
as many watermark-preserving positions as the best rule in a natural class of
homogeneous pivot-based procedures under the same false-discovery constraint.
Figure~\ref{fig:real-llm-motivating-comparison} shows that our method achieves better watermark localization performance than the baseline methods \texttt{AOL}~\citep{zhao2025efficiently} and \texttt{WISER}~\citep{bonnerjee2025fast}, under common edits such as random substitution, insertion, and deletion.

\subsection{Related Work}

\paragraph{Watermark detection and localization.}
Our work is most closely related to statistical watermark detection and fine-grained localization.
\citet{li2025robust} study global detection under a mixture model for edited text, and \citet{li2025optimal} estimate the remaining watermark proportion in hybrid AI--human text.
These works provide statistical tools for aggregate questions, but they do not identify where the surviving watermark signals are located.
For localization in settings with contiguous watermarked segments, \citet{li2024segmenting} and \citet{bonnerjee2025fast} develop change-point detection approaches.
\citet{zhao2025efficiently} propose the adaptive online locator \texttt{AOL} for identifying watermarked segments in mixed-source text.
Our focus is complementary: we ask when sparse, potentially fragmented surviving watermark evidence can be reliably localized. We formulate this task as a token-level multiple testing problem, derive sharp limits for detection, discovery, and classification, and establish adaptive optimality within the pivot-based rule classes specified below.
More broadly, LLM watermarking has also been studied through
biased \citep{kirchenbauer2023watermark,tsur2025heavywater} or distribution-preserving watermark design~\citep{scott2023watermarking,zhao2024provable,xie2024debiasing,giboulot2024watermax,he2026improving,huang2026towards}, robustness-oriented design~\citep{kuditipudi2023robust,yoo2023robust,zhu2024duwak,hou2024semstamp,ren2023robust,christ2024pseudorandom,golowich2024edit}, and statistical detection~\citep{li2024statistical,li2025robust,he2026empirical}.
Rather than proposing a new watermarking scheme, we focus on the statistical limits of localizing surviving watermark evidence in mixed-source text.

\JRSSBparagraphvspace{-1.1em}
\paragraph{Sparse signal discovery and multiple testing.}
Our formulation is most directly connected to the signal discovery framework of~\citet{CaiSun2017OSD}, which studies how to identify sparse signals under false discovery control.
We follow this perspective by distinguishing detection, discovery, and classification, but our setting differs because the alternative laws are heterogeneous and depend on unknown, time-varying NTP distributions generated by an autoregressive LLM.
Our work is also related to oracle compound decision rules and local-FDR methods for large-scale multiple testing~\citep{sun2007oracle}, where local posterior information provides a natural principle for selecting non-null positions.
In our setting, however, direct local-FDR estimation is difficult because the alternative law at each position is unobserved and varies with the local NTP distribution.
We therefore use tail counts of pivotal statistics to control false discoveries without estimating the full heterogeneous alternative distribution.
More broadly, our phase-transition analysis is related to sparse mixture detection and higher criticism, where sharp detection boundaries are characterized for rare and weak signals~\citep{donoho2004higher,donoho2015higher}, and to goodness-of-fit tests based on divergence statistics~\citep{jager2007goodness}.

\subsection{Organization of the Paper}
The remainder of the paper is organized as follows. In Section~\ref{sec:pre}, we review the basics of LLM watermarking. In Section~\ref{sec:method}, we formulate the watermark localization problem and present our adaptive localization method. In Section~\ref{sec:theory}, we investigate the information-theoretic limits of the problem and establish the optimality of our method. In Section~\ref{sec:simulation}, we present simulation studies to validate our theoretical findings. In Section~\ref{sec:LLM-experiments}, we conduct experiments on LLM-generated text to evaluate the empirical performance of our method. We conclude in Section~\ref{sec:discuss} with a discussion of future research directions. All technical proofs and additional experimental details are provided in the Supplementary Material.

%% file: 2-prelim.tex
\section{Preliminaries}\label{sec:pre}

\paragraph{Watermarking protocol.} 
We study a standard watermarking protocol involving three parties: the model provider, the user, and the verifier~\citep{kuditipudi2023robust,xie2024debiasing,li2025robust}. As a motivating example, a student may use an LLM to assist with a homework assignment. The model provider embeds a statistical watermark into the generated text using a secret key $\Key$, which is available to the verifier, such as an instructor, but hidden from the user. Before submission, the user may revise or edit the generated output. The verifier observes only the final text, which may already have been modified, and aims to localize the watermark signal, namely, to identify which parts of the text still preserve watermark signals. The verifier does not observe the original prompt, the unedited model output, or the model's internal parameters. Thus, the protocol consists of three stages: watermark embedding by the model provider, possible human editing by the user, and watermark localization from the final edited text by the verifier.

\JRSSBparagraphvspace{-1.1em} 
\paragraph{Watermark embedding.}  
We next describe how a watermark is embedded during text generation. LLMs generate text autoregressively: at position $t$, given previous tokens $\token_{1:(t-1)}:=\token_1\cdots \token_{t-1}$, the model computes a next-token prediction (NTP) distribution $\bP_t := (P_{t, w})_{w\in\Voca}$ over the vocabulary $\Voca$ and samples the next token accordingly. A watermark modifies this sampling step in a secret-key-dependent way. Specifically, the model computes a pseudorandom variable $\xi_t=\AM(\token_{(t-m):(t-1)},\Key)$, where $\AM$ is a cryptographic hash function, $\Key$ is the watermarking key, and $m$ is the context window size, and then selects the next token by a deterministic decoder $\token_t=\SM(\bP_t,\xi_t)$. This creates a dependence between the generated token $\token_t$ and the pseudorandom variable $\xi_t$.
The variable $\xi_t$ is called pseudorandom because it behaves statistically like a random draw, while being exactly reproducible from $\AM$, $\Key$, and the local text $\token_{(t-m):(t-1)}$. For theoretical analysis, we model $\xi_t$ as i.i.d. samples from a reference distribution $\pi$, reflecting standard cryptographic idealizations~\citep{barak2021book,schneier1996applied,wu2023dipmark,zhao2024permute}. A decoder $\SM$ is unbiased if $\PB_{\xi\sim\pi}(\SM(\bP,\xi)=\token)=P_{\token}$ for every NTP distribution $\bP = (P_{\token})_{\token \in \Voca}$ and candidate token $\token$. Thus, unbiased watermarking preserves the marginal distribution of generated tokens while creating a hidden dependence on $\xi_t$ that can later be used for verification.

\JRSSBparagraphvspace{-1.1em} 
\paragraph{Pivotal statistics for watermark evidence.}
Since our primary interest is watermark localization, we adopt the verifier's perspective. Given the final text $\token_{1:n}$, the verifier reconstructs the corresponding pseudorandom sequence $\xi_{1:n}$ using the hash function $\AM$ and the shared key $\Key$. Hence, the data for statistical analysis are the paired observations $\{(\token_t,\xi_t)\}_{t=1}^n$. Watermark evidence is encoded in the dependence between a token and its pseudorandom variable. If a token is human-written, or if its watermark dependence has been destroyed by editing, then $\token_t$ is independent of $\xi_t$; if the watermark signal is preserved, then $\token_t$ remains coupled with $\xi_t$ through the decoder. To measure this dependence, we use a pivotal statistic $\Yars_t=Y(\token_t,\xi_t)$~\citep{li2024statistical}, which is designed to follow a known null distribution $\mu_0$ whenever $\token_t$ is independent of $\xi_t$, regardless of the marginal distribution of $\token_t$. When the watermark dependence is preserved, $\Yars_t$ instead follows an alternative distribution $\mu_{1,\bP_t}$ determined by the local NTP distribution $\bP_t$. Thus, pivotal statistics reduce watermark localization to distinguishing null-like positions from watermark-preserving positions. We formalize the resulting multiple testing problem in Section~\ref{sec:setting}.

\JRSSBparagraphvspace{-1.1em} 
\paragraph{The Gumbel-max watermark.}
For concreteness, our theoretical analysis focuses on the Gumbel-max watermark~\citep{scott2023watermarking}, one of the most influential unbiased watermarking schemes. We use it as our main example because both its decoder and pivotal statistic admit explicit distributional forms, making it a clean setting for deriving sharp localization limits.\footnote{These ideas are not specific to Gumbel-max and may extend to other schemes with suitable pivotal statistics.}
The Gumbel-max watermark is based on the Gumbel-max trick for sampling from a multinomial distribution~\citep{gumbel1948statistical,papandreou2011perturb,maddison2014sampling,jang2016categorical}. 
Let $\xi=(U_w)_{w\in\Voca}$ consist of i.i.d. $U(0,1)$ random variables. The trick states that $\arg\max_{w\in\Voca}\log U_w/P_w$ follows the categorical distribution $\bP=(P_w)_{w\in\Voca}$ exactly. This motivates the unbiased decoder
\begin{equation}\label{eq:it}
\SMmax(\bP,\xi):=\arg\max_{w\in\Voca}\frac{\log U_w}{P_w}.
\end{equation}
For this watermark, the pivotal statistic is $\Yars_t=Y(\token_t,\xi_t)=U_{t,\token_t}$, where $\xi_t=(U_{t,w})_{w\in\Voca}$ collects the uniform pseudorandom at position $t$. If the observed token $\token_t$ is independent of $\xi_t$, then $\Yars_t\sim U(0,1)$. If the token is generated via the Gumbel-max decoder $\SMmax$, larger pseudorandom values are more likely to be selected, and the pivot becomes stochastically larger. More precisely, for a given NTP distribution $\bP$, its distribution is $\mu_{1,\bP}(Y\le r)=\sum_{w\in\Voca}P_wr^{1/P_w}$ for $r\in[0,1]$~\citep{li2024statistical}. This explicit null-versus-alternative structure is the basis for our localization analysis.

%% file: 3-results.tex
\section{Method}
\label{sec:method}

\subsection{Problem Formulation}
\label{sec:setting}

We first describe the statistical data structure behind watermark localization. From Section~\ref{sec:pre}, the verifier observes the final text $\token_{1:n}$, reconstructs the pseudorandom variables $\xi_{1:n}$, and computes the pivotal statistic $\Yars_t := Y(\token_t,\xi_t)$ at each position $t$. This scalar statistic measures the watermark evidence in the pair $(\token_t,\xi_t)$. Its key property is that, if the token is human-written or if human editing has broken its dependence on the pseudorandom variable, then $\token_t$ is independent of $\xi_t$, and $\Yars_t$ follows a known null distribution $\mu_0$ regardless of the marginal distribution of $\token_t$. In contrast, if the watermark signal survives, then $\token_t$ remains coupled with $\xi_t$ through the decoder, and $\Yars_t$ follows an alternative distribution $\mu_{1,\bP_t}$ determined by the local NTP distribution $\bP_t$.

This null-versus-signal structure motivates a latent label $\theta_t\in\{0,1\}$, where $\theta_t=1$ indicates that the watermark dependence at position $t$ survives editing, and $\theta_t=0$ indicates that it is erased. At the level of pivotal statistics, this structure is summarized as
\begin{equation}
\label{eq:mixture-Y}
\Yars_t\mid(\bP_t,\theta_t)\sim (1-\theta_t)\mu_0+\theta_t\mu_{1,\bP_t}
~~\text{for}~~t=1,\ldots, n.
\end{equation}
The localization target is to identify the signal set $I:=\{t:\theta_t=1\}$, consisting of positions where watermark evidence remains. We also denote the non-watermarking set by $N:=\{t:\theta_t=0\}$. Therefore, watermark localization is a token-level multiple decision problem: at each position, the verifier decides whether the local state is null-like or watermark-preserving.

A \textit{localization rule} is a binary sequence $\Bdelta=(\delta_1,\ldots,\delta_n)$, where $\delta_t=1$ means that position $t$ is declared to preserve watermark evidence. For a given rule $\Bdelta$, we define $S_{\Bdelta}:=\{t:\delta_t=1\}$ as the associated \textit{discovery set}, namely the set of positions selected by the rule $\Bdelta$ as watermark-preserving. In this paper, we focus on \textit{local rules} that are separate on pivotal statistics at each position, denoted by $\Decision:=\{\Bdelta:\delta_t=\phi_{t}(\Yars_t),\ t=1,\ldots,n\}$, where each $\phi_t:[0,1]\to\{0,1\}$ is a measurable function that may depend on the position $t$ and the text length $n$. Unless otherwise stated, all decision rules belong to $\Decision$. 

\begin{rem}[Connection to prior formulation]
Our formulation in \eqref{eq:mixture-Y} is motivated by the mixture model of \citet{li2025robust}, but differs in that we explicitly introduce the latent binary sequence $\{\theta_t\}_{t=1}^n$ to encode which token positions still preserve watermark dependence after human edits. This additional structure turns robust watermark detection into a localization problem.
\end{rem}



\subsection{Our Method: \Algo~}
\label{sed:algo}

We now present our method \Algo~in Algorithm~\ref{alg:adaptive-discovery-u}. At a high level, the method searches for positions whose pivotal statistics look unusually abnormal under the null law $\mu_0$, and declares such positions as discoveries. Under the null case, $\theta_t = 0$, so that the $p$-value $p_t:=1-F_0(\Yars_t)$ is i.i.d. $U(0,1)$, where $F_0(y) := \mu_0(Y \le y)$ is the CDF of the null distribution $\mu_0$. Hence, very small $p$-values, or equivalently very large values of $F_0(\Yars_t)$, provide evidence that the corresponding positions may preserve watermark signal. Therefore, the problem reduces to choosing a threshold: positions with $F_0(\Yars_t)$ above this threshold are declared as discoveries.

\begin{algorithm}[t]
\caption{\Algo: \textbf{S}canning \textbf{P}ivots \textbf{O}ver \textbf{T}hresholds}
\label{alg:adaptive-discovery-u}
\begin{algorithmic}[1]
\Statex \textbf{Input:} Given text $\token_{1:n}$, hash function $\AM$, secret key $\Key$, pivot function $Y$, grid endpoints $0<u_{\min}<u_{\max}<1$, grid size $M_n$, target level $\lambda_n\in(0,1)$, and slack $\eta_n\in(0,\lambda_n)$.

\State \textbf{Compute pseudorandomness.} For each $t=1,\ldots,n$, reconstruct $\xi_t=\AM(\token_{(t-m):(t-1)},\Key)$.

\State \textbf{Compute pivotal statistics.} For each $t=1,\ldots,n$, compute $\Yars_t=Y(\token_t,\xi_t)$.

\State \textbf{Construct tail grid.} Set $\mathcal U_n:=\{u_j=u_{\min}+(j-1)\frac{u_{\max}-u_{\min}}{M_n-1}:j=1,\ldots,M_n\}$

\State \textbf{Estimate tail mass.} For each $u\in\mathcal U_n$, compute
\begin{equation}
\label{eq:Sn-hat-u}
\widehat S_n(u):=\frac1n\sum_{t=1}^n\mathbf 1\{F_0(Y_t)>\tau(u)\}.
\end{equation}

\State \textbf{Estimate surviving fraction.} Obtain an estimate $\widehat\varepsilon_n$ of the surviving watermark fraction.

\State \textbf{Estimate null proportion in the tail.} For each $u\in\mathcal U_n$, set
\begin{equation}
\label{eq:Ttail-hat-u}
\widehat T_n(u):=\frac{(1-\widehat\varepsilon_n)S_0(\tau(u))}{\widehat S_n(u)\vee n^{-1}}
=\frac{(1-\widehat\varepsilon_n)n^{-u}}{\widehat S_n(u)\vee n^{-1}}.
\end{equation}

\State \textbf{Select adaptive threshold.} Let
\begin{equation}
\label{eq:uhat-def}
\widehat u_n
:=
\inf\Bigl\{u\in\mathcal U_n:\ \widehat T_n(u)\le \lambda_n-\eta_n\Bigr\},
\qquad
\widehat\tau_n:=\tau(\widehat u_n)=1-n^{-\widehat u_n}.
\end{equation}

\State \textbf{Finalize decision.} Let $\Bdelta=(\delta_1,\ldots,\delta_n)$ by setting $\delta_t=1$ if $\Yars_t>\widehat\tau_n$, and $\delta_t=0$ otherwise.

\Statex \textbf{Output discoveries.} Return the estimated discovery set $S_{\Bdelta}=\{t\in\{1,\ldots,n\}:\delta_t=1\}$.
\end{algorithmic}
\end{algorithm}

The key difficulty is that the optimal threshold is hard to find. We want to output a discovery set $S_{\Bdelta}$ while controlling false discoveries. Here, a false discovery is a selected position whose watermark dependence has been erased, that is, a position with $\theta_t=0$ but $\delta_t=1$. A threshold that is too low may include too many false discoveries, while a threshold that is too high may remove many true discoveries. We introduce $\lambda_n$ as the target level for controlling the fraction of false discoveries among the selected positions.

Ideally, the threshold should depend on problem-specific quantities, such as the surviving watermark fraction $\eps_n$ and the NTP distribution $\bP_t$, which are unfortunately unavailable in practice. To address this issue, \Algo~scans thresholds of the form $\tau(u)=1-n^{-u}$. For each candidate $u$, it computes the empirical tail mass $\widehat S_n(u)$, which is the fraction of positions with $F_0(\Yars_t)>\tau(u)$. Under the null case where $\theta_t=0$, this tail probability is $S_0(\tau(u))=1-\tau(u)=n^{-u}$. If the surviving watermark fraction is $\varepsilon_n$, then roughly a fraction $1-\varepsilon_n$ of positions are null-like, so the expected null contribution to this tail is about $(1-\varepsilon_n)S_0(\tau(u))$. Replacing $\varepsilon_n$ by an estimate $\widehat\varepsilon_n$ and comparing this estimated null contribution with the observed tail mass give a quantity $\widehat T_n(u):=(1-\widehat\varepsilon_n)S_0(\tau(u))/(\widehat S_n(u)\vee n^{-1})$.\footnote{Here $a\vee b:=\max\{a,b\}$; the term $\vee n^{-1}$ in the denominator is used for numerical stability.}
We interpret $\widehat T_n(u)$ as an empirical proxy for the mFDR at the candidate threshold $\tau(u)$.

To control false discoveries, the algorithm selects the smallest grid point $u$ such that $\widehat T_n(u)\le \lambda_n-\eta_n$, where $\eta_n$ is a slack term used to guard against random fluctuations. The statistic $\widehat T_n(u)$ estimates the false-discovery level of the tail selected by the threshold $\tau(u)=1-n^{-u}$. Thus, a large value of $\widehat T_n(u)$ indicates that the selected tail may still be largely explained by null positions and is therefore not reliable enough. Once $\widehat T_n(u)$ falls below $\lambda_n-\eta_n$, the tail is estimated to contain sufficiently few false discoveries, with the slack accounting for data variability. Since $\tau(u)$ increases with $u$, choosing the smallest admissible $u$ yields the largest discovery set among the grid thresholds that pass the false-discovery check. This allows the method to retain as many candidate watermark-preserving positions as possible while controlling false discoveries.



\begin{rem}[Choice of fraction estimator]
Our theory is not tied to a specific estimator of the surviving watermark fraction. It only requires the accuracy condition in~\eqref{eq:estimator-accruacy-condition}. In the real-LLM experiments, the plug-in implementation uses the optimal estimator of~\citet{li2025optimal}.
\end{rem}

\begin{rem}[Connection to prior methods]
Our method \Algo~is related to the signal discovery framework of \citet{CaiSun2017OSD}, which studies Gaussian mixture models and uses local false discovery rate ideas through density estimation. The main challenge in our setting is that the data distributions in \eqref{eq:mixture-Y} are highly heterogeneous, since they depend on unknown and time-varying NTP distributions. Direct density estimation is thus difficult. Instead, \Algo~uses tail probabilities $\widehat S_n(u)$ to estimate the false discovery level and selects the tail threshold adaptively for token-level localization.
\end{rem}

%% file: 4-theory.tex
\section{Theoretical Guarantees}
\label{sec:theory}

In this section, we establish the theoretical properties of our method. To facilitate the analysis, we first introduce the general assumptions in Section~\ref{sec:assumption} and the performance measures in Section~\ref{sec:metric}. We then characterize the statistical limits and phase transitions of the resulting inference tasks in Section~\ref{sec:screenability-and-discoverability}. Finally, in Section~\ref{sec:optimality}, we show that our method is adaptively optimal and is near-optimal in the number of discoveries under a fixed marginal false-discovery control.

\subsection{General Assumptions}
\label{sec:assumption}

We first impose a probabilistic assumption on the watermark-generation and human-editing process. Each observed token $w_t$ is modeled through a latent source mixture: after editing, each position either preserves the watermark dependence and is generated by the watermarked decoder, or the dependence is erased and the token behaves as an ordinary draw from the same NTP distribution. The latent variable $\theta_t$ records this post-edit survival state. This assumption provides the token-level basis for the null-versus-signal structure of the pivotal statistic $Y_t$. 

\begin{asmp}[Watermark generation and editing mechanism]
\label{asmp:main}
Let $\bP_t$ denote the NTP distribution at position $t$. Define the generation filtration $\FM_{t-1}:=\sigma(\{\token_j,\xi_j,\bP_{j+1}\}_{j=1}^{t-1})$, which contains the history used to form $\bP_t$. Let $\mathcal H_t:=\sigma(\theta_1,\ldots,\theta_t)$ be the editing filtration and $\mathcal G_t:=\FM_{t-1}\vee\mathcal H_t$ be the joint generation-editing filtration. We assume the following.

\begin{enumerate}
\item[\rm{(a)}] \textbf{Perfect pseudorandomness.} $\xi_1,\ldots,\xi_n$ are i.i.d., and $\xi_t$ is independent of $\mathcal G_t$ for every $t$.

\item[\rm{(b)}] \textbf{Latent source mixture.} Conditional on $\mathcal G_t$, the observed token is generated by
\[
\token_t=
\begin{cases}
\SM(\bP_t,\xi_t), & \theta_t=1,\\
\text{an independent draw from }\bP_t, & \theta_t=0.
\end{cases}
\]
In the second case, $\token_t$ is conditionally independent of $\xi_t$.

\item[\rm{(c)}] \textbf{Surviving watermark fraction.} There exist constants $0<c\le C<\infty$, independent of $t$ and $n$, such that
\[
c \cdot \varepsilon_n\le \mathbb P(\theta_t=1\mid \FM_{t-1})\le C \cdot \varepsilon_n
\quad\text{a.s. for all } \quad t\ge1.
\]
\end{enumerate}
\end{asmp}

Assumption~\ref{asmp:main} separates the generation process from the editing process. 
Condition (a) is the standard idealization that the cryptographic pseudorandomness behaves as fresh randomness ~\citep{kirchenbauer2023watermark,li2024statistical} and is independent of the past and of the editing decision \citep{li2025robust,li2025optimal}. 
Condition (b) encodes the local null-versus-signal structure with $\theta_t$ as a latent survival indicator~\citep{li2025robust}: when $\theta_t=1$, the token $\token_t$ is produced by the watermarked decoder $\mathcal{S}$ and remains coupled with $\xi_t$; when $\theta_t=0$, the token $\token_t$ has the same marginal NTP distribution but is independent of $\xi_t$, since human writing does not have access to the pseudorandom variable.
The third row of Figure~\ref{fig:localization-schematic} illustrates this structure, where the ground-truth vector is $(\theta_1, \ldots, \theta_9)=(1,0,0,1,0,0,0,0,1)$. 
Condition (c) controls the overall amount of surviving watermark evidence: $\varepsilon_n$ is the surviving watermark fraction, up to constant factors, while the editing decision is allowed to depend on the previously generated text. 
Altogether, Conditions (b) and (c) characterize the mixture relation between the observed token $\token_t$ and the reconstructed pseudorandom variable $\xi_t$. Since the verifier does not know which positions still preserve the watermark dependence, the latent indicators $\theta_t$ provide a principled way to encode this uncertainty at the token level. Thus, the assumption does not model human editing behavior in detail, but preserves the essential statistical structure needed for analyzing watermark localization.


\begin{asmp}[Asymptotic regime $(p,q,\Wrate)$]
\label{asmp:HCL+}
For a text of length $n$, the surviving watermark fraction satisfies $\varepsilon_n\asymp n^{-p}$ for some $p\in[0,1]$.\footnote{For two positive sequences $a_n$ and $b_n$, we write $a_n\asymp b_n$ if there exist universal constants $0<c_1\le c_2<\infty$, independent of $n$, such that $c_1b_n\le a_n\le c_2b_n$ for all sufficiently large $n$.}  When $p=0$, we further assume that $\varepsilon_n$ is bounded away from one: there exists $\pi_\star\in(0,1)$ such that $\varepsilon_n\le \pi_\star$ for all sufficiently large $n$.
For each position $t$, the vocabulary $\Voca_n$ admits the decomposition
\[
\Voca_n=\{w_{t,n}^\star\}\cup L_{t,n}\cup C_{t,n},
\qquad
P_{t,w_{t,n}^\star}=1-\Delta_n,\quad \Delta_n\asymp n^{-q}.
\]
Here $w_{t,n}^\star$ is the dominant token under $\bP_t$ and may vary with $t$. Thus, the dominant token can change across positions, but its probability is assumed to have the same asymptotic form. The residual mass is decomposed as
\[
\Delta_n=\Delta_n^{\mathrm{light}}+\Delta_n^{\mathrm{core}},
\qquad 
\Delta_n^{\mathrm{light}}\asymp n^{-q},
\qquad 
\Delta_n^{\mathrm{core}}\le c\,n^{-q}.
\]
The light set satisfies $|L_{t,n}|\asymp n^{\Wrate}$ and $P_{t,w}\asymp n^{-(\Wrate+q)}$ for all $w\in L_{t,n}$. The core set satisfies $|C_{t,n}|\asymp n^{r}$ for some $r<\Wrate$, and there exist constants $0<c_C\le C_C<\infty$ such that
\[
  c_C\,n^{-(\Wrate+q)} \le P_{t,w} \le C_C\,n^{-(r+q)}
  \qquad \text{for all } w\in C_{t,n},\ t=1,\dots,n.
\]
Consequently, the whole vocabulary satisfies $|\Voca_n|\asymp n^\Wrate$. In the special case $q=0$, we additionally assume that there exists $\eta_\star\in(0,1)$ such that $P_{t,w_{t,n}^\star}\ge \eta_\star$ for all $t$ and all sufficiently large $n$.
\end{asmp}

To study the statistical limits of watermark localization, we consider the asymptotic regime in Assumption~\ref{asmp:HCL+}, in the spirit of~\citet{li2024statistical,li2025robust}. The regime is parameterized by $(p,q,\Wrate)$. The parameter $p$ describes how quickly the surviving watermark fraction $\varepsilon_n\asymp n^{-p}$ decays after human edits. The parameter $q$ measures how concentrated each NTP distribution is around its dominant token: larger $q$ means that the dominant token has probability closer to one, and hence the watermark signal becomes weaker. The parameter $\Wrate$ describes the growth of the low-probability vocabulary tail.

\begin{figure}[t!]
\centering
\includegraphics[width=0.6\textwidth]{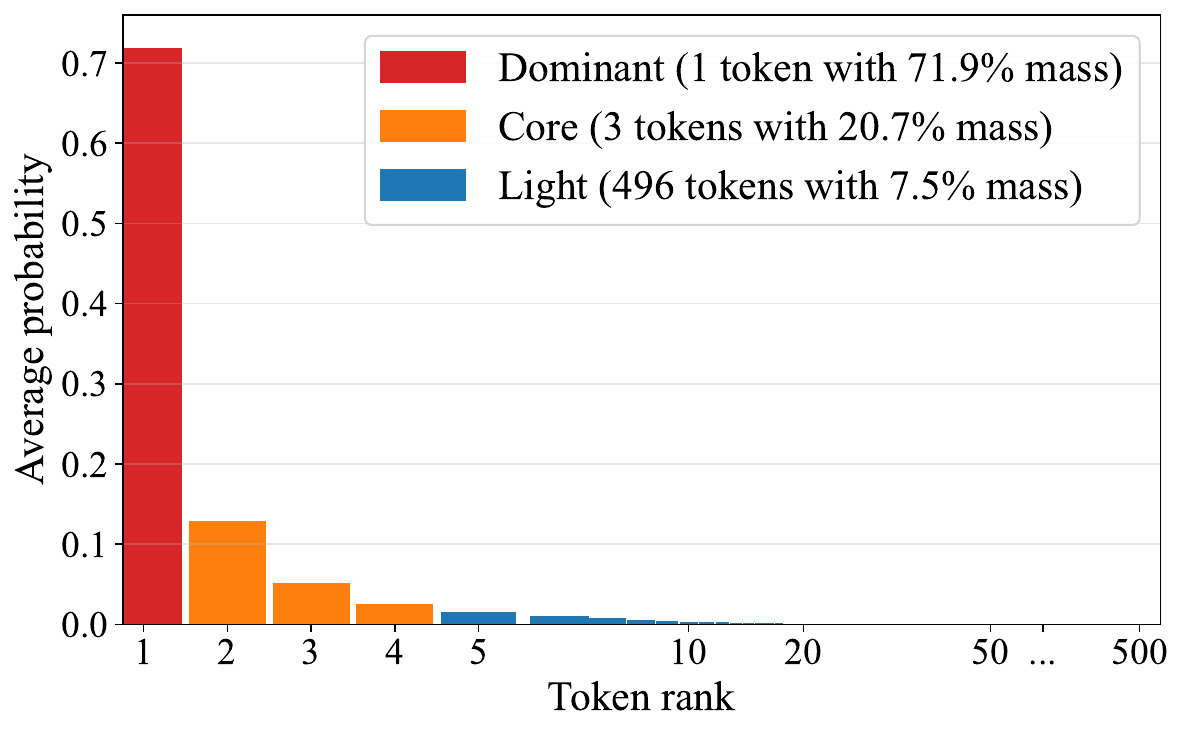}
\vspace{-10pt}
\caption{
Empirical illustration of the dominant--core--light structure in Assumption~\ref{asmp:HCL+}.
We sample prompts from C4 datasets \citep{raffel2020exploring} and generate continuations using OPT-1.3B \citep{zhang2022opt}, then compute the average ranked next-token probabilities using temperature $0.5$.
}
\label{fig:hcl-ntp-illustration}
\vspace{-10pt}
\end{figure}

The dominant--core--light decomposition in Assumption~\ref{asmp:HCL+} aims to capture a common shape of LLM next-token distributions. At each position $t$, the NTP distribution $\bP_t$ may have a dominant prediction $w_{t,n}^\star$, and this dominant token is allowed to vary with the context. We only require its probability to follow the common scale $1-\Delta_n$, with $\Delta_n\asymp n^{-q}$. The remaining probability mass $\Delta_n$ is split into a core set $C_{t,n}$ and a light set $L_{t,n}$. The light set $L_{t,n}$ contains most of the vocabulary, with $|L_{t,n}|\asymp n^\Wrate$, the same order as $|\Voca_n|$, and each light token has probability of order $n^{-(\Wrate+q)}$. In contrast, the core set $C_{t,n}$ is much smaller, with $|C_{t,n}|\asymp n^{r}$ for some $r<\Wrate$, but contains relatively larger-probability alternatives, with probabilities ranging from order $n^{-(\Wrate+q)}$ to $n^{-(r+q)}$. Thus, the light set represents the large low-probability tail that drives vocabulary growth, while the core set represents a smaller group of more likely alternatives. This distinction explains why $\Wrate$ controls the amount of tail information available for localization. The decomposition is also consistent with the empirical heavy-tailed behavior of language distributions~\citep{zipf2016human,moreno2016large}, where a few tokens receive most of the probability mass while many weak alternatives remain available. Figure~\ref{fig:hcl-ntp-illustration} provides an empirical illustration of this pattern using next-token distributions from OPT-1.3B on C4 prompts.
\begin{rem}[Interpretation of $\Wrate$]
The vocabulary of any deployed model is fixed. The growing-tail asymptotic should instead be viewed as a triangular-array approximation to successive model generations, which often adopt larger token vocabularies; for example, the Llama vocabulary increased fourfold, from $32{,}000$ tokens in Llama~2 to $128{,}000$ in Llama~3~\citep{touvron2023llama2,dubey2024llama}. Thus, $\Wrate$ represents growth in the effective rare-token support rather than literal vocabulary growth within a single text.
\end{rem}

The parameters $(p,q,\Wrate)$ are used only to characterize theoretical difficulty; our method in Section~\ref{sed:algo} does not require knowing them. Instead, they allow us to state sharp phase-transition results: $p$ controls the sparsity of surviving watermark signals, $q$ controls token-level signal strength through NTP concentration, and $\Wrate$ controls how much useful tail information is available for localization. The additional condition for $q=0$ only prevents the dominant-token probability from vanishing; when $q>0$, this is automatic for sufficiently large $n$.



\subsection{Performance Measures}
\label{sec:metric}

We now introduce the performance measures used in our theoretical analysis. For completeness, we first formally define global detection. We then focus on token-level localization: following the signal discovery framework of~\citet{CaiSun2017OSD}, we define false positive and missed discovery rates, and use them to formalize discovery and classification.

\begin{defn}[Global detection]
\label{def:detection}
A sequence of tests $\psi_n=\psi_n(\Yars_{1:n})\in\{0,1\}$ achieves \emph{detection} if the sum of Type I and Type II errors vanishes for the global testing problem
\begin{equation}
\label{eq:previous-mixture-hypothesis}
H_0:\ \Yars_t\sim\mu_0\ \text{for all }t
\qquad\text{versus}\qquad
H_1:\ \Yars_t\mid \bP_t \sim (1-\varepsilon_n)\mu_0+\varepsilon_n\mu_{1,\bP_t}\ \text{for all }t,
\end{equation}
that is,
\[
\mathbb P_{H_0}(\psi_n=1)+\mathbb P_{H_1}(\psi_n=0)\to 0.
\]
Here $\psi_n=1$ means that the text is declared to contain surviving watermark signal. We say the testing problem in~\eqref{eq:previous-mixture-hypothesis} is \emph{detectable} if there exists a sequence of tests that achieves global detection.
\end{defn}

Throughout the localization definitions below, decision rules are understood to belong to the local class $\Decision$ introduced in Section~\ref{sed:algo}, unless explicitly stated otherwise. 
Global detection is different: it is a document-level testing problem, and the test $\psi_n$ may use the full pivot sequence $\Yars_{1:n}$.

\begin{defn}[Counts of discoveries and missed signals]
\label{def:EFP-ETP}
Given a   rule $\Bdelta=(\delta_1,\ldots,\delta_n)$ and ground-truth labels $\theta_t\in\{0,1\}$, define
\[
\mathrm{FP}_{\Bdelta}=\sum_{t:\theta_t=0}\delta_t,\qquad
\mathrm{TP}_{\Bdelta}=\sum_{t:\theta_t=1}\delta_t,\qquad
\mathrm{FN}_{\Bdelta}=\sum_{t:\theta_t=1}(1-\delta_t).
\]
Here $\mathrm{FP}_{\Bdelta}$, $\mathrm{TP}_{\Bdelta}$, and $\mathrm{FN}_{\Bdelta}$ count false discoveries, true discoveries, and missed signals, respectively. Since these quantities are random, we define their expectations as
\[
\mathrm{EFP}_{\Bdelta}=\mathbb{E}[\mathrm{FP}_{\Bdelta}],\qquad
\mathrm{ETP}_{\Bdelta}=\mathbb{E}[\mathrm{TP}_{\Bdelta}],\qquad
\mathrm{EFN}_{\Bdelta}=\mathbb{E}[\mathrm{FN}_{\Bdelta}].
\]
\end{defn}

\begin{defn}[Marginal false discovery and missed discovery rates]\label{def:FPR-MDR}
With the above notation, define
\[
\mathrm{mFDR}_{\Bdelta}
=\frac{\mathrm{EFP}_{\Bdelta}}{\mathrm{EFP}_{\Bdelta}+\mathrm{ETP}_{\Bdelta}},
\qquad
\mathrm{MDR}_{\Bdelta}
=\frac{\mathrm{EFN}_{\Bdelta}}{\mathrm{EFN}_{\Bdelta}+\mathrm{ETP}_{\Bdelta}}.
\]
Here, $\operatorname{mFDR}_{\delta}$ is the ratio of the expected number of
false discoveries to the expected total number of discoveries, while
$\operatorname{MDR}_{\delta}$ is the ratio of the expected number of
missed signals to the expected total number of watermark-preserving
positions.
\end{defn}



\begin{defn}[Discovery]
\label{def:discovery}
A sequence of decision rules $\Bdelta$ achieves \emph{discovery} if
\[
\mathrm{mFDR}_{\Bdelta}\ \to\ 0
\qquad\text{and}\qquad
\mathbb{P}\left(\left|S_{\Bdelta}\right| \geq 1\right) \to 1.
\]
Discovery requires at least one selected position while the marginal
false discovery rate vanishes. The localization problem in~\eqref{eq:mixture-Y} is said to be \emph{discoverable} if there exists a sequence of local decision rules achieving discovery.
\end{defn}

\begin{rem}[Multiple discoveries]
The discovery condition can be strengthened to $\mathbb{P}(|S_{\Bdelta}|\ge K)\to 1$ for any fixed integer $K\ge 1$. This modification does not affect the discovery boundary; see the proof of Theorem~\ref{thm:discovery-boundary}.
\end{rem}

\begin{defn}[Classification]
\label{def:classification}
A sequence of decision rules $\Bdelta$ achieves \emph{classification} if
\[
\mathrm{mFDR}_{\Bdelta} \to 0 
\qquad \text{and} \qquad
\mathrm{MDR}_{\Bdelta} \to 0.
\]
Classification requires both false discoveries and missed watermark-preserving positions to vanish asymptotically. The localization problem in~\eqref{eq:mixture-Y} is said to be \emph{classifiable} if there exists a sequence of local decision rules achieving classification.
\end{defn}

The above definitions form a hierarchy of inference goals. Global detection in Definition~\ref{def:detection} is the coarsest task: it only asks whether surviving watermark signal exists somewhere in the text, without identifying any location. For localization, the marginal false discovery rate (mFDR) in Definition~\ref{def:FPR-MDR} measures how many selected positions are actually null, while the missed discovery rate (MDR) measures how many watermark-preserving positions are not selected. Based on these aspects, discovery in Definition~\ref{def:discovery} is the weakest successful localization goal: it requires finding at least one surviving watermark signal while keeping the mFDR vanishing. Classification in Definition~\ref{def:classification} is stronger, requiring asymptotically correct separation of null and watermark-preserving positions. We use these criteria above to show that global detection can be possible in regimes where localization is not, and to characterize which levels of localization are statistically achievable.


\subsection{Statistical Limits and Phase Transitions}
\label{sec:screenability-and-discoverability}

Throughout this subsection, we work under Assumptions~\ref{asmp:main} and~\ref{asmp:HCL+}, so the difficulty of the problem is indexed by $(p,q,\Wrate)$. We characterize the fundamental limits of the inference goals defined above by identifying, in this parameter space, when each goal is achievable or impossible. These achievable and impossible regions yield the phase transitions studied below.


\JRSSBparagraphvspace{-1.1em}
\paragraph{Detection boundary.}
We begin with global detection in Definition~\ref{def:detection}, which only asks whether any surviving watermark signal is present in the text. Following the terminology in~\citep{donoho2004higher,donoho2015higher,li2025robust}, we say that $H_0$ and $H_1$ in~\eqref{eq:previous-mixture-hypothesis} \emph{merge asymptotically} if the total variation distance between the joint distributions of $\Yars_{1:n}$ under $H_0$ and $H_1$ tends to zero. In this case, no test can reliably distinguish the two hypotheses. Conversely, if the two distributions separate asymptotically, detection is possible. The following theorem gives the detection boundary under the $(p,q,\Wrate)$ regime.


\begin{thm}[Detection boundary]
\label{thm:gumbel-general}
Under Assumptions~\ref{asmp:main}---~\ref{asmp:HCL+}, let $p,q\in[0,1]$ and $\Wrate\in[0,1)$.
    \begin{itemize}
        \item If $\max\{p+q,\,2p+q-\Wrate\}>1$, then $H_0$ and $H_1$ merge asymptotically. Hence, for any test, the sum of Type I and Type II errors tends to $1$ as $n\to\infty$.
        \item If $\max\{p+q,\,2p+q-\Wrate\}<1$, then $H_0$ and $H_1$ separate asymptotically. The likelihood-ratio test that rejects $H_0$ when the log-likelihood ratio is positive has a vanishing sum of Type I and Type II errors.
    \end{itemize}
\end{thm}


Theorem~\ref{thm:gumbel-general} gives the global detection boundary in the $(p,q,\Wrate)$ regime. Away from the boundary case, detection is possible exactly when both $p+q<1$ and $2p+q<1+\Wrate$ hold. Equivalently, the active boundary is $p+q=1$ when $p<\Wrate$, and $2p+q=1+\Wrate$ when $p\ge\Wrate$. These two constraints reflect two ways in which global detection can fail. The condition $p+q<1$ rules out the case where surviving watermark evidence is too sparse or too weak to produce visible tail deviations. The condition $2p+q<1+\Wrate$ captures the aggregate contribution of the growing light tail: when $\Wrate$ is larger, more low-probability alternatives are available, and their collective contribution can make $H_1$ distinguishable from $H_0$.
This result also clarifies the role of vocabulary growth. When $\Wrate=0$, the boundary reduces to the fixed-vocabulary detection boundary of~\citet{li2025robust}. When $\Wrate>0$, the detectable region expands because the light tail provides additional aggregate evidence. Thus, global detection can benefit from a growing vocabulary tail, although it still only answers whether surviving watermark signal exists somewhere in the text and does not identify its locations.

\JRSSBparagraphvspace{-1.1em}
\paragraph{Discovery boundary.}
We then turn to discovery in Definition~\ref{def:discovery}, the weakest form of token-level localization. Unlike global detection, discovery requires selecting actual token positions, and thus its analysis must control the dependence among local decisions. For this purpose, we impose an additional weak-dependence condition to rule out pathological long-range dependence in the generation and editing process.

\begin{asmp}[Geometric mixing]
\label{asmp:alphamixing}
Let $\bar{\mathcal G}_t:=\mathcal F_t\vee\mathcal H_t$ be the full generation-editing information up to time $t$, and define the future sigma-field $\bar{\mathcal G}_{t+k}^{+}:=\sigma(\token_s,\xi_s,\bP_{s+1},\theta_s:s\ge t+k)$. Define
\[
\alpha_{\bar{\mathcal G}}(k):=\sup_{t\ge0}\alpha(\bar{\mathcal G}_t,\bar{\mathcal G}_{t+k}^{+})
\quad \text{for all} \quad k\ge1,
\]
where $\alpha(\mathcal A,\mathcal B):=\sup_{A\in\mathcal A,\ B\in\mathcal B}|\mathbb P(A\cap B)-\mathbb P(A)\mathbb P(B)|$ is the strong-mixing coefficient between two $\sigma$-fields. We assume that this dependence decays geometrically: there exist universal constants $C>0$ and $\rho\in(0,1)$ such that $\alpha_{\bar{\mathcal G}}(k)\le C\rho^k$ for all $k\ge1$.
\end{asmp}

The coefficient $\alpha_{\bar{\mathcal G}}(k)$ measures the maximal dependence between the information up to time $t$ and the information starting from time $t+k$; see, e.g., \citet{bradley2005basic}. Thus, Assumption~\ref{asmp:alphamixing} requires the dependence between the past/current generation-editing process and the distant future to decay rapidly as the gap $k$ grows. It does not impose exact independence. Rather, it rules out long-range dependence that would make token-level localization difficult to analyze.

This condition is compatible with how practical LLM-generated text is produced. Although transformers can attend to previous tokens within their context window~\citep{vaswani2017attention}, deployed autoregressive LLMs operate with finite context windows. From this perspective, an autoregressive LLM with a fixed context length can be approximately viewed as a finite-memory Markov process~\citep{zekri2024large}. Empirical studies also suggest that long-context models do not use all positions in the context uniformly effectively: performance can degrade when relevant information is distant or poorly positioned in the prompt~\citep{liu2024lost,li2024loogle,an2025does}. Therefore, Assumption~\ref{asmp:alphamixing} allows both the generated tokens and the editing indicators to depend on previous text, but requires this dependence to weaken with distance. Technically, this weak-dependence condition provides the concentration control needed for empirical tail counts and false-discovery analysis in the discovery problem.

\begin{thm}[Discovery boundary]
\label{thm:discovery-boundary}
Under Assumptions~\ref{asmp:main}--\ref{asmp:alphamixing}, let $p,q\in[0,1]$ and $\Wrate\in[0,1)$. 
\begin{itemize}
    \item If $p<\Wrate$ and $p+q<1$, then there exists a local decision rule that achieves discovery.
    \item If $p\ge \Wrate$ or $p+q>1$, then no local decision rule can achieve discovery.
\end{itemize}
In particular, when the effective vocabulary size is fixed, that is, $\Wrate=0$, discovery is impossible.
\end{thm}

Theorem~\ref{thm:discovery-boundary} gives the phase transition for the weakest nontrivial localization task, discovery. The condition $p+q<1$ ensures that the surviving watermark evidence is strong enough to produce sufficiently extreme token-level evidence. The additional condition $p<\Wrate$ is specific to localization: it requires the light set $L_{t,n}$ to grow fast enough relative to the sparsity of surviving watermark positions, so that at least one watermark-preserving position can stand out from null positions. This also explains why discovery is impossible when $\Wrate=0$: with a fixed effective vocabulary size, the light set does not grow, and the token-level separation needed for localization disappears.

This boundary shows that discovery is strictly harder than global detection. Detection can succeed by aggregating weak watermark evidence over the whole text, whereas discovery requires at least one individual position to be reliably identified. Consequently, there are regimes where the text can be detected as containing surviving watermark signal, but no local decision rule can reliably locate even one such position. In terms of phase boundaries, discovery requires both $p<\Wrate$ and $p+q<1$, while detection only requires $\max\{p+q,2p+q-\Wrate\}<1$. Hence, away from boundary cases, the discoverable region is a strict subset of the detectable region.
Figure~\ref{fig:phase-alpha-all} illustrates this separation. As $\Wrate$ increases, the light set becomes larger and provides more low-probability alternatives, which strengthens token-level separation and makes both detection and discovery easier. At the same time, the detectable-but-undiscoverable region shrinks. In the degenerate case $\Wrate=0$, discovery is impossible, even though global detection may still be possible in some regimes.

\begin{figure}[t!]
\centering
\includegraphics[width=\textwidth]{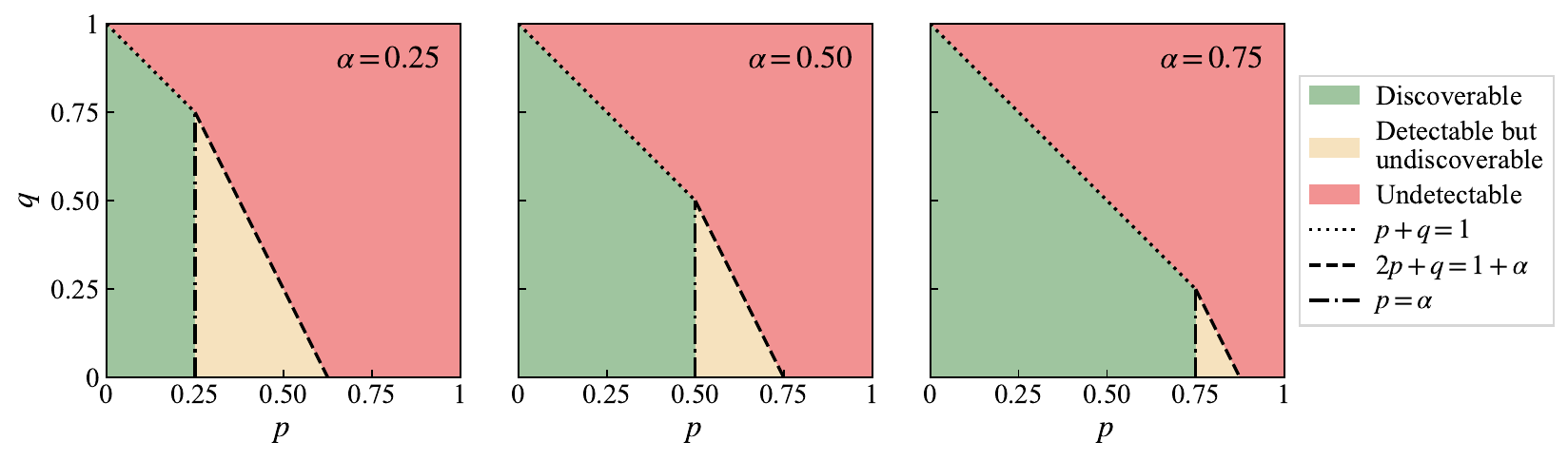}
\vspace{-20pt}
\caption{
Phase transitions for different values of $\Wrate\in\{0.25,0.50,0.75\}$.
The discoverable region is $p<\Wrate$ and $p+q<1$, while the undetectable region is $\max\{p+q,\,2p+q-\Wrate\}>1$.
The remaining region is detectable but not discoverable, illustrating the gap between detection and localization.
}
\label{fig:phase-alpha-all}
\vspace{-10pt}
\end{figure}

\JRSSBparagraphvspace{-1.1em} 
\paragraph{Impossibility of classification.}
Discovery focuses on false-discovery control for the selected set, whereas classification additionally requires near-complete recovery of the signal set: almost all watermark-preserving positions must be selected while the selected set still contains only a negligible fraction of null positions.
Under our general editing model, these two requirements cannot be made compatible. 
Indeed, $\theta_t=1$ only means that the watermark dependence at position $t$ survives; it does not guarantee that the pivotal statistic $\Yars_t$ exhibits strong local evidence. 
As a result, some watermark-preserving positions may look nearly indistinguishable from null positions. 
To avoid missing these weak positions, a decision rule would have to select more aggressively, but doing so would also include too many null positions and prevent the false discovery rate from vanishing. 
Conversely, a conservative rule may control false discoveries, but it must miss a non-negligible fraction of watermark-preserving positions. 
The next theorem shows that this incompatibility rules out classification throughout the $(p,q,\Wrate)$ regime.

\begin{thm}[Impossibility of classification]
\label{thm:class}
Suppose Assumptions~\ref{asmp:main}--\ref{asmp:HCL+} hold. Then, for any $p,q\in[0,1]$ and $\Wrate\in[0,1)$, no sequence of local decision rules can achieve classification.
\end{thm}

\begin{rem}[Contrast with Gaussian mixtures]
The impossibility in Theorem~\ref{thm:class} contrasts with the Gaussian mixture setting of~\citet{CaiSun2017OSD}, where classification can be possible when the mean shift is sufficiently large. 
There, a stronger signal shifts the alternative distribution away from the null, so most non-null observations can be separated from null observations. 
In watermark localization, however, $\theta_t=1$ only records the survival of watermark dependence; it does not imply that the pivotal statistic $\Yars_t$ is far from the null law $\mu_0$. 
Since watermark pivotal statistics are typically bounded, the null and alternative laws can still overlap substantially even when the watermark dependence survives. 
Thus, surviving watermark dependence is not analogous to a large mean shift, which explains why classification is impossible under the model considered here.
\end{rem}

\subsection{Adaptive Optimality}
\label{sec:optimality}

In this subsection, we establish the adaptive optimality of \Algo\ for discovery. 
Here, ``adaptive'' means that the algorithm does not require the problem-dependent parameters~\citep{donoho2004higher,donoho2015higher}, which in our setting include $p,q,\Wrate$ and the NTP distributions $\bP_{1:n}$. 
The term ``optimal'' refers to boundary optimality: the algorithm achieves discovery throughout the discoverable region.

\JRSSBparagraphvspace{-1.1em} 
\paragraph{Adaptive optimality for detection.}
We first revisit global detection for completeness. 
Although global detection is not the main focus of this paper, it is useful to ask whether an existing procedure can attain the detection boundary in the $(p,q,\Wrate)$ regime. 
We show that this is the case for the robust detection method of~\citet{li2025robust}, called \texttt{Tr-GoF}. 
This method is a truncated goodness-of-fit test: it compares the empirical distribution of the pivotal statistics with the null distribution $\mu_0$ after a suitable truncation, and rejects the global null $H_0$ when the deviation is sufficiently large. 
Although \texttt{Tr-GoF} was originally developed for the fixed-vocabulary setting, the following theorem shows that it remains adaptively optimal under vocabulary growth.

\begin{thm}[Adaptive optimality for global detection]
\label{thm:adaptivityTrGoF}
Suppose Assumptions~\ref{asmp:main}--\ref{asmp:alphamixing} hold. If $\max\{p+q,\,2p+q-\Wrate\}<1$, then \emph{\texttt{Tr-GoF}} achieves global detection.
\end{thm}

\JRSSBparagraphvspace{-1.1em} 
\paragraph{Adaptive optimality for discovery.}
We next turn to the main focus of this paper: adaptive discovery. The key result is that \Algo\ attains the discovery boundary in Theorem~\ref{thm:discovery-boundary} without knowing the parameters that determine this boundary. This adaptivity comes from two ingredients. 
First, \Algo\ scans over $M_n$ tail thresholds and selects the threshold from the data, rather than using the unknown optimal tail scale. 
Second, it uses an estimate $\widehat\varepsilon_n$ of the surviving watermark fraction, rather than requiring $\varepsilon_n$ as prior knowledge. 
A feasible high-accuracy choice of $\widehat\varepsilon_n$ is the fraction estimator of~\citet{li2025optimal}, which is also used in our experiments. 
For generality, the theorem below only requires the estimator to satisfy the accuracy condition in~\eqref{eq:estimator-accruacy-condition}.

\begin{thm}[Adaptive optimality for discovery]
\label{thm:adaptivity}
Suppose Assumptions~\ref{asmp:main}--\ref{asmp:alphamixing} hold. Run \emph{\Algo} in Algorithm~\ref{alg:adaptive-discovery-u} with $M_n\asymp(\log n)^a$ for some $a\ge1$, $\lambda_n=C_1/\log n$, and $\eta_n=C_2/(\log n)^2$, where $C_1,C_2>0$ are constants independent of $p,q,\Wrate,n$. Suppose the estimator $\widehat\varepsilon_n$ satisfies
\begin{equation}
\label{eq:estimator-accruacy-condition}
\mathbb P\!\left(|\widehat\varepsilon_n-\varepsilon_n|>c\eta_n\right)=o(n^{-1})
\end{equation}
for some constant $c>0$. If $p<\Wrate$ and $p+q<1$, then \emph{\Algo} achieves discovery, provided that $0< u_{\min}<p+q<u_{\max}<1$.
\end{thm}

\begin{rem}[Choice of the threshold grid]
The condition $u_{\min}<p+q<u_{\max}$ is a coverage condition for the threshold grid. 
It ensures that the scan includes the relevant tail scale around $p+q$, where the adaptive threshold is selected asymptotically. 
This condition does not require prior knowledge of $p$ or $q$: since discovery is possible only when $p+q<1$, one may take $u_{\min}$ close to $0$ and $u_{\max}$ close to $1$ in practice to achieve it. 
In our experiments, we use $[u_{\min},u_{\max}]=[0.005,0.98]$ for simulations and $[0.05,0.95]$ for LLM experiments.
\end{rem}

\JRSSBparagraphvspace{-1.1em} 
\paragraph{Near-optimal discovery power.}
Beyond boundary optimality, we also ask whether \Algo\ discovers nearly as many watermark-preserving positions as possible under false positive control. 
To make this comparison meaningful and interpretable, we compare \Algo\ with homogeneous per-token local rules based on the pivotal statistic \citep{sun2007oracle}. 
Specifically, define
\begin{equation}
\label{eq:admissible-class-FPR-homY}
\Decision^{\mathrm{hom}}(\lambda_n)
:=
\Bigl\{\Bdelta:\ \exists\,\varphi:[0,1]\to\{0,1\}\ \text{such that}\ 
\delta_t=\varphi(\Yars_t)\ \text{for all }t,\ 
\mathrm{mFDR}_{\Bdelta}\le \lambda_n
\Bigr\}.
\end{equation}
This class consists of rules that apply the same pivot-based decision function across positions and satisfy the same false positive rate constraint. 

\begin{thm}[Near-optimal number of discoveries]
\label{thm:25optimal}
Suppose the assumptions and setup in Theorem~\ref{thm:adaptivity} hold. 
Let $\Bdelta_{\emph{\Algo}}$ denote the decision rule returned by Algorithm~\ref{alg:adaptive-discovery-u} with the same parameters as in Theorem~\ref{thm:adaptivity}. 
Then there exist constants $c_1,c_2>0$ such that
\[
\mathrm{mFDR}_{\Bdelta_{\emph{\Algo}}}\le c_1\lambda_n+o(1),
\qquad
\mathrm{ETP}_{\Bdelta_{\emph{\Algo}}}
\ge c_2\cdot \sup_{\Bdelta' \in \Decision^{\mathrm{hom}}(\lambda_n)}
\mathrm{ETP}_{\Bdelta'}
+o(1).
\]
\end{thm}

Theorem~\ref{thm:25optimal} shows that \Algo\ is not only boundary-optimal but also quantitatively efficient. 
Among homogeneous local rules satisfying the target false positive rate constraint, \Algo\ achieves a constant fraction of the largest possible expected number of true discoveries, up to lower-order terms. 
Thus, the adaptive threshold chosen by \Algo\ does not only cross the correct phase boundary, but also retains near-oracle discovery power compared to a natural class of interpretable rules.

\begin{rem}[The choice of benchmark class]
\label{rem:why-hom-tail}
We use $\mathcal{D}_n^{\mathrm{hom}}$ as the benchmark class because it matches the information available to a practical verifier. 
After observing the final text, the verifier can compute a pivotal statistic $\Yars_t$ at each position, but does not observe the NTP distributions, the editing indicators, or other hidden generation states. 
Thus, a natural comparison is with rules that apply a common pivot-based decision function across positions under the same false positive rate constraint. 
We do not benchmark against fully time-varying or history-dependent rules, since such rules may rely on information unavailable to the verifier or on position-specific tuning that is not comparable to a practical localization procedure.
\end{rem}


%% file: 5-experiments.tex
\section{Simulations}
\label{sec:simulation}
In this section, we use simulations to verify the phase transition predicted by Theorem~\ref{thm:discovery-boundary} and to illustrate the adaptive behavior of our \Algo.

\subsection{Experimental Setup}
\label{sec:sim-setup}

We simulate pivotal statistics $Y_{1:n}$ from the mixture model in~\eqref{eq:mixture-Y} to study the theoretical phase transitions.
The simulation pipeline is summarized in Figure~\ref{fig:simulation-pipeline}.
We first specify the asymptotic parameters.
For a given text length $n$ and each triple $(p,q,\alpha)$, we set the surviving watermark fraction to $\varepsilon_n=0.5n^{-p}$ and the residual mass away from the dominant token to $\Delta_n=n^{-q}$.
In our implementation, we consider $\alpha\in\{0,0.25,0.5,0.75\}$ and set $|\mathcal W_n|=2+\lfloor n^\alpha\rfloor$, with one core token and $\lfloor n^\alpha\rfloor$ light tokens.
This corresponds to setting $r=0$ in Assumption~\ref{asmp:HCL+}.
Throughout this subsection, we use independent pseudorandom variables so that the simulation isolates the statistical structure of the localization problem.

Second, we generate each NTP distribution $\bP_t$ using a latent Markov process $Z_t\in\{0,1\}$.
The role of $Z_t$ is to allow the NTP distribution to vary mildly over time while keeping the same asymptotic structure.
Specifically, $Z_t$ is a two-state Markov chain started from stationarity, with $\mathbb P(Z_1=1)=1/2$ and $\mathbb P(Z_t=Z_{t-1}\mid Z_{t-1})=0.95$.
Thus, each regime, corresponding to a constant value of $Z_t$, tends to persist for many consecutive positions.
For each regime $z\in\{0,1\}$, we predefine an NTP distribution $\bP_n^{(z)}$ satisfying the dominant--core--light structure in Assumption~\ref{asmp:HCL+} and set $\bP_t=\bP_n^{(Z_t)}$.

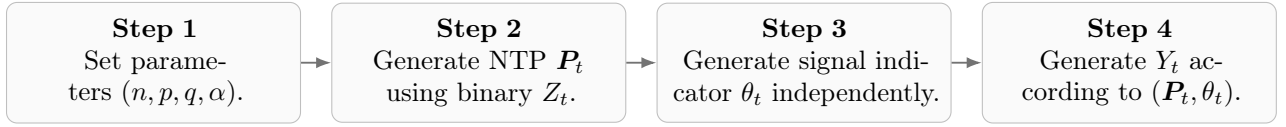
\begin{figure}[t!]
\centering
\begin{tikzpicture}[font=\small, node distance=4mm]
\tikzset{
  simbox/.style={
    draw=black!25,
    rounded corners=4pt,
    fill=black!2,
    align=center,
    text width=0.195\textwidth,
    inner sep=4pt
  },
  simarrow/.style={
    -{Latex[length=2mm]},
    line width=0.7pt,
    draw=black!55
  }
}
\node[simbox] (step1) {
\textbf{Step 1}\\
Set parameters $(n,p,q,\alpha)$.
};
\node[simbox, right=of step1] (step2) {
\textbf{Step 2}\\
Generate NTP $\bP_t$ using binary $Z_t$.
};
\node[simbox, right=of step2] (step3) {
\textbf{Step 3}\\
Generate signal indicator $\theta_t$ independently.
};
\node[simbox, right=of step3] (step4) {
\textbf{Step 4}\\
Generate $Y_t$ according to $(\bP_t,\theta_t)$.
};
\draw[simarrow] (step1) -- (step2);
\draw[simarrow] (step2) -- (step3);
\draw[simarrow] (step3) -- (step4);
\end{tikzpicture}
\caption{
Simulation pipeline for generating the pivotal statistics $\{Y_t\}_{t=1}^n$.
}
\label{fig:simulation-pipeline}
\vspace{-10pt}
\end{figure}

Third, we generate the survival indicators $\{\theta_t\}$ independently of $\{Z_t\}$, but with temporal dependence across positions.
The indicator $\theta_t\in\{0,1\}$ records whether the watermark signal survives at position $t$.
We sample $\{\theta_t\}$ from a two-state homogeneous Markov chain started from stationarity, with stationary mean $\mathbb P(\theta_t=1)=\varepsilon_n$.
Hence, on average, an $\varepsilon_n$ fraction of positions retain watermark signal, typically in short bursts rather than in complete isolation.

Finally, given $(\bP_t,\theta_t)$, we generate the pivotal statistic $Y_t$ using the Gumbel-max watermark from Section~\ref{sec:pre}.
If $\theta_t=0$, the watermark signal does not survive, and we sample $Y_t$ independently from the null law $\mu_0$.
If $\theta_t=1$, we sample directly from the watermark-induced alternative law $\mu_{1,\bP_t}$ using the equivalent scalar construction $Y_t=U^{P_{t,w_t}}$, where $w_t\sim\bP_t$ and $U\sim\mathrm{Unif}(0,1)$ are independent.
Thus, each $Y_t$ follows either the null $\mu_0$ or the watermark-induced alternative law $\mu_{1,\bP_t}$, according to $\theta_t$, while the dependence across positions is inherited from the geometrically mixing process $(\bP_t,\theta_t)$.

\subsection{Discovery Boundary and Empirical Phase Transitions}
\label{sec:sim-boundary}

We now visualize the empirical discovery boundary of \Algo.
For a decision rule $\Bdelta = (\delta_1, \ldots, \delta_n)$, we evaluate its finite-sample discovery performance by the following \textit{discovery error}
\begin{equation}
\label{eq:discovery-err}
\mathrm{Err}_{\Bdelta}
=
\mathrm{mFDR}_{\Bdelta}
+
p_{\mathrm{miss}}(\Bdelta),
\qquad
p_{\mathrm{miss}}(\Bdelta)
=
\mathbb P(|S_{\Bdelta}|=0),
\end{equation}
where $S_{\Bdelta}=\{t:\delta_t=1\}$ is the discovery set, $\mathrm{mFDR}_{\Bdelta}$ is the marginal false discovery rate, and $p_{\mathrm{miss}}(\Bdelta)$ is the probability that the rule $\Bdelta$ makes no discovery.
This error directly reflects the two requirements of discovery: controlling marginal false discoveries and ensuring a nontrivial discovery set (that contains at least one position).
Smaller values therefore indicate better discovery performance.

In the implementation of \Algo, the target false discovery level is set as $\lambda_n=C/\log n$, where $C$ is a calibration constant.
For each parameter setting, we tune this calibration constant over a predetermined grid and report the smallest resulting error.
This tuning is used only to visualize the empirical phase transition.
In the following, we consider two complementary views: one-dimensional phase-transition slices and two-dimensional heatmap diagrams.

\begin{figure}[t!]
\centering
\includegraphics[width=\textwidth]{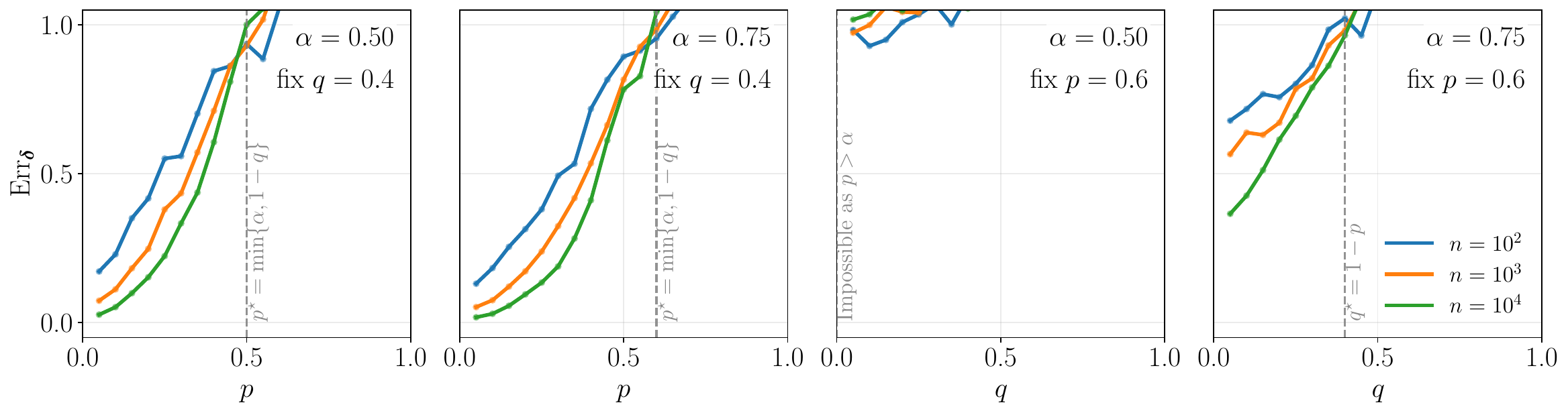}
\vspace{-20pt}
\caption{
Empirical discovery phase-transition slices for \Algo\ with $\alpha\in\{0.50,0.75\}$.
Each panel reports the smallest value of discovery errors defined in \eqref{eq:discovery-err} over a calibration grid along one slice, averaged over $200$ trials.
The vertical dashed lines show the theoretical discovery boundary.
}
\label{fig:discovery-boundary-panels7}
\vspace{-10pt}
\end{figure}

\JRSSBparagraphvspace{-1.1em}
\paragraph{Phase transition for a fixed $p$ or $q$.}
We first examine one-dimensional slices of the $(p,q)$ plane.
We either fix $q=0.4$ and vary $p$, or fix $p=0.6$ and vary $q$.
For each $\alpha\in\{0,0.25,0.5,0.75\}$ and $n\in\{10^2,10^3,10^4\}$, we run $200$ independent Monte Carlo trials.
For clarity, we present the representative slices for $\alpha\in\{0.5,0.75\}$ here.
Let $\mathcal R(a,b,K)=\{a+k(b-a)/(K-1):k=0,1,\ldots,K-1\}$ denote $K$ equally spaced grid points from $a$ to $b$.
For the fixed-$q$ slices, we evaluate the error over $p\in\mathcal R(0.05,0.95,19)$; for the fixed-$p$ slices, we evaluate the error over $q\in\mathcal R(0.05,0.95,19)$.
At each grid point, we tune the calibration constant $C$ over $\{0.05,0.10,\ldots,2.00\}$, set $\eta_n=0$, and report the smallest error.

According to Theorem~\ref{thm:discovery-boundary}, the discoverable region is $\{(p,q):p<\alpha,\ p+q<1\}$.
Thus, for fixed $q$, the critical transition in the $p$-direction is $p^\star(q;\alpha)=\min\{\alpha,1-q\}$.
In particular, when $q=0.4$, the predicted transition occurs at $p^\star(0.4;\alpha)=\min\{\alpha,0.6\}$.
For fixed $p$, discovery is possible only if $p<\alpha$, and then only for $q<1-p$.
When $p=0.6$, the predicted transition is therefore at $q=0.4$ if $\alpha>0.6$, while the entire slice is non-discoverable when $\alpha\le0.6$.
Figure~\ref{fig:discovery-boundary-panels7} confirms these predictions.
For fixed $q=0.4$, the error remains small when $p<p^\star(0.4;\alpha)$ and increases sharply after crossing the predicted boundary.
For fixed $p=0.6$, the transition occurs near $q=0.4$ when $\alpha>0.6$, while the error remains high across the whole slice when $\alpha\le0.6$.
As $n$ increases, the empirical transition becomes sharper and aligns more closely with the theoretical boundary.

\begin{figure}[t!]

\centering
\includegraphics[width=\textwidth]{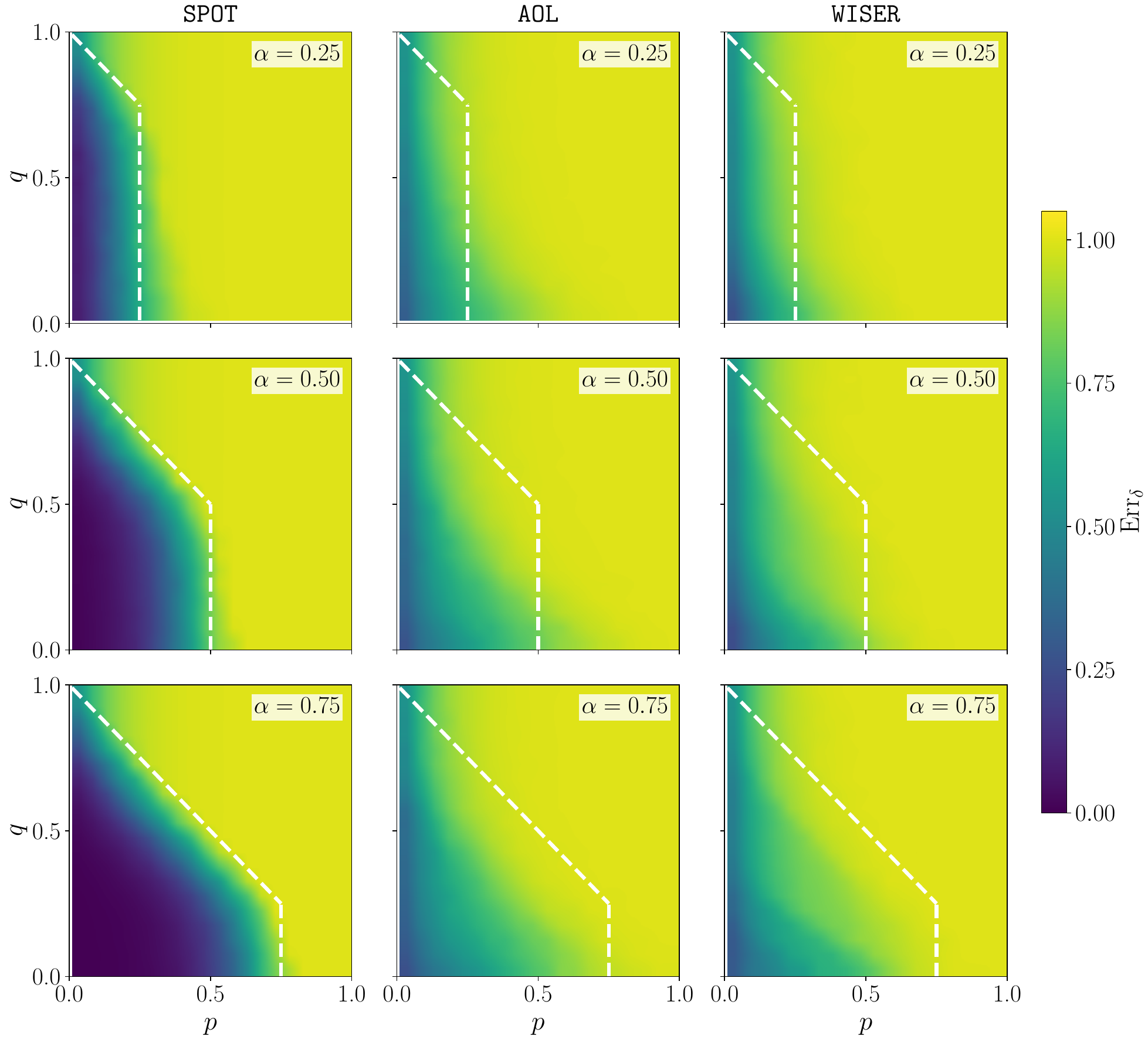}
\caption{
Empirical discovery heatmap diagrams for \Algo, \texttt{AOL}, and \texttt{WISER} with $n=10^4$.
Columns correspond to the three methods, and rows to $\alpha\in\{0.25,0.5,0.75\}$.
Each panel reports the smallest value of discovery errors in \eqref{eq:discovery-err} on a $(p,q)$ grid, averaged over $500$ trials, with displayed values clipped at $1.05$.
White dashed lines show the theoretical discovery boundaries.
}
\label{fig:discovery-heat-panels7}
\end{figure}

\paragraph{Heatmap diagrams.}
We next evaluate the empirical boundary over a two-dimensional $(p,q)$ grid, comparing \Algo\ with \texttt{AOL}~\citep{zhao2025efficiently} and \texttt{WISER}~\citep{bonnerjee2025fast}.
For the heatmap plots, we fix $n=10^4$.
For each $\alpha\in\{0.25,0.5,0.75\}$, we evaluate the error over $p\in\mathcal R(0.01,1,20)$ and $q\in\mathcal R(q_{\min},1,20)$, where $q_{\min}=\log_n(1+n^{-\alpha})$.\footnote{The lower bound on $q$ ensures that the tail threshold $\tau(q)=1-n^{-q}$ is at least $1/(1+n^\alpha)$, the inverse of the number of non-dominant tokens.}
For each $(p,q)$, we run $500$ independent Monte Carlo trials.
At each grid point, we tune the calibration constant for \Algo\ over $60$ log-spaced values from $0.01$ to $10$, and report the smallest error.
Further details on the parameter settings for \texttt{AOL} and \texttt{WISER} and the fraction estimator used by \Algo\ are provided in Supplementary Material~\ref{sec:additional-simulation}.
For visualization, displayed values are clipped at $1.05$, so that all larger errors are shown as equally unfavorable.

Figure~\ref{fig:discovery-heat-panels7} shows the resulting heatmaps.
Darker regions correspond to smaller discovery error, while lighter regions correspond to larger error.
The dashed curves mark the theoretical boundaries $p=\alpha$ and $p+q=1$ from Theorem~\ref{thm:discovery-boundary}.
Across all three values of $\alpha$, the low-error region for \Algo\ lies mostly inside the theoretically discoverable region $\{(p,q):p<\alpha,\ p+q<1\}$, while the high-error region dominates outside this region.
Overall, the empirical transition bands align well with the predicted discovery boundary.
Compared with \texttt{AOL} and \texttt{WISER}, \Algo\ exhibits a larger low-error region within the discoverable region, particularly for $\alpha=0.5$ and $0.75$.
This comparison illustrates its advantage under the simulated short-burst survival process, where surviving watermark evidence is fragmented.

\section{Open-source Model Experiments}
\label{sec:LLM-experiments}

\subsection{Experiment Setup}

We follow the setup of~\citet{li2025robust} and evaluate discovery performance on open-source LLMs under controlled edits. Specifically, we sample $1000$ documents from the news-like subset of the C4 dataset~\citep{raffel2020exploring}. For each document, we use the last $50$ tokens as the prompt and ask OPT-1.3B~\citep{zhang2022opt} to generate an additional $n=400$ tokens as the continuation.
During watermark generation, each pseudorandom variable is computed from the previous $m=3$ tokens. We apply repeated-context masking, which adds the watermark only when the length-$m$ prefix context has not appeared earlier in the generated history. This technique aims to reduce the frequency of watermarking and improve text quality~\citep{dathathri2024scalable}. After obtaining the watermarked text, we apply the considered edit mechanisms to simulate human editing. We conduct experiments at temperatures $T\in\{0.5,0.7,1\}$, ranging from low- to high-temperature generation.

\JRSSBparagraphvspace{-1em}
\paragraph{Ground-truth labels and surviving fraction.}
For post-edit text, we define ground-truth labels by comparing the final text with the original pre-edit watermarked token sequence, following~\citet{li2025optimal}.
We label a position as watermark-preserving only if the current token and its preceding $m$ tokens match a contiguous block in the original pre-edit sequence. We also require that every position in the corresponding original block was actually watermarked.
We denote the resulting ground-truth labels by $\theta^{\mathrm{gt}}_1,\ldots,\theta^{\mathrm{gt}}_n$.
The oracle surviving watermark fraction is then defined as $\varepsilon_{\mathrm{true}}=n^{-1}\sum_{t=1}^n \theta^{\mathrm{gt}}_t$.
The detailed labeling rules are given under \hyperref[par:LLM-ground-truth]{Ground-truth correspondence in Supplementary Material~\ref*{sec:LLM-details}}.

\JRSSBparagraphvspace{-1em}
\paragraph{Baseline methods.}
We compare \Algo\ with \texttt{AOL} (\textbf{A}daptive \textbf{O}nline \textbf{L}ocator), the token-level localization method in Algorithm 2 of~\citet{zhao2025efficiently}, and \texttt{WISER}, the change-point detection method of~\citet{bonnerjee2025fast}.\footnote{\texttt{WISER} required substantial parameter tuning in our experiments, as its default parameter performed poorly. All reported \texttt{WISER} results use configurations selected from the parameter grids described in Supplementary Material~\ref{sec:LLM-methods}.}
All methods output token-level decisions indicating whether each position preserves watermark evidence. For \texttt{WISER}, we convert its estimated watermarked intervals into token-level decisions. For \Algo, we report two variants: \Algo{}\texttt{-oracle}, which uses the oracle surviving watermark fraction $\varepsilon_{\mathrm{true}}$, and \Algo{}\texttt{-plugin}, which uses the plug-in fraction estimator from~\citet{li2025optimal}. Since all methods involve calibration constants, such as the constant $C$ in $\lambda_n=C/\log n$ for \Algo, we report the best performance after tuning these constants over fixed grids.

\JRSSBparagraphvspace{-1em}
\paragraph{Evaluation metrics.}
We consider several types of human modifications and evaluate token-level localization under each edit type.  Let $S^\star$ denote the set of true watermark-preserving locations, and let $\widehat S$ denote the set of locations selected by a method. We report two localization metrics: the empirical true-positive rate (TPR), defined as $\widehat{\mathrm{TPR}}=|\widehat S\cap S^\star|/|S^\star|$, and the intersection-over-union (IoU), defined as $\mathrm{IoU}=|\widehat S\cap S^\star|/|\widehat S\cup S^\star|$.
The former measures the fraction of true watermark-preserving locations recovered, while the latter provides a stricter overlap measure that penalizes both missed locations and extra selected locations.
We also compute the empirical false-positive rate (FPR) as $\widehat{\mathrm{FPR}}=|\widehat S\cap (S^\star)^c|/|(S^\star)^c|$, which measures the fraction of null locations incorrectly selected as watermark-preserving.
While our theoretical discovery criterion is formulated in terms of
mFDR, we use the conventional token-level FPR in the empirical
comparison because it provides a common and directly interpretable
operating point across localization methods.
To compare methods under the same false-positive control, we select parameters separately for TPR and IoU within the highest nonempty $0.01$-wide empirical-FPR bin not exceeding the prescribed upper bound. The metrics are averaged across documents and then across edit levels. Details of averaging and parameter selection are given in Supplementary Material~\ref{sec:LLM}.

\subsection{Localization Performance}
\label{sec:robust-evaluation}

Following~\citet{li2025robust}, we evaluate the localization performance of \Algo{}, \texttt{AOL}, and \texttt{WISER} under three types of text modifications: (i) random edits, (ii) adversarial edits, and (iii) block permutation.
Random edits include substitution, insertion, and deletion.
For a given edit rate, defined as the fraction of pre-edit tokens selected for modification, we randomly select that fraction of pre-edit tokens and either replace them, insert new tokens before them, or delete them.
For substitutions and insertions, the new tokens are sampled uniformly from the vocabulary $\mathcal W$.
Adversarial edits are more targeted: under a fixed edit budget, we divide the text into approximately equal-length regions, randomly sample four tokens in each region, and replace the one with the largest pre-edit pivotal value.
For block permutation, we divide the text into short, contiguous blocks and randomly reorder them, keeping the tokens and their order within each block unchanged. This changes the context across block boundaries and can disrupt the watermark signals of nearby tokens.
The edit mechanisms and post-edit processing are detailed in Supplementary Material~\ref{sec:LLM}.


\begin{figure}[t!]
\centering
\includegraphics[width=\textwidth]{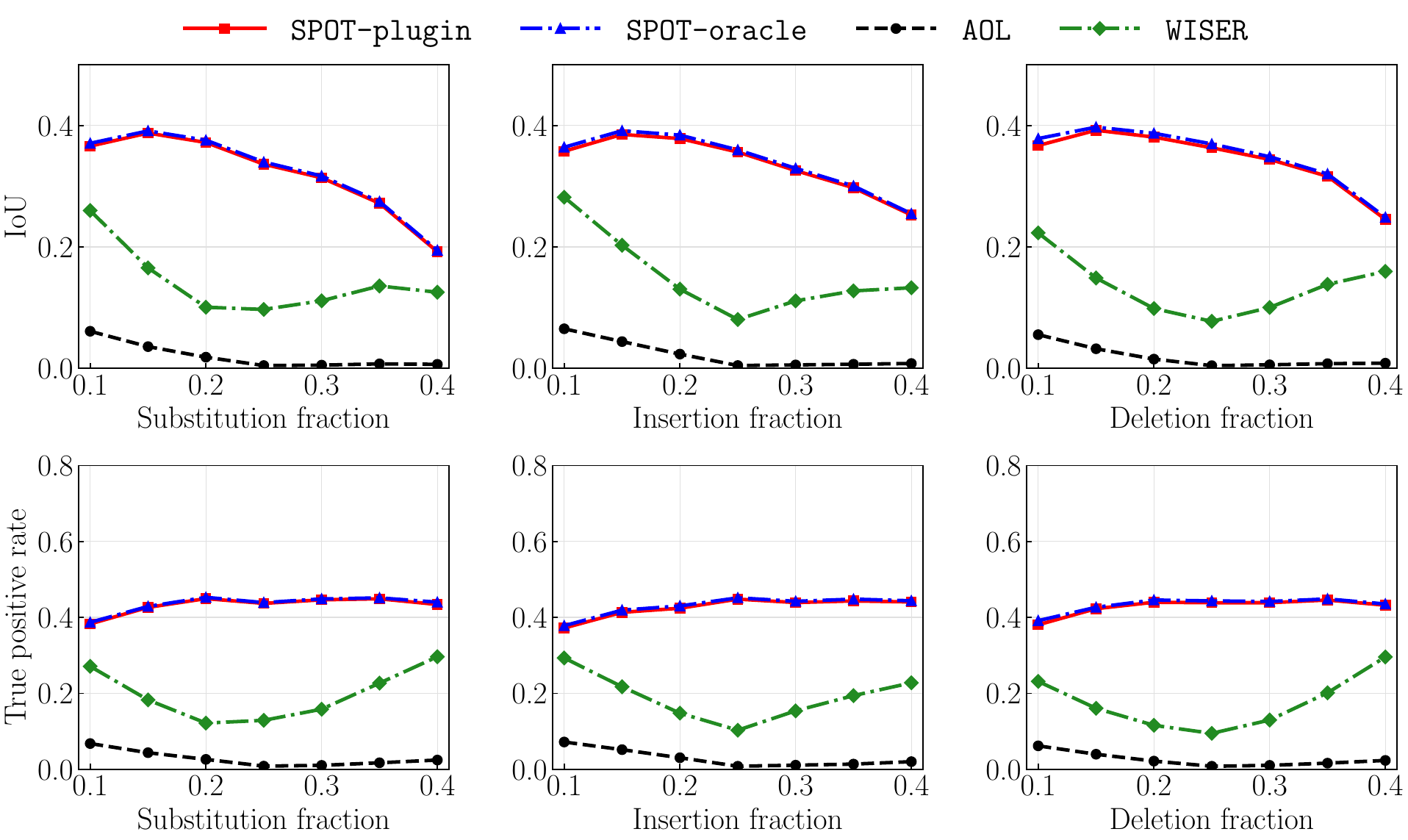}
\caption{
Localization performance under three random edit mechanisms at temperature $T=1$, measured by IoU (top row) and TPR (bottom row).
Columns correspond to substitution, insertion, and deletion edits from left to right, with target FPR $0.05$.
Each panel reports the mean performance as a function of the edit fraction.
}
\label{fig:real-llm-random-edits-T1}
\end{figure}

\paragraph{Results for random edits.}

The random-edit results at temperature $T=1$ are shown in Figure~\ref{fig:real-llm-random-edits-T1}.
Table~\ref{tab:real-llm-iou-lambda-005} further reports the IoU and TPR averaged over those edit levels $\{0.1, 0.15, 0.2, 0.25, 0.3, 0.35, 0.4\}$ for each temperature and each random edit type.
As expected, increasing the edit fraction generally decreases the IoU of \Algo\ across all random edit types, since more watermark signals are removed or disrupted.
In contrast, the TPR curves of \Algo\ are nearly flat as the edit fraction increases, showing that the method continues to recover a stable fraction of the surviving watermark-preserving locations under the same FPR constraint.
This is consistent with the definition $\widehat{\mathrm{TPR}}=|\widehat S\cap S^\star|/|S^\star|$: even though stronger edits reduce the number of surviving locations, \Algo\ identifies a similar proportion of those that remain.
Thus, the decline in IoU mainly reflects the increasing difficulty of overlap-based localization under heavier edits, while the stable TPR indicates that the tail-thresholding rule remains robust.
We next examine how these patterns vary across methods and generation temperatures.

\begin{enumerate}
\item \Algo\ \textit{performs better at moderate and high temperatures.}
At $T=1$, both \Algo{}\texttt{-oracle} and \Algo{}\texttt{-plugin} substantially outperform \texttt{AOL} and \texttt{WISER} across the three random edit types, in both IoU and TPR. The advantage is especially clear in the TPR curves in Figure~\ref{fig:real-llm-random-edits-T1}: both variants of \Algo\ recover a much larger fraction of the surviving watermark-preserving locations across the whole range of edit fractions. At $T=0.7$, the improvement becomes smaller but remains consistent. As shown in Table~\ref{tab:real-llm-iou-lambda-005}, both variants retain higher average IoU and TPR than the two baselines under all three random-edit settings. This suggests that \Algo\ is particularly effective when the LLM output is more diverse, where more usable watermark evidence remains available for localization.

\item \Algo\ \textit{remains competitive at low temperatures in its oracle version.}
At lower temperatures, the gap between methods becomes smaller because watermark localization is intrinsically harder. As shown in Table~\ref{tab:real-llm-iou-lambda-005}, at $T=0.5$, \Algo{}\texttt{-oracle} still achieves the highest average IoU and TPR under all three random edit types. By contrast, \Algo{}\texttt{-plugin} is closer to \texttt{WISER}, with slightly lower IoU under random deletion. The larger oracle--plug-in gap at low temperatures is consistent with the higher relative estimation errors reported in Supplementary Material~\ref{appen:new-LLM-results}, suggesting that fraction estimation contributes to the additional performance loss of the plug-in version. The lower performance of \Algo\ and \texttt{WISER} at low temperatures is consistent with prior empirical observations that low-temperature generation makes LLM outputs more deterministic, leaves fewer usable watermark signals, and thereby makes watermark inference more difficult~\citep{kirchenbauer2023watermark,lu2024entropy,he2026empirical,tsur2025heavywater}.
\end{enumerate}

\begin{table}[t!]
\centering\small
\setlength{\tabcolsep}{4pt}
\caption{
Average IoU and TPR across edit levels at different generation temperatures.
For each metric, the tuning parameter is selected separately within the highest nonempty $0.01$-wide empirical-FPR bin not exceeding $0.05$.
For random edits, entries are averaged over edit rates $\{0.10,0.15,0.20,0.25,0.30,0.35,0.40\}$.
For adversarial edits, entries are averaged over edit budgets $K\in\{40,45,50,55,60\}$.
For block permutation, entries are averaged over block sizes $b\in\{4,6,8\}$.
Bold marks the largest value among the four methods for each edit type, temperature, and metric.
}
\label{tab:real-llm-iou-lambda-005}
\resizebox{\textwidth}{!}{%
\begin{tabular}{llrrrrrr}
\toprule
\multirow{2}{*}{Edit type} & \multirow{2}{*}{Method} & \multicolumn{3}{c}{IoU} & \multicolumn{3}{c}{TPR} \\
\cmidrule(lr){3-5}\cmidrule(lr){6-8}
 & & $T=1$ & $T=0.7$ & $T=0.5$ & $T=1$ & $T=0.7$ & $T=0.5$ \\
\midrule
\multirow{4}{*}{Random substitution} & \texttt{AOL} & 0.0200 & 0.0260 & 0.0202 & 0.0282 & 0.0421 & 0.0437 \\
 & \texttt{WISER} & 0.1422 & 0.0783 & 0.0480 & 0.1980 & 0.1211 & 0.0941 \\
 & \Algo{}\texttt{-plugin} & 0.3199 & 0.1241 & 0.0487 & 0.4322 & 0.2017 & 0.1183 \\
 & \Algo{}\texttt{-oracle} & \textbf{0.3233} & \textbf{0.1312} & \textbf{0.0618} & \textbf{0.4355} & \textbf{0.2084} & \textbf{0.1341} \\
\midrule
\multirow{4}{*}{Random insertion} & \texttt{AOL} & 0.0226 & 0.0296 & 0.0246 & 0.0294 & 0.0424 & 0.0455 \\
 & \texttt{WISER} & 0.1525 & 0.0911 & 0.0563 & 0.1909 & 0.1256 & 0.0948 \\
 & \Algo{}\texttt{-plugin} & 0.3362 & 0.1384 & 0.0576 & 0.4259 & 0.2024 & 0.1202 \\
 & \Algo{}\texttt{-oracle} & \textbf{0.3405} & \textbf{0.1461} & \textbf{0.0746} & \textbf{0.4305} & \textbf{0.2101} & \textbf{0.1365} \\
\midrule
\multirow{4}{*}{Random deletion} & \texttt{AOL} & 0.0186 & 0.0236 & 0.0210 & 0.0257 & 0.0358 & 0.0401 \\
 & \texttt{WISER} & 0.1352 & 0.0883 & 0.0592 & 0.1756 & 0.1201 & 0.0959 \\
 & \Algo{}\texttt{-plugin} & 0.3441 & 0.1418 & 0.0587 & 0.4284 & 0.2003 & 0.1059 \\
 & \Algo{}\texttt{-oracle} & \textbf{0.3500} & \textbf{0.1518} & \textbf{0.0736} & \textbf{0.4334} & \textbf{0.2092} & \textbf{0.1199} \\
\midrule
\multirow{4}{*}{Adversarial edits} & \texttt{AOL} & 0.0754 & 0.0735 & 0.0528 & 0.0851 & 0.0861 & 0.0732 \\
 & \texttt{WISER} & 0.2798 & \textbf{0.1349} & \textbf{0.0860} & 0.2976 & \textbf{0.1501} & \textbf{0.1093} \\
 & \Algo{}\texttt{-plugin} & 0.3259 & 0.1255 & 0.0417 & 0.3480 & 0.1439 & 0.0750 \\
 & \Algo{}\texttt{-oracle} & \textbf{0.3311} & 0.1343 & 0.0649 & \textbf{0.3523} & 0.1499 & 0.0872 \\
\midrule
\multirow{4}{*}{Block permutation} & \texttt{AOL} & 0.0816 & 0.0740 & 0.0491 & 0.0896 & 0.0837 & 0.0681 \\
 & \texttt{WISER} & 0.2288 & 0.1399 & \textbf{0.0911} & 0.2401 & 0.1524 & \textbf{0.1101} \\
 & \Algo{}\texttt{-plugin} & 0.3482 & 0.1420 & 0.0394$^{\dagger}$ & 0.3656 & 0.1581 & 0.0620$^{\dagger}$ \\
 & \Algo{}\texttt{-oracle} & \textbf{0.3579} & \textbf{0.1600} & 0.0781 & \textbf{0.3750} & \textbf{0.1737} & 0.0965 \\
\bottomrule
\end{tabular}}
\par\smallskip
\begin{minipage}{\textwidth}\footnotesize
$^{\dagger}$At $T=0.5$, the permutation plug-in average includes lower-bin selections for $b=8$, with achieved FPR $0.01503$. For both separately selected metrics, these points occupy $(0.01,0.02]$ and satisfy the cap without reaching the intended top bin.
\end{minipage}
\end{table}

\paragraph{Results for adversarial edits.}
We next consider adversarial edits, where the editor is assumed to have access to the hash function $\AM$ and the secret key $\Key$ and can therefore target tokens carrying the strongest watermark signals.
To approximate this setting, we first compute the pivotal statistics for the LLM-generated response, then divide the continuation into $K$ contiguous regions, sample four candidate tokens within each region, and replace the candidate with the largest pivotal value by a uniformly sampled token.
Here, $K\in\{40,45,50,55,60\}$ controls the edit budget and reflects the modification strength.
The results are reported in the fourth edit-type group of Table~\ref{tab:real-llm-iou-lambda-005}.

At $T=1$, both \Algo{}\texttt{-plugin} and \Algo{}\texttt{-oracle} outperform \texttt{AOL} and \texttt{WISER} in both metrics.
At $T=0.7$, \Algo{}\texttt{-oracle} and \texttt{WISER} perform comparably, while \Algo{}\texttt{-plugin} remains competitive.
At $T=0.5$, all methods have low IoU and TPR, with small absolute differences, although \texttt{WISER} achieves slightly higher scores than both variants of \Algo.
Overall, the advantage of \Algo\ under adversarial edits is clearest at high temperature.
As with random edits, lower temperatures leave less usable watermark evidence, making localization harder and narrowing the absolute performance gaps across methods.


\paragraph{Results for block permutation.}
For block permutation, we partition the continuation into blocks of size $b\in\{4,6,8\}$. We randomly select approximately $75\%$ of the original adjacent block pairs and reorder the blocks so that these pairs no longer occur consecutively in their original order, while the remaining pairs stay together. This experiment models the effect of a reader reordering text segments on watermark localization while keeping the tokens and their order within each segment unchanged. The construction is detailed in Supplementary Material~\ref{sec:LLM-details}.
The IoU and TPR results are reported in the last edit-type group of Table~\ref{tab:real-llm-iou-lambda-005}.
At $T=1$ and $T=0.7$, both \Algo{}\texttt{-oracle} and \Algo{}\texttt{-plugin} outperform \texttt{AOL} and \texttt{WISER} in both averaged metrics.
At $T=0.5$, all methods have low IoU and TPR, and the differences are small in absolute terms, although \texttt{WISER} achieves slightly higher scores than both variants of \Algo.
The larger oracle--plug-in gap at this temperature is consistent with the higher relative fraction-estimation errors reported in Supplementary Material~\ref{appen:new-LLM-results}.

\begin{table}[t!]
\centering\small
\caption{Mean runtime in milliseconds per document at $T=1$, using the IoU-selected configuration at empirical FPR target $0.05$. Times include pivot recomputation, localization, and shared preparation/calibration amortized over $1000$ documents. They exclude parameter search, data loading, and output writing. The final row weights the five edit types equally. Bold marks the shortest time in each row.}
\label{tab:real-llm-runtime-T1}
\begin{tabular}{lrrrr}
\toprule
 & \multicolumn{4}{c}{Time (ms)} \\
\cmidrule(lr){2-5}
Edit type & \texttt{AOL} & \texttt{WISER} & \Algo{}\texttt{-plugin} & \Algo{}\texttt{-oracle} \\
\midrule
Random substitution & 168.9 & 152.7 & 157.8 & \textbf{152.3} \\
Random insertion & 168.8 & 152.7 & 157.8 & \textbf{152.4} \\
Random deletion & 85.2 & 76.4 & 82.1 & \textbf{76.1} \\
Adversarial edits & 168.3 & 152.3 & 157.1 & \textbf{151.7} \\
Block permutation & 167.2 & 151.6 & 156.8 & \textbf{151.3} \\
\midrule
\textbf{Average} & 151.7 & 137.1 & 142.3 & \textbf{136.8} \\
\bottomrule
\end{tabular}
\end{table}

\paragraph{Computational cost.}
Table~\ref{tab:real-llm-runtime-T1} reports the wall-clock runtime at $T=1$. The total cost consists of two components: (i) recomputing the token-level pivotal statistics from the edited text and (ii) running the localization method once the pivots are available. Pivot recomputation is shared by all methods and dominates the overall runtime. Given the computed pivots, the average localization time, including amortized preparation, is $5.64$ ms per sample for \Algo{}\texttt{-plugin}, compared with $15.0$ ms for \texttt{AOL} and $0.451$ ms for \texttt{WISER}. Thus, \Algo{}\texttt{-plugin} is faster than \texttt{AOL}, while \texttt{WISER} has the lowest localization cost. The shared pivot computation makes the total runtimes much closer across methods. These timings exclude parameter tuning, which can be substantially more costly for \texttt{WISER} because of its larger hyperparameter search space; see Supplementary Material~\ref{sec:LLM-methods} for details.

%% file: 6-discuss.tex
\section{Discussion}
\label{sec:discuss}

This paper studies watermark localization in mixed-source LLM text from a statistical perspective. We formulate localization as a token-level multiple testing problem based on pivotal statistics, where each position has a latent indicator recording whether the watermark dependence survives editing. Under a regime that captures sparse surviving signals, concentrated next-token distributions, and growing vocabularies, we characterize the statistical limits of three inference goals: global detection, discovery, and classification. Our results show that discovery is strictly harder than global detection, and that consistent classification is impossible within the class of coordinatewise pivot-based localization rules. We then propose \Algo{}, an adaptive tail-thresholding method that does not require knowledge of the asymptotic exponents or the time-varying
next-token distributions. The method achieves the optimal discovery boundary and attains near-optimal discovery power among natural local rules. Simulations support the theoretical phase transitions, while
real-LLM experiments show the localization performance
of \Algo{} under common edit mechanisms.

Several directions remain open. First, our current implementation uses the fraction estimator of~\citet{li2025optimal} to estimate the surviving watermark fraction. Although our theory allows any estimator satisfying a suitable accuracy condition, the experiments show that this step can become a practical bottleneck, especially at low temperatures where next-token distributions are more concentrated. This raises the question of whether localization can be performed without a separate fraction-estimation step. More broadly, the current procedure is batch-based: it uses a global fraction estimate for the entire text before choosing the localization threshold. An interesting direction is to develop online or streaming localization methods, where text arrives sequentially and the threshold is updated using accumulating evidence. Such methods may better adapt to documents whose source composition changes over time and reduce the need for a fixed global estimate of the surviving watermark fraction.

Second, our impossibility result shows that consistent classification is impossible within the class of coordinatewise localization rules. This leaves open the possibility of stronger recovery guarantees under additional structural assumptions or for localization procedures that exploit information across multiple token positions. Our model allows the latent survival indicators to switch frequently between watermark-preserving and null states, making exact token-level recovery too demanding. In practice, however, mixed-source text may have more block-like structure~\citep{li2024segmenting}: LLM-generated passages, human-written passages, and heavily revised segments may each persist over multiple consecutive tokens. Under such segment-level regularity, or when the target is region-level rather than exact token-level recovery, stronger forms of localization may become possible. Such structure can be exploited by classic change-point detection methods, which use positional order to recover contiguous watermarked segments. Our method instead targets sparse, potentially fragmented evidence without exploiting the ordering of the observed pivotal statistics. This distinction raises an interesting question: can a single localization method adapt to both regimes and attain optimal performance without prior knowledge of the underlying survival structure? We leave it as future work.

Third, our analysis focuses on token-level watermark localization. For Gumbel-max and other token-level watermarking schemes, the extension is relatively direct whenever valid pivotal statistics can be constructed for individual token positions. Semantic or sentence-level watermarks are different \citep{hou2024semstamp,ren2023robust,huo2026pmark}: they may encode watermark information through sentence meanings, paraphrase-invariant features, or vector representations of larger text units rather than token-level dependence on pseudorandomness. In such settings, the localization unit may be a sentence, span, or semantic embedding, and the null-versus-signal formulation must be redefined. Extending localization theory to these non-token-level watermarks is an important open direction.

%% file: 7-appendix.tex
\hypersetup{linkcolor=black,citecolor=black,urlcolor=black}
\newpage
\appendix

\begingroup
\makeatletter
\let\oldpartname\partname
\let\oldthepart\thepart
\renewcommand{\partname}{}
\renewcommand{\thepart}{}
\part{Supplementary Material}
\let\partname\oldpartname
\let\thepart\oldthepart
\makeatother
\endgroup

This supplementary material is organized as follows. Section~\ref{sec:Markov} gives an example of a generation--editing process satisfying our model and dependence assumptions. Section~\ref{proof} provides proofs of the theoretical results. Section~\ref{sec:additional-simulation} presents additional details and results for the simulation studies. Section~\ref{sec:LLM} describes the language model experiments, implementation details of the compared methods, and additional empirical results. Codes of the experiments are available at \url{https://github.com/SHL20/WatermarkLocalization}.

\parttoc

\section{An Example of a Generation--Editing Process}
\label{sec:Markov}

This section gives an example of a joint generation--editing process satisfying Assumptions~\ref{asmp:main} and~\ref{asmp:alphamixing}.
The example illustrates how dependence between generation and editing can be modeled through a finite-memory state summary and exogenous editing innovations.

\begin{defn}[Coupled generation--editing process]
\label{def:coupled_mechanism}
We introduce the following notations.
\begin{itemize}
    \item For $t\ge0$, define $\mathcal F_t:=\sigma(\{w_j,\zeta_j,\bP_{j+1}\}_{j=1}^{t})$ as the generation $\sigma$-field, with $\mathcal F_0$ trivial.
Let $\varphi$ be an injective measurable map into a state space $\mathsf X$, with a measurable inverse on its image, and define $X_t:=\varphi(w_t,\zeta_t,\bP_{t+1}) \in \mathsf X$ for $t\ge1$.
Then $\mathcal F_t=\sigma(X_1,\ldots,X_t)$.

\item Let $\{U_t\}_{t\ge1}$ be i.i.d.\ random variables, independent of $(X_0,\theta_0)$, where $\theta_0\in\{0,1\}$.
Set $\mathcal H_0:=\sigma(\theta_0)$ and define $\theta_t=\Phi_n(\theta_{t-1},X_{t-1},U_t)$ for $t\ge1$, where $\Phi_n:\{0,1\}\times\mathsf X\times\mathsf U\to\{0,1\}$ is measurable and may depend on $n$.
Let $\mathcal H_t:=\sigma(\theta_1,\ldots,\theta_t)$

\item Define $\mathcal G_0:=\mathcal F_0\vee\mathcal H_0$, $\mathcal G_t:=\mathcal F_{t-1}\vee\mathcal H_t$ for $t\ge1$, and $\bar{\mathcal G}_t:=\mathcal F_t\vee\mathcal H_t$.
Assume $U_t\perp\!\!\!\perp\mathcal G_{t-1}$ for all $t\ge1$.
For $k\ge1$, define $\bar{\mathcal G}_{t+k}^{+}:=\sigma(X_s,\theta_s:s\ge t+k)$.

\item Finally, suppose that for each $n$ there exists a time-homogeneous Markov kernel $K_n$ on $\mathsf X$ such that, for all $t\ge1$ and $A\in\mathcal B(\mathsf X)$,
\[
\mathbb P(X_t\in A\mid \mathcal F_{t-1}\vee\sigma(\theta_t))
=
K_n(X_{t-1},\theta_t;A).
\]
Thus, conditional on $(X_{t-1},\theta_t)$, the next generation state $X_t$ is independent of earlier history.
Define $S_t:=(X_t,\theta_t)$ and $\mathcal S_t:=\sigma(S_0,\ldots,S_t)$.
\end{itemize}
\end{defn}

\begin{lem}[Joint Markov property]
\label{lem:S_markov_theta_not}
For each fixed $n$, $\{S_t\}_{t\ge0}$ is a time-homogeneous Markov chain with respect to $\{\mathcal S_t\}$.
The marginal process $\{\theta_t\}_{t\ge0}$ need not be Markov unless $\mathbb P(\theta_t=1\mid X_{t-1},\theta_{t-1})$ is almost surely a function of $\theta_{t-1}$ alone.
\end{lem}

\begin{proof}[Proof of Lemma~\ref{lem:S_markov_theta_not}]
By construction, $\theta_t=\Phi_n(\theta_{t-1},X_{t-1},U_t)$, and the fresh-innovation condition implies that the conditional law of $\theta_t$ given $\mathcal S_{t-1}$ depends only on $(X_{t-1},\theta_{t-1})$.
Given $(X_{t-1},\theta_t)$, the Markov-kernel assumption gives the conditional law of $X_t$ as $K_n(X_{t-1},\theta_t;\cdot)$, independent of earlier history.
Thus the conditional law of $S_t=(X_t,\theta_t)$ given $\mathcal S_{t-1}$ depends only on $S_{t-1}$.

To see that $\{\theta_t\}$ need not be Markov, take $\mathsf X=\{0,1\}^2$, let $U_t$ be unused, set $\Phi_n(i,(a,b),u):=a$, and let $K_n((a,b),a;\cdot)=\delta_{(b,a)}(\cdot)$.
Then $X_t=(\theta_{t-1},\theta_t)$ and $\theta_{t+1}=\theta_{t-1}$.
Hence $\mathbb P(\theta_{t+1}=1\mid\theta_t=0,\theta_{t-1}=1)=1$, while $\mathbb P(\theta_{t+1}=1\mid\theta_t=0,\theta_{t-1}=0)=0$, so the law of $\theta_{t+1}$ given $\theta_t$ depends on $\theta_{t-1}$.
\end{proof}

\begin{thm}[Verification of the main assumptions]
\label{thm:relaxation_implies_A31_A32}
Assume the following two conditions.

\begin{enumerate}[label=(\roman*),leftmargin=2.5em]
\item There exist constants $0<c\le C<\infty$ and a sequence $\varepsilon_n\in(0,1)$ such that $c\varepsilon_n\le\mathbb P(\theta_t=1\mid X_{t-1},\theta_{t-1})\le C\varepsilon_n$ almost surely for all $t\ge1$. Equivalently, with $p_n(i,x):=\mathbb P(\Phi_n(i,x,U_t)=1)$, we have $p_n(i,x)\in[c\varepsilon_n,C\varepsilon_n]$ for all $i\in\{0,1\}$ and $x\in\mathsf X$.

\item The Markov chain $\{S_t\}_{t\ge0}$ satisfies a Doeblin condition uniformly in $n$: there exist $\eta\in(0,1)$ and a probability measure $\nu$ on $\mathsf X\times\{0,1\}$ such that $\mathbb P(S_{t+1}\in A\mid S_t=s)\ge\eta\nu(A)$ for all measurable $A$ and all states $s$.
\end{enumerate}

Then $c\varepsilon_n\le\mathbb P(\theta_t=1\mid\mathcal F_{t-1})\le C\varepsilon_n$ almost surely for all $t\ge1$ and
\[
\alpha_{\bar{\mathcal G}}(k)
:=
\sup_{t\ge1}\alpha(\bar{\mathcal G}_t,\bar{\mathcal G}_{t+k}^{+})
\le 2 \cdot (1-\eta)^k,
\qquad k\ge1.
\]
\end{thm}

\begin{proof}[Proof of Theorem~\ref{thm:relaxation_implies_A31_A32}]
For the first claim, by the tower property,
\[
\mathbb P(\theta_t=1\mid\mathcal F_{t-1})
=
\mathbb E\!\left[
\mathbb P(\theta_t=1\mid\mathcal F_{t-1},X_{t-1},\theta_{t-1})
\mid\mathcal F_{t-1}
\right].
\]
Since $\theta_t$ is a function of $(X_{t-1},\theta_{t-1},U_t)$ and $U_t$ is independent of $\mathcal F_{t-1}\vee\mathcal H_{t-1}$, the inner term equals $\mathbb P(\theta_t=1\mid X_{t-1},\theta_{t-1})$, which lies in $[c\varepsilon_n,C\varepsilon_n]$ by condition (i).
The same bounds therefore hold after conditioning on $\mathcal F_{t-1}$.

For the mixing claim, let $P$ be the one-step transition kernel of $\{S_t\}$.
The Doeblin condition implies the total-variation contraction
\begin{equation}
\label{eq:Doeblin_contraction}
\sup_{s,s'}\|P^k(s,\cdot)-P^k(s',\cdot)\|_{\mathrm{TV}}
\le
2(1-\eta)^k,
\qquad k\ge1.
\end{equation}
We first use this contraction to control the mixing of the joint state filtration.
Fix $t\ge0$ and $k\ge1$.
For any $A\in\mathcal S_t=\sigma(S_0,\ldots,S_t)$ and $B\in\sigma(S_{t+k},S_{t+k+1},\ldots)$, the Markov property gives
$\mathbb P(B\mid\mathcal S_t)=\mathbb P(B\mid S_t)$.
Writing $\mu_t:=\mathcal L(S_t)$, we have $\mathbb P(B)=\int \mathbb P(B\mid S_t=s)\,\mu_t(\rd s).$
Therefore,
\[
\begin{aligned}
\big|\mathbb P(A\cap B)-\mathbb P(A)\mathbb P(B)\big|
&\le
\sup_s \big|\mathbb P(B\mid S_t=s)-\mathbb P(B)\big| \\
&\le
\sup_{s,s'}\big|\mathbb P(B\mid S_t=s)-\mathbb P(B\mid S_t=s')\big| \\
&\le
2(1-\eta)^k,
\end{aligned}
\]
where the last inequality follows from~\eqref{eq:Doeblin_contraction}.
Taking the supremum over $A$ and $B$ yields $\alpha_{\mathcal S}(k)\le2(1-\eta)^k$.
Since $\bar{\mathcal G}_t\subseteq\mathcal S_t$ and $\bar{\mathcal G}_{t+k}^{+}\subseteq\sigma(S_{t+k},S_{t+k+1},\ldots)$, monotonicity of strong-mixing coefficients gives $\alpha_{\bar{\mathcal G}}(k)\le2(1-\eta)^k$.
This proves the result.
\end{proof}

\begin{rem}[Exact vs.\ band survival probabilities]
\label{rem:exact_vs_band}
The band condition in Assumption~\ref{asmp:main}, namely $c\varepsilon_n\le\mathbb P(\theta_t=1\mid\mathcal F_{t-1})\le C\varepsilon_n$ almost surely, allows the survival probability to depend on the generated history through $\mathcal F_{t-1}$.
This is more flexible than the exact exogeneity condition $\mathbb P(\theta_t=1\mid\mathcal F_{t-1})=\varepsilon_n$, which corresponds to the special case $c=C=1$.
For the phase-transition results, replacing $\varepsilon_n$ by constant multiples only changes constants in thresholds and does not affect the exponents.
Thus the detection, discovery, and classification boundaries are unchanged under the exact and band variants.
\end{rem}

\begin{rem}[Run-length interpretation of Doeblin]
\label{rem:doeblin-runlength}
The Doeblin condition can be interpreted as a uniform reset or forgetting property.
More generally, suppose $S_t$ satisfies an $m$-step Doeblin minorization: there exist $m\in\mathbb N$, $\eta\in(0,1]$, and a probability measure $\nu$ such that $P^m(s,\cdot)\ge\eta\nu(\cdot)$ for all states $s$.
Then, regardless of the current context and edit status, within $m$ steps the joint process has probability at least $\eta$ of behaving as if it were drawn from $\nu$.

This condition also controls the length of consecutive surviving-watermark tokens.
Let $A_0:=\{(x,0):x\in\mathsf X\}$ and $A_1:=\{(x,1):x\in\mathsf X\}$.
Starting from $S_{t_0}\in A_1$, define $L:=\inf\{\ell\ge1:S_{t_0+\ell}\in A_0\}$.
If $\nu(A_0)=:p_0>0$, then $\mathbb P_s(S_m\in A_0)\ge\eta p_0=:\delta>0$ for every state $s$.
Applying the Markov property at times $m,2m,\ldots$ gives $\mathbb P(L>km)\le(1-\delta)^k$ for all $k\ge0$, and hence $\mathbb P(L>\ell)\le(1-\delta)^{\lfloor\ell/m\rfloor}\le\exp(-\delta\ell/m)$ for all $\ell\ge0$.
Thus $L$ has exponential tails and finite moments depending only on $(m,\eta,p_0)$.

This does not mean that the longest observed run over $n$ positions is uniformly bounded.
With exponential tails, the longest run among $n$ positions is typically of order $\log n$.
Thus the Doeblin condition permits occasional longer bursts, but rules out heavy-tailed or polynomially long persistent regimes.
\end{rem}

\input{8-proof}

\section{Additional Details for the Simulation Studies}
\label{sec:additional-simulation}

This section provides additional details for the simulation studies in Section~\ref{sec:simulation}.
We describe how the NTP distributions, survival indicators, and pivotal statistics are generated, and explain why the resulting process matches the mixture model used in our theory.

\paragraph{Generation of the NTP process.}
For each parameter tuple $(n,p,q,\alpha)$, we set $\varepsilon_n=0.5n^{-p}$ and $\Delta_n=n^{-q}$, as in the main text.
We consider $\alpha\in\{0,0.25,0.5,0.75\}$ and set $|\mathcal W_n|=2+\lfloor n^\alpha\rfloor$, with one dominant token, one core token, and $\lfloor n^\alpha\rfloor$ light tokens.

To introduce mild temporal variation in the NTP distribution, we use a latent regime process $Z_t\in\{0,1\}$.
Specifically, $Z_t$ is a two-state time-homogeneous Markov chain initialized at stationarity, with $\mathbb P(Z_1=1)=1/2$ and $\mathbb P(Z_t=Z_{t-1}\mid Z_{t-1})=\rho_Z$, where $\rho_Z=0.95$.
Thus, each regime, corresponding to a constant value of $Z_t$, tends to persist for many consecutive positions.

For each regime $z\in\{0,1\}$, we predefine an NTP distribution $\bP_n^{(z)}$ on this decomposition
\[
\mathcal W_n=\{w_{n,z}^\star\}\cup C_{n,z}\cup L_{n,z},
\qquad
|C_{n,z}|=1,
\qquad
|L_{n,z}|=\lfloor n^\alpha\rfloor .
\]
This corresponds to the case $r=0$ in Assumption~\ref{asmp:HCL+}.
The probability mass is assigned as
\[
P_n^{(z)}(w_{n,z}^\star)
=
1-\Delta_{n,z}^{\mathrm{core}}-\Delta_{n,z}^{\mathrm{light}},
\]
\[
P_n^{(z)}(w)
=
\Delta_{n,z}^{\mathrm{core}}
\quad (w\in C_{n,z}),
\qquad
P_n^{(z)}(w)
=
\frac{\Delta_{n,z}^{\mathrm{light}}}{|L_{n,z}|}
\quad (w\in L_{n,z}).
\]
Here $\Delta_{n,z}^{\mathrm{core}}$ and $\Delta_{n,z}^{\mathrm{light}}$ are regime-dependent constant-factor perturbations of the dominant--core--light scales, followed by clipping and renormalization so that $\Delta_{n,z}^{\mathrm{core}}+\Delta_{n,z}^{\mathrm{light}}\le \Delta_n$ and $\Delta_{n,z}^{\mathrm{core}}+\Delta_{n,z}^{\mathrm{light}}\asymp \Delta_n$.
Therefore, the two regimes may differ in both the token decomposition and the constant factors in the probability masses, while preserving the same asymptotic exponents.
We then set $\bP_t=\bP_n^{(Z_t)}$.
This construction gives a temporally dependent and heterogeneous NTP process that remains within the dominant--core--light regime assumed in the theory.

\paragraph{Generation of the survival indicators.}
Independently of $\{Z_t\}_{t=1}^n$, we generate the survival indicators $\theta_t\in\{0,1\}$ from a two-state Markov chain.
The indicator $\theta_t=1$ means that the watermark signal survives at position $t$, while $\theta_t=0$ means that the position follows the null law.
The chain is initialized at stationarity with $\mathbb P(\theta_t=1)=\varepsilon_n$.
In the implementation, its transition matrix is
\[
Q_{\theta,n}
=
\begin{pmatrix}
1-a_n & a_n\\
b_n & 1-b_n
\end{pmatrix},
\qquad
a_n=0.6\varepsilon_n,
\qquad
b_n=0.6(1-\varepsilon_n).
\]
Since $a_n/(a_n+b_n)=\varepsilon_n$, the stationary fraction of watermark-preserving positions is exactly $\varepsilon_n$.
Moreover, $b_n\asymp1$, so the surviving positions tend to form short bursts rather than long contiguous blocks.
This captures the mixed-source structure induced by local edits while keeping the survival process independent of the NTP driver.

\begin{rem}
This construction is a special case of the Markov framework in Supplementary Material~\ref{sec:Markov}.
To see the connection, one may take the latent state to be $X_t=Z_t$, since $\bP_t$ is a deterministic function of $Z_t$.
The survival process is exogenous because it is generated independently of the NTP process, with stationary mean $\varepsilon_n$.
\end{rem}

\paragraph{Generation of the pivotal statistics.}
Given $(\bP_t,\theta_t)$, we generate the pivotal statistic $Y_t$ according to the mixture model in~\eqref{eq:mixture-Y}.
If $\theta_t=0$, we sample $Y_t\sim\mu_0$, independently conditional on the latent process.
For the Gumbel-max watermark, $\mu_0$ is the uniform distribution on $[0,1]$.

If $\theta_t=1$, we sample from the watermark-induced alternative law $\mu_{1,\bP_t}$.
Instead of explicitly generating the full vocabulary-sized pseudorandom vector, it suffices to sample the scalar probability level of the selected token.
Specifically, we draw $w_t\sim\bP_t = (P_{t,w})_{w \in \Voca}$ and an independent $U_t\sim\mathrm{Unif}(0,1)$, and set $Y_t=U_t^{P_{t, w_t}}$.
This construction has exactly the desired alternative pivot distribution.
Indeed, for any $r\in[0,1]$,
\[
\mathbb P(Y_t\le r\mid \bP_t,\theta_t=1)
=
\sum_{w\in\mathcal W_n}
P_{t, w}\,
\mathbb P\!\left(U_t^{P_{t, w}}\le r\right)
=
\sum_{w\in\mathcal W_n}
P_{t, w} r^{1/P_{t, w}}
=
F_{1,\bP_t}(r),
\]
which is the alternative pivot law in Lemma~\ref{lem:alt-pivot-density}.
Equivalently, the conditional density is $f_{1,\bP_t}$ in~\eqref{eq:alt-pivot-density}.
Therefore, conditional on $(\bP_t,\theta_t)$, we have $Y_t\sim(1-\theta_t)\mu_0+\theta_t\mu_{1,\bP_t}$, which matches the mixture specification used in the theory.
For pure-null simulations, we set $\theta_t\equiv0$, so that $Y_t\stackrel{\mathrm{i.i.d.}}{\sim}\mu_0$.

This equivalence is the reason that the simulation code can work directly with pivotal statistics and scalar probability levels.
Algorithm~\ref{alg:adaptive-discovery-u} only uses the data through $\{Y_t\}_{t=1}^n$, and the construction above is distributionally equivalent to simulating the corresponding Gumbel-max watermark pivotal statistics.

\paragraph{Heatmap implementation and calibration.}
The heatmaps in Figure~\ref{fig:discovery-heat-panels7} use the construction above and the $(p,q)$ grid specified in Section~\ref{sec:sim-boundary}. All three methods are evaluated on the same $500$ simulated sequences at each grid point, with parameters selected over the grids described below.

For \Algo, we estimate the surviving watermark fraction using the tail-count estimator at $\tau^\dagger=0.85$,
\[
\widehat\varepsilon_n=\left[\frac{n^{-1}\sum_{t=1}^n\mathbf 1\{Y_t>\tau^\dagger\}-(1-\tau^\dagger)}{\tau^\dagger}\right]_0^1,
\]
where $[x]_0^1=\min\{1,\max\{0,x\}\}$. This differs from the optimal-weight fraction estimator used in the real-LLM experiments.
Our method \Algo\ scans $M_n=\max\{20,\lceil3.5(\log n)^{1.5}\rceil\}$ equally spaced $u$ values in $[0.005,0.98]$, with $\eta_n=0$. If no scanned threshold meets the calibration criterion, \Algo\ uses the largest scanned $u$.
We tune the calibration constant $C$ for \Algo\ over $60$ log-spaced values from $0.01$ to $10$, and the decision threshold for \texttt{AOL} over $61$ log-spaced values from $1$ to $16$.

For \texttt{WISER}, we tune four parameters. The parameter $\alpha_{\mathrm W}$ determines the null quantile used to screen candidate blocks, and $\beta$ sets the block size $b=\max\{1,\lceil\beta\sqrt n\rceil\}$. The parameter $c$ sets the minimum number of consecutive candidate blocks retained as a segment, while $\rho$ weights the estimated signal strength when refining segment endpoints. Here $\alpha_{\mathrm W}$ is distinct from the vocabulary exponent $\alpha$.
We evaluate all $120$ combinations of $\alpha_{\mathrm W}\in\{0.01,0.10,0.50,0.90,0.99\}$, $\beta\in\{0.01,0.05,0.25,1\}$, $c\in\{1,2\}$, and $\rho\in\{0.50,0.90,0.99\}$.
The endpoint-search constants in \texttt{WISER} are fixed at $C=\gamma=0.1$, and its screening thresholds are calibrated using $1000$ null simulations with seed $1234$. This $C$ is separate from the calibration constant in \Algo.

For all three methods, we compute the discovery error for each parameter configuration as total false discoveries divided by total discoveries (with denominator at least one), plus the fraction of the $500$ replications with no discoveries. We report the smallest error over each method's parameter grid, resolving ties by the first grid index. Clipping at $1.05$ is applied only for visualization.

\paragraph{Additional phase-transition slices.}
Figure~\ref{fig:discovery-boundary-panels000025} reports additional one-dimensional phase-transition slices for $\alpha\in\{0,0.25\}$.
These panels complement Figure~\ref{fig:discovery-boundary-panels7} in the main text, which reports the corresponding slices for $\alpha\in\{0.5,0.75\}$.
As in the main simulations, each point is obtained by averaging over $200$ independent Monte Carlo trials, and the displayed value is the smallest discovery error in~\eqref{eq:discovery-err} over the calibration grid.
The vertical dashed lines mark the theoretical boundary from Theorem~\ref{thm:discovery-boundary}.
For $\alpha=0$, the discovery error remains large across the displayed ranges, consistent with the impossibility of discovery in the fixed-vocabulary regime.
For $\alpha=0.25$, the empirical transition becomes visible and aligns with the predicted boundary, with sharper transitions as $n$ increases.

\begin{figure}[t!]
\centering
\includegraphics[width=\textwidth]{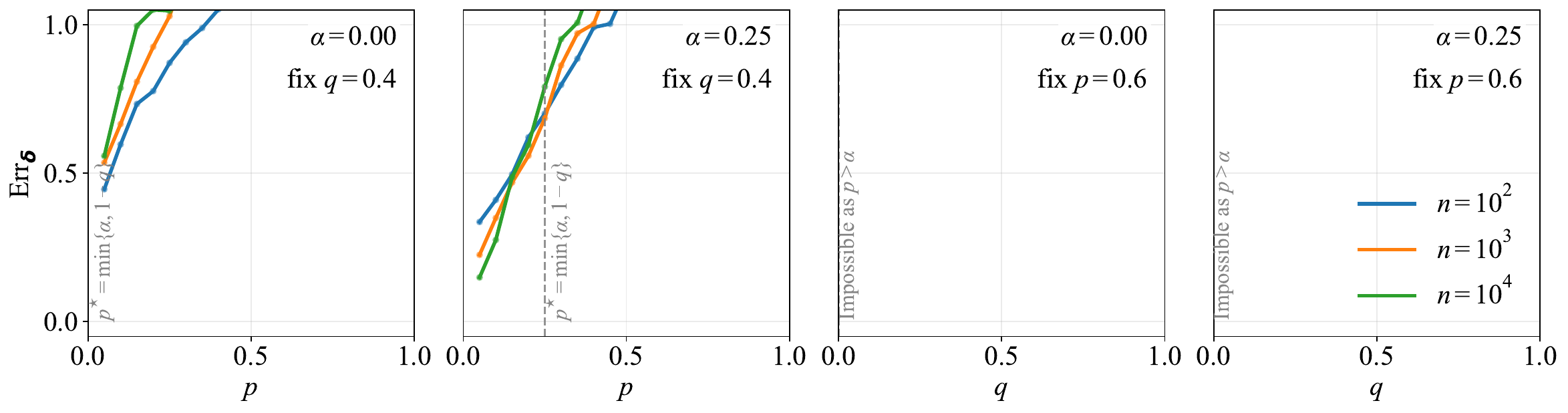}
\vspace{-15pt}
\caption{
Empirical discovery phase-transition slices for \Algo\ with $\alpha\in\{0,0.25\}$.
Each panel reports the smallest discovery error in~\eqref{eq:discovery-err} over a calibration grid along one slice, averaged over $200$ trials.
The vertical dashed lines show the theoretical discovery boundary.
}
\label{fig:discovery-boundary-panels000025}
\end{figure}

\section{Details and Additional Results of Language Model Experiments}
\label{sec:LLM}

This section provides additional details for the real-LLM experiments in Section~\ref{sec:LLM-experiments}.
We describe the watermark generation procedure, prompt construction, edit mechanisms, empirical FPR control, implementation details for the compared methods, and additional results under FPR levels not shown in the main text.

\subsection{Additional Experimental Details}
\label{sec:LLM-details}

\paragraph{Watermark generation.}
We use a context window of size $m=3$.
At position $t$, the pseudorandom variable is computed from the previous $m$ tokens as $\xi_t=\AM(s_{(t-m):(t-1)},\Key)$, where $\AM$ is the hash function used in~\citep{zhao2024permute}.
To reduce repetitive outputs, we apply repeated-context masking: the watermark is added only when the length-$m$ context has not appeared earlier in the generated history, following~\citep{hu2023unbiased,wu2023dipmark,dathathri2024scalable}.
This is the 1-sequence repeated-context masking strategy of~\citet{dathathri2024scalable}. We use the same ordered-context hash scheme and key during generation and verification. The pre-edit data retain the watermark-activation indicators $s_i$ used to construct the ground-truth labels.
If the current context is masked, the token is sampled directly from the temperature-scaled NTP distribution using multinomial sampling.

\paragraph{Prompt construction and post-edit processing.}
We use the \texttt{realnewslike} split of the C4 dataset and follow the setup in Section~C.1 of the supplementary material accompanying~\citet{li2025robust}.
Each C4 example is tokenized with a truncation length $2048-20=2028$, and we keep the first $1000$ examples whose truncated token sequence has length at least $50+400=450$.
For each retained example, we use the last $50$ tokens as the prompt and ask OPT-1.3B to generate a watermarked continuation of length $n=400$.

After applying a random or adversarial edit, we decode the edited token sequence into text and then re-tokenize it. Block permutation directly reorders the original tokens without re-tokenization.
We then pad or truncate the resulting sequence to the target length used by the verifier: $400$ tokens for substitution, insertion, and adversarial edits, and $200$ for deletion. Shorter sequences are left-padded. Block permutation retains all $400$ original tokens.
The ground-truth labeling rules are given under Ground-truth correspondence below.

We use fixed random seeds for reproducibility. After editing, we recompute the pivotal statistics using the same watermark key and context width as during generation. All methods are evaluated on the same edited texts and ground-truth labels.

\paragraph{Random edits.}
For each edit rate $r\in\{0.10,0.15,\ldots,0.40\}$, we select $\lfloor400r\rfloor$ positions in the original continuation. We randomly select the editing positions so that edits are spread throughout the text. At each selected position, we replace the token, insert a new token immediately before it, or delete it, depending on the edit type. Replacement and inserted tokens are sampled uniformly from the vocabulary. The edit rate is measured before re-tokenization.

\paragraph{Adversarial edits.}
The edit budget $K\in\{40,45,50,55,60\}$ specifies the total number of token replacements. To spread these replacements across the text, we split the continuation from left to right into $K$ parts of approximately equal length and replace one token in each part. We randomly sample four distinct positions within each part, select the token with the largest pre-edit pivotal value, and replace it with a token sampled uniformly from the vocabulary. This targets relatively strong watermark evidence without using ground-truth labels or localization results.

\paragraph{Block permutation.}
We split the continuation into blocks of $b\in\{4,6,8\}$ consecutive tokens and randomly reorder them. Before reordering, we randomly choose approximately $75\%$ of the originally adjacent block pairs to separate. The reordering ensures that these pairs no longer appear next to each other in their original order, while the remaining pairs stay together. The tokens and their order within each block are preserved, and any shorter final block remains at the end. This simulates rearranging text segments while keeping their internal content unchanged.

\phantomsection
\label{par:LLM-ground-truth}
\paragraph{Ground-truth correspondence.}
For each position in the final post-edit sequence, we examine the block consisting of the current token and its previous $m$ tokens.
For random and adversarial edits, we label the current position as watermark-preserving only if the entire $(m+1)$-token block is aligned with an identical contiguous block in the original pre-edit sequence. Because repeated-context masking skips watermarking at some positions, we additionally require that the watermark was applied at every position in that original block.
This definition reflects how the verifier reconstructs pseudorandomness, since the pseudorandom variable at each position is determined by the previous $m$ tokens.

For random and adversarial edits, we obtain the correspondence by matching identical token sequences in the original and edited text. Unmatched tokens, padding, and positions without a complete $(m+1)$-token window receive label zero.
For block permutation, the reordering directly tells us the original position of each token. We compare the preceding $m$ tokens at its original and final positions, requiring the same token values in the same order. The original $(m+1)$-token block must also have been fully watermarked, and complete context windows must exist in both sequences. Thus, a relocated token can retain a positive label when its watermarking context is preserved.

\subsection{Additional Details for Compared Methods}
\label{sec:LLM-methods}

\paragraph{Plug-in fraction estimator.}
For \Algo{}\texttt{-plugin}, we use the optimal fraction estimator of~\citet{li2025optimal} to estimate the surviving watermark fraction $\eps_n$ from the pivotal statistics of the observed text.
The estimator is based on the mixture formulation that the observed pivotal statistic follows $(1-\eps_n)F_0+\eps_n F_1$, where $F_0$ is the known null law and $F_1$ is an unknown population-level alternative law associated with surviving watermark signals.
It estimates $F_1$ using an auxiliary sample of pivotal statistics from fully watermarked text, yielding an empirical alternative law $\widehat F_1$ and density ratio $\widehat r=\rd\widehat F_1/\rd F_0$.
It then uses the variance-reducing weight $\omega_{\varepsilon}(y) = \frac{1-\widehat r(y)}{(1-\varepsilon)+\varepsilon \widehat r(y)}$, and computes $\widehat\varepsilon_n$ by numerically solving the estimating equation specified in~\citet{li2025optimal}.

In our implementation, the reference sample consists of pre-edit pivots at watermark-active positions from the same generated data collection. At most $50000$ pivots are sampled to construct a $500$-bin histogram estimate of the alternative density. This reference is shared across documents at each temperature.

\paragraph{\texttt{AOL} and calibration.}
As a localization baseline, we use the adaptive online locator \texttt{AOL} of~\citet{zhao2025efficiently}.
In our notation, let $Y_t\in[0,1]$ denote the per-token pivotal statistic.
For the Gumbel watermark, \texttt{AOL} uses the transformed score $s_t=-\log(1-Y_t)$.
Given the score sequence $s_1,\ldots,s_n$, \texttt{AOL} applies the Aligator online smoothing algorithm with multiple random circular starting points, producing $M$ fitted values $\widehat{\theta}_t^{(1)},\ldots,\widehat{\theta}_t^{(M)}$ for each token position.
These fitted values are averaged to form the final localization score,
$\widehat{\theta}_t=M^{-1}\sum_{i=1}^M\widehat{\theta}_t^{(i)}$.
Token $t$ is then declared watermark-preserving whenever $\widehat{\theta}_t>\zeta$, where $\zeta$ is the calibration threshold for \texttt{AOL}, analogous to the calibration constant $C$ in our method.

The default threshold in the official implementation of \texttt{AOL} can yield high IoU but also a high FPR.
For a fair comparison under fixed empirical FPR constraints, we therefore sweep the threshold $\zeta$ over a predetermined grid.
Specifically, for \texttt{AOL}, we sweep $\zeta\in\{1.00,1.01,\ldots,4.00\}$.
For \Algo, we sweep the calibration constant over $\{0.01,0.02,\ldots,1.00\}\cup\{1.05,1.10,\ldots,5.50\}$.
For each method, we select the best IoU and TPR separately using the empirical-FPR bin rule described below.

\paragraph{\texttt{WISER}.}
We use the implementation of~\citet{bonnerjee2025fast} at \url{https://github.com/subroy13/WISER}. The method screens candidate blocks, merges adjacent candidates into segments, and refines their endpoints. The resulting intervals define the token-level decisions.
We use the same parameter definitions, endpoint-search constants, and null-calibration settings as in Supplementary Material~\ref{sec:additional-simulation}.

At each edit setting, we evaluate $1506$ distinct $(\alpha_{\mathrm W},\beta,c,\rho)$ configurations from four parameter grids. We list them below. Let $\mathcal P$ consist of the pairs $(\beta,1)$ with $\beta\in\{0.05,\allowbreak 0.10,\allowbreak 0.15,\allowbreak 0.25,\allowbreak 0.30,\allowbreak 0.40,\allowbreak 0.45,\allowbreak 0.50,\allowbreak 0.60\}$, together with $(0.10,2)$, and let $\mathcal A=\{0.05,\allowbreak 0.10,\allowbreak 0.20,\allowbreak 0.25,\allowbreak 0.30,\allowbreak 0.35,\allowbreak 0.55,\allowbreak 0.75,\allowbreak 0.90,\allowbreak 0.95\}$. The four grids are
\begin{enumerate}[leftmargin=*,nosep]
\item $(\beta,c)\in\mathcal P$, $\alpha_{\mathrm W}\in\mathcal A$, and $\rho\in\{0.50,\allowbreak 0.70,\allowbreak 0.80,\allowbreak 0.85,\allowbreak 0.90,\allowbreak 0.95\}$.
\item $(\beta,c)\in\mathcal P$, $\alpha_{\mathrm W}\in\{0.005,\allowbreak 0.01,\allowbreak 0.02,\allowbreak 0.03,\allowbreak 0.04\}$, and $\rho\in\{0.30,\allowbreak 0.40,\allowbreak 0.50,\allowbreak 0.60,\allowbreak 0.70,\allowbreak 0.80,\allowbreak 0.85,\allowbreak 0.90,\allowbreak 0.95\}$.
\item $\alpha_{\mathrm W}\in\mathcal A$, $c=1$, $\beta\in\{0.15,\allowbreak 0.25,\allowbreak 0.30,\allowbreak 0.40,\allowbreak 0.45,\allowbreak 0.50,\allowbreak 0.60\}$, and $\rho\in\{0.96,\allowbreak 0.97,\allowbreak 0.98,\allowbreak 0.99\}$.
\item $\beta=0.10$, $c\in\{2,\allowbreak 3,\allowbreak 4\}$, $\alpha_{\mathrm W}\in\{0.55,\allowbreak 0.75,\allowbreak 0.90,\allowbreak 0.95,\allowbreak 0.96,\allowbreak 0.97,\allowbreak 0.98,\allowbreak 0.99\}$, and $\rho\in\{0.80,\allowbreak 0.85,\allowbreak 0.90,\allowbreak 0.95,\allowbreak 0.96,\allowbreak 0.97,\allowbreak 0.98,\allowbreak 0.99\}$.
\end{enumerate}
Some parameter combinations occur in more than one of the four grids. Counting each distinct combination only once gives $1506$ configurations. We use the same combined grid across all real-LLM edit settings and temperatures. The grid sizes differ across methods, and the reported comparisons describe their empirical performance within these specified grids.

\paragraph{FPR control and averaging.}
To compare localization methods under the same false-positive constraint, we evaluate each method over a fixed grid of calibration parameters.
We consider several target FPR levels, from conservative to less conservative, to study how localization performance changes with the allowed false-positive rate.
For each parameter configuration and edit level, we first compute FPR, IoU, and TPR for each document and then average each metric across documents. For TPR, this average uses only documents with at least one true watermark-preserving position.
To compare methods near the same target FPR $\lambda$, we first consider configurations whose average FPR lies in $(\lambda-0.01,\lambda]$. If none is available, we move to the next lower interval of width $0.01$, continuing until a configuration is available. Within the selected interval, we report the largest average IoU and the largest average TPR, selecting parameters separately for the two metrics.
Finally, each entry in Tables~\ref{tab:real-llm-iou-lambda-005}, \ref{tab:real-llm-iou-tpr-lambda-004}, and~\ref{tab:real-llm-iou-tpr-lambda-007} averages these selected scores equally across the edit levels: seven rates for random edits, five budgets for adversarial edits, or three block sizes for block permutation. The FPR constraint therefore applies to the average across documents.

\paragraph{Runtime accounting.}
Table~\ref{tab:real-llm-runtime-T1} reports the average runtime per document at $T=1$, using the parameters selected for IoU at target FPR $0.05$. The time includes recomputing the pivotal statistics and applying the localization method on CPU, but excludes parameter search, data loading, and output writing.
For \Algo{}\texttt{-plugin}, localization includes fraction estimation and threshold selection, whereas \Algo{}\texttt{-oracle} uses the true surviving fraction. For \texttt{AOL}, it includes computing the localization scores and applying the decision threshold. For \texttt{WISER}, it includes score transformation, segment detection and endpoint refinement, and converting the detected intervals into token-level decisions.
Preparing the reference distribution for \Algo{}\texttt{-plugin} and calibrating the null thresholds for \texttt{WISER} are shared tasks performed once for the $1000$ documents. We divide each of these preparation times by $1000$ and add it to the corresponding method's average runtime.

\subsection{Additional Results}
\label{appen:new-LLM-results}

\paragraph{Accuracy of the fraction estimator.}
The preceding experiments show a recurring gap between \Algo{}\texttt{-oracle} and \Algo{}\texttt{-plugin}, especially at lower temperatures and under stronger text modifications.
To better understand this gap, we evaluate the accuracy of the fraction estimator $\widehat{\varepsilon}_n$ from~\citet{li2025optimal}.
Specifically, we report the relative root mean squared error
\begin{equation}
\label{eq:RMSE}
\mathrm{RelRMSE}
=\frac{\left[D^{-1}\sum_{d=1}^{D}(\widehat\varepsilon_d-\varepsilon_{d,\mathrm{true}})^2\right]^{1/2}}
{D^{-1}\sum_{d=1}^{D}\varepsilon_{d,\mathrm{true}}},\qquad D=1000.
\end{equation}
Here, the numerator is the root mean squared error across documents, and the denominator is their mean true surviving fraction. We compute this ratio separately at each edit level and then average over the same levels as in Table~\ref{tab:real-llm-iou-lambda-005}. This quantity measures the typical estimation error on the scale of the mean true surviving fraction.
For example, values around $0.25$, $0.5$, and $1$ correspond roughly to typical errors of $25\%$, $50\%$, and $100\%$ of the true surviving fraction, respectively.
Thus, values substantially below one indicate that the plug-in estimate is reasonably accurate, whereas values near or above one indicate that fraction estimation can become a practical bottleneck.

Table~\ref{tab:real-llm-relative-rmse} supports this interpretation.
At high temperature, the relative RMSE is small across the edit types, and \Algo{}\texttt{-plugin} closely tracks \Algo{}\texttt{-oracle}.
As the temperature decreases, the relative RMSE increases sharply, and the gap between the plug-in and oracle versions becomes larger.
The effect is particularly visible at $T=0.5$, where the relative RMSE is often close to or above one.
These results suggest that the low-temperature degradation of \Algo{}\texttt{-plugin} is partly driven by inaccurate estimation of the surviving watermark fraction.

\begin{table}[t!]
\centering
\caption{Relative RMSE of the plug-in surviving-fraction estimator, computed as in \eqref{eq:RMSE} at each edit level and averaged equally over the same levels as Table~\ref{tab:real-llm-iou-lambda-005}. Values are ratios, not percentages. Smaller is better.}
\label{tab:real-llm-relative-rmse}
\begin{tabular}{lrrr}
\toprule
Edit type & $T=1$ & $T=0.7$ & $T=0.5$ \\
\midrule
Random substitution & 0.3143 & 0.9078 & 2.6951 \\
Random insertion & 0.2537 & 0.6942 & 1.9082 \\
Random deletion & 0.2975 & 0.7691 & 2.7192 \\
Adversarial edits & 0.1681 & 0.3031 & 0.9916 \\
Block permutation & 0.1343 & 0.3691 & 1.3576 \\
\bottomrule
\end{tabular}
\end{table}

\paragraph{TPR versus IoU under fixed FPR control.}
Before presenting the additional real-LLM results, we briefly clarify the difference between the two localization metrics used in our evaluation.  Let $S^\star$ denote the set of true watermark-preserving locations and let $\widehat S$ denote the set selected by a method.  Throughout the experiments, we compare methods under a fixed empirical FPR constraint.  Under this controlled-error comparison, the empirical true-positive rate $\widehat{\mathrm{TPR}}=|\widehat S\cap S^\star|/|S^\star|$ is the most direct measure of localization power: it reports the fraction of true watermark-preserving locations recovered while keeping the false-positive level fixed.

The intersection-over-union score, $\mathrm{IoU}=|\widehat S\cap S^\star|/|\widehat S\cup S^\star|$, measures a stricter notion of set overlap.  It penalizes missed watermark-preserving locations through $S^\star\setminus\widehat S$ and extra selected locations through $\widehat S\setminus S^\star$.  This makes IoU a conservative summary of localization quality.  However, because our tables already enforce an empirical FPR constraint, IoU partially penalizes false positives a second time through the denominator $|\widehat S\cup S^\star|$.  It can also vary with the size of $S^\star$, which changes with the edit strength.  For this reason, we use TPR at fixed empirical FPR as the primary controlled-error metric, analogous to reporting power at a fixed size in classical testing, and report IoU as a complementary robustness metric.

\paragraph{IoU and TPR for different FPR levels.}
In the main text, due to space limitations, we report only the IoU and TPR curves at temperature $T=1$ under the target FPR level $0.05$. To better illustrate the dependence on the target FPR, we provide additional localization results under three random edit mechanisms at $T=1$ for target FPR levels $0.03$, $0.06$, and $0.09$. The IoU results are shown in Figure~\ref{fig:real-llm-random-edits-IoU-T1}, and the TPR results are shown in Figure~\ref{fig:real-llm-random-edits-TPR-T1}. We also report the corresponding $T=0.7$ results in Figures~\ref{fig:tpr-random-edits-t07-fprlevels} and~\ref{fig:iou-random-edits-t07-fprlevels}. These figures complement Figure~\ref{fig:real-llm-random-edits-T1} in the main text. The qualitative pattern remains consistent: \Algo{}\texttt{-oracle} and \Algo{}\texttt{-plugin} generally outperform or remain competitive with \texttt{AOL} and \texttt{WISER} under random substitution, insertion, and deletion across different FPR levels. At $T=0.7$, the advantage is still visible but smaller, and the absolute localization performance is lower than at $T=1$, consistent with the lower-temperature setting providing weaker usable watermark evidence and making localization more difficult.

\begin{figure}[p]
\centering
\includegraphics[width=\textwidth]{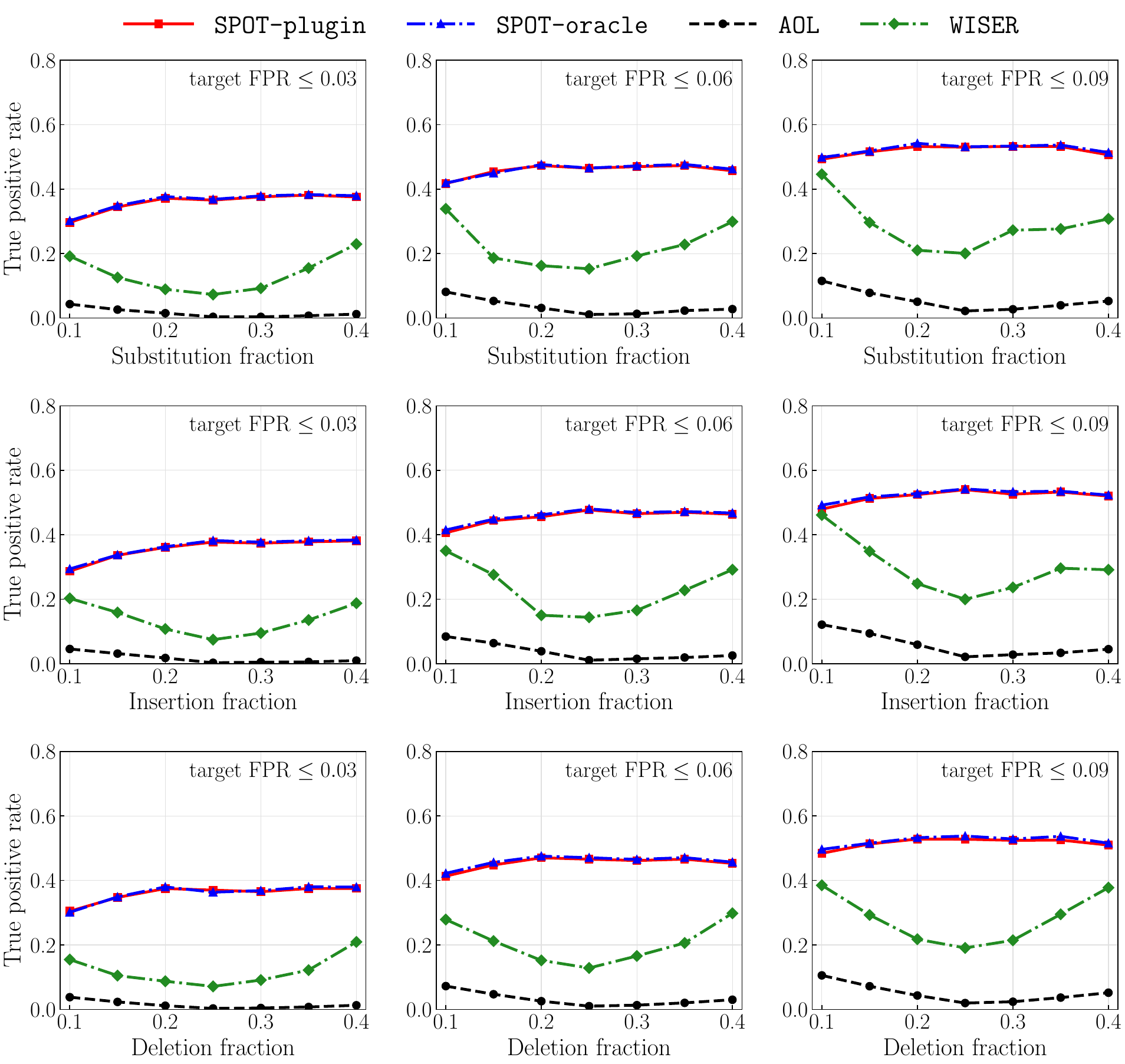}
\caption{Localization performance under three random edit mechanisms at temperature $T=1$, measured by TPR. Columns correspond to target FPR levels $0.03$, $0.06$, and $0.09$ from left to right. Rows correspond to substitution, insertion, and deletion edits from top to bottom. Each panel reports the mean TPR as a function of the edit fraction.}
\label{fig:real-llm-random-edits-TPR-T1}
\end{figure}

\begin{figure}[p]
\centering
\includegraphics[width=\textwidth]{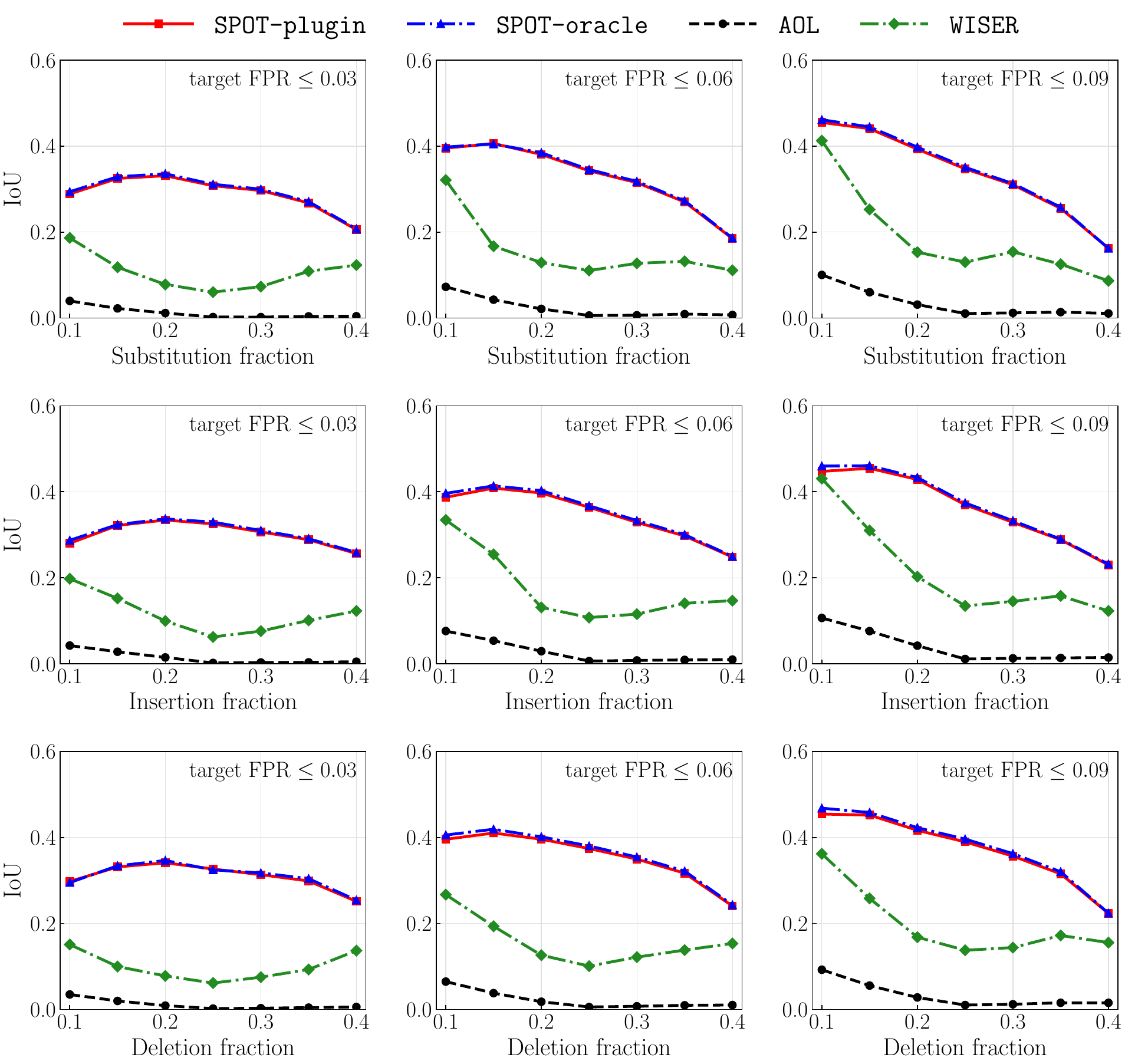}
\caption{Localization performance under three random edit mechanisms at temperature $T=1$, measured by IoU. Columns correspond to target FPR levels $0.03$, $0.06$, and $0.09$ from left to right. Rows correspond to substitution, insertion, and deletion edits from top to bottom. Each panel reports the mean IoU as a function of the edit fraction.}
\label{fig:real-llm-random-edits-IoU-T1}
\end{figure}

\begin{figure}[p]
\centering
\includegraphics[width=\textwidth]{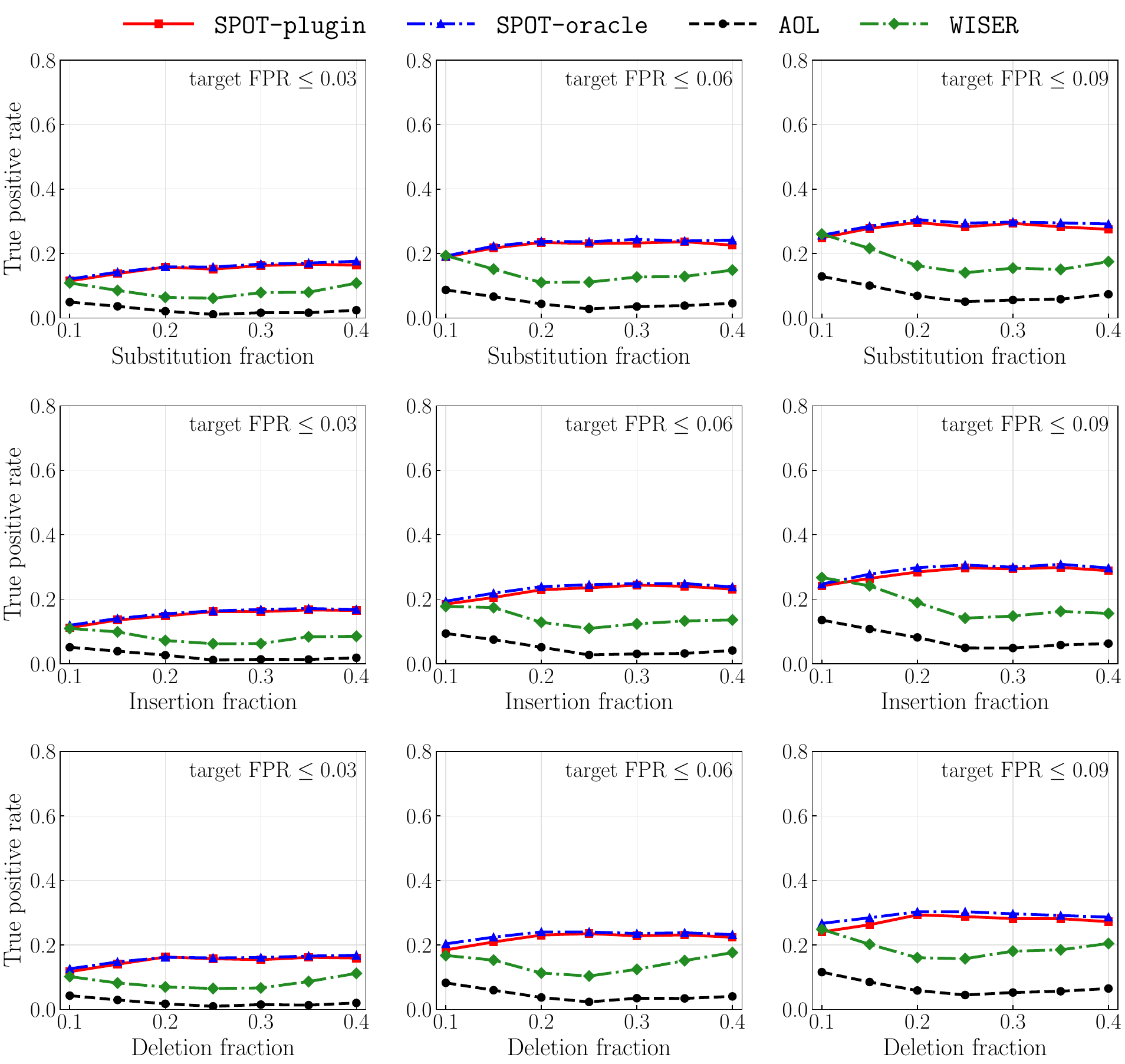}
\caption{Localization performance under three random edit mechanisms at temperature $T=0.7$, measured by TPR. Columns correspond to target FPR levels $0.03$, $0.06$, and $0.09$ from left to right. Rows correspond to substitution, insertion, and deletion edits from top to bottom. Each panel reports the mean TPR as a function of the edit fraction.}
\label{fig:tpr-random-edits-t07-fprlevels}
\end{figure}

\begin{figure}[p]
\centering
\includegraphics[width=\textwidth]{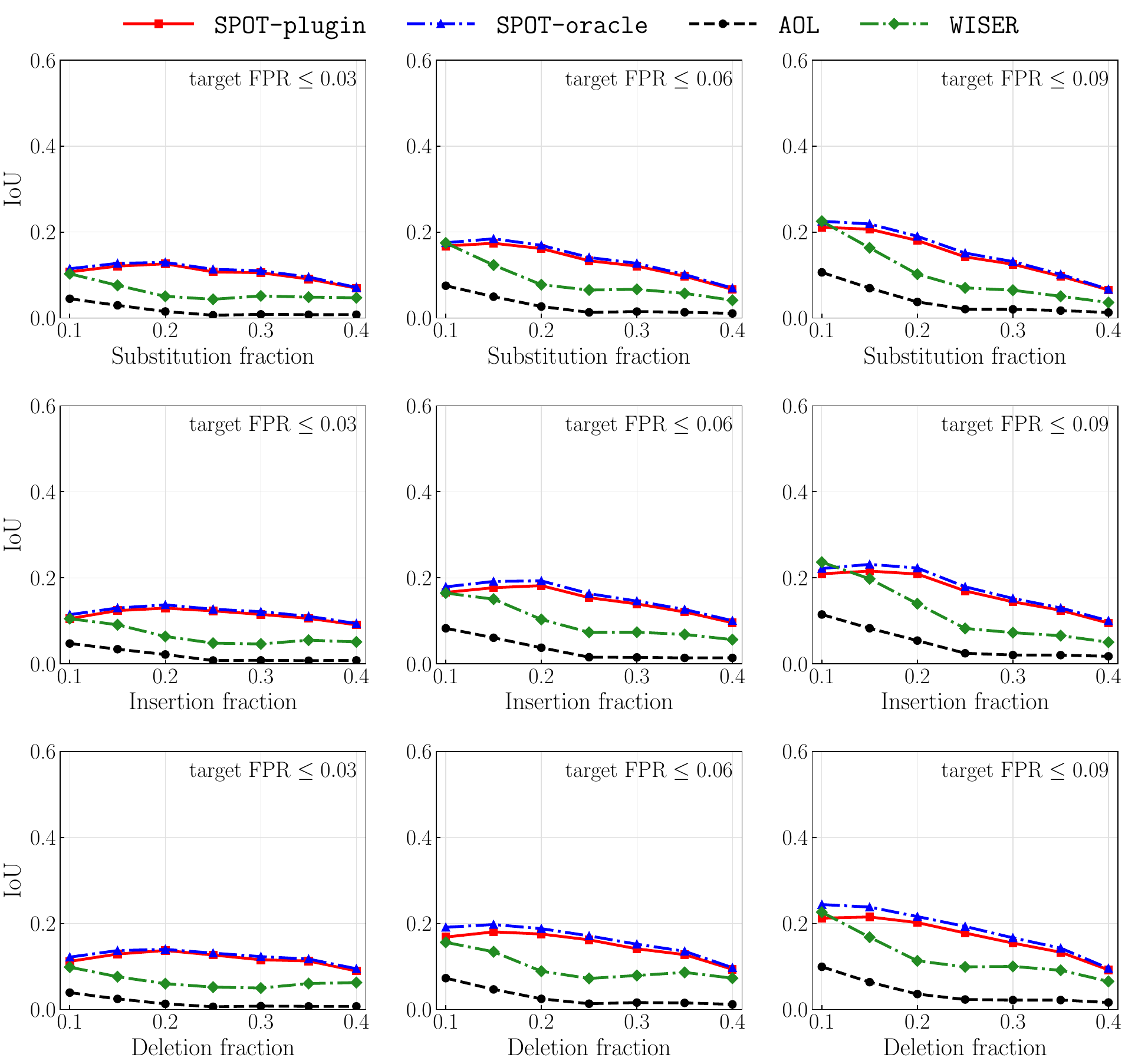}
\caption{Localization performance under three random edit mechanisms at temperature $T=0.7$, measured by IoU. Columns correspond to target FPR levels $0.03$, $0.06$, and $0.09$ from left to right. Rows correspond to substitution, insertion, and deletion edits from top to bottom. Each panel reports the mean IoU as a function of the edit fraction.}
\label{fig:iou-random-edits-t07-fprlevels}
\end{figure}

\paragraph{IoU and TPR across temperatures.}
Tables~\ref{tab:real-llm-iou-tpr-lambda-004} and~\ref{tab:real-llm-iou-tpr-lambda-007} report additional IoU and TPR results across generation temperatures at empirical FPR targets $\lambda\in\{0.04,0.07\}$, complementing Table~\ref{tab:real-llm-iou-lambda-005} in the main text. Performance generally improves as temperature increases for \Algo{} and \texttt{WISER}, whereas \texttt{AOL} shows no consistent trend. The relative performance of the methods is also consistent with the main-text results: \Algo{}\texttt{-oracle} leads the random-edit averages at all three temperatures, while \texttt{WISER} performs better under low-temperature adversarial edits and block permutation. The plug-in permutation results at $T=0.5$ include lower-bin selections at target $0.04$, as specified in the table note.

\begingroup
\makeatletter
\setlength{\@fptop}{0pt}
\makeatother
\begin{table}[t!]
\centering\small
\setlength{\tabcolsep}{4pt}
\caption{Average IoU and TPR at empirical FPR target $0.04$ for three generation temperatures. Each metric is selected separately in the highest occupied $0.01$-wide FPR bin at or below the target. Entries equally average the seven random-edit rates $0.10,0.15,\ldots,0.40$, adversarial budgets $K\in\{40,45,50,55,60\}$, or permutation block sizes $b\in\{4,6,8\}$. Bold marks the largest value among the four methods for each edit type, temperature, and metric.}
\label{tab:real-llm-iou-tpr-lambda-004}
\resizebox{\textwidth}{!}{%
\begin{tabular}{llrrrrrr}
\toprule
\multirow{2}{*}{Edit type} & \multirow{2}{*}{Method} & \multicolumn{3}{c}{IoU} & \multicolumn{3}{c}{TPR} \\
\cmidrule(lr){3-5}\cmidrule(lr){6-8}
 & & $T=1$ & $T=0.7$ & $T=0.5$ & $T=1$ & $T=0.7$ & $T=0.5$ \\
\midrule
\multirow{4}{*}{Random substitution} & \texttt{AOL} & 0.0160 & 0.0216 & 0.0170 & 0.0215 & 0.0336 & 0.0347 \\
 & \texttt{WISER} & 0.1243 & 0.0707 & 0.0435 & 0.1679 & 0.1042 & 0.0800 \\
 & \Algo{}\texttt{-plugin} & 0.3083 & 0.1161 & 0.0455 & 0.4004 & 0.1787 & 0.1018 \\
 & \Algo{}\texttt{-oracle} & \textbf{0.3121} & \textbf{0.1217} & \textbf{0.0570} & \textbf{0.4051} & \textbf{0.1840} & \textbf{0.1154} \\
\midrule
\multirow{4}{*}{Random insertion} & \texttt{AOL} & 0.0184 & 0.0246 & 0.0206 & 0.0231 & 0.0337 & 0.0364 \\
 & \texttt{WISER} & 0.1330 & 0.0777 & 0.0494 & 0.1638 & 0.1044 & 0.0784 \\
 & \Algo{}\texttt{-plugin} & 0.3230 & 0.1277 & 0.0535 & 0.3961 & 0.1787 & 0.1045 \\
 & \Algo{}\texttt{-oracle} & \textbf{0.3264} & \textbf{0.1339} & \textbf{0.0681} & \textbf{0.3980} & \textbf{0.1837} & \textbf{0.1174} \\
\midrule
\multirow{4}{*}{Random deletion} & \texttt{AOL} & 0.0150 & 0.0195 & 0.0171 & 0.0201 & 0.0284 & 0.0312 \\
 & \texttt{WISER} & 0.1247 & 0.0762 & 0.0512 & 0.1555 & 0.0994 & 0.0794 \\
 & \Algo{}\texttt{-plugin} & 0.3291 & 0.1301 & 0.0527 & 0.3962 & 0.1760 & 0.0880 \\
 & \Algo{}\texttt{-oracle} & \textbf{0.3353} & \textbf{0.1379} & \textbf{0.0661} & \textbf{0.4022} & \textbf{0.1819} & \textbf{0.1012} \\
\midrule
\multirow{4}{*}{Adversarial edits} & \texttt{AOL} & 0.0637 & 0.0614 & 0.0440 & 0.0711 & 0.0703 & 0.0588 \\
 & \texttt{WISER} & 0.2452 & 0.1161 & \textbf{0.0723} & 0.2579 & 0.1271 & \textbf{0.0883} \\
 & \Algo{}\texttt{-plugin} & 0.2972 & 0.1095 & 0.0350 & 0.3135 & 0.1228 & 0.0611 \\
 & \Algo{}\texttt{-oracle} & \textbf{0.3011} & \textbf{0.1164} & 0.0569 & \textbf{0.3167} & \textbf{0.1276} & 0.0730 \\
\midrule
\multirow{4}{*}{Block permutation} & \texttt{AOL} & 0.0690 & 0.0635 & 0.0405 & 0.0751 & 0.0708 & 0.0549 \\
 & \texttt{WISER} & 0.2097 & 0.1163 & \textbf{0.0747} & 0.2183 & 0.1251 & \textbf{0.0881} \\
 & \Algo{}\texttt{-plugin} & 0.3126 & 0.1229 & 0.0339$^{\dagger}$ & 0.3253 & 0.1345 & 0.0443$^{\dagger}$ \\
 & \Algo{}\texttt{-oracle} & \textbf{0.3262} & \textbf{0.1369} & 0.0678 & \textbf{0.3390} & \textbf{0.1465} & 0.0813 \\
\bottomrule
\end{tabular}}
\par\smallskip
\begin{minipage}{\textwidth}\footnotesize
$^{\dagger}$At $T=0.5$, the permutation plug-in average includes lower-bin selections for $b=6,8$, with achieved FPRs $0.01308$ and $0.01503$. For both separately selected metrics, these points occupy $(0.01,0.02]$ and satisfy the cap without reaching the intended top bin.
\end{minipage}
\end{table}

\begin{table}[t!]
\centering\small
\setlength{\tabcolsep}{4pt}
\caption{Average IoU and TPR at empirical FPR target $0.07$ for three generation temperatures. Each metric is selected separately in the highest occupied $0.01$-wide FPR bin at or below the target. Entries equally average the seven random-edit rates $0.10,0.15,\ldots,0.40$, adversarial budgets $K\in\{40,45,50,55,60\}$, or permutation block sizes $b\in\{4,6,8\}$. Bold marks the largest value among the four methods for each edit type, temperature, and metric.}
\label{tab:real-llm-iou-tpr-lambda-007}
\resizebox{\textwidth}{!}{%
\begin{tabular}{llrrrrrr}
\toprule
\multirow{2}{*}{Edit type} & \multirow{2}{*}{Method} & \multicolumn{3}{c}{IoU} & \multicolumn{3}{c}{TPR} \\
\cmidrule(lr){3-5}\cmidrule(lr){6-8}
 & & $T=1$ & $T=0.7$ & $T=0.5$ & $T=1$ & $T=0.7$ & $T=0.5$ \\
\midrule
\multirow{4}{*}{Random substitution} & \texttt{AOL} & 0.0269 & 0.0338 & 0.0271 & 0.0405 & 0.0599 & 0.0644 \\
 & \texttt{WISER} & 0.1745 & 0.0900 & 0.0528 & 0.2544 & 0.1501 & 0.1178 \\
 & \Algo{}\texttt{-plugin} & 0.3330 & 0.1375 & 0.0543 & 0.4827 & 0.2435 & 0.1493 \\
 & \Algo{}\texttt{-oracle} & \textbf{0.3368} & \textbf{0.1456} & \textbf{0.0683} & \textbf{0.4862} & \textbf{0.2528} & \textbf{0.1681} \\
\midrule
\multirow{4}{*}{Random insertion} & \texttt{AOL} & 0.0313 & 0.0396 & 0.0324 & 0.0435 & 0.0610 & 0.0648 \\
 & \texttt{WISER} & 0.1963 & 0.1072 & 0.0650 & 0.2589 & 0.1573 & 0.1199 \\
 & \Algo{}\texttt{-plugin} & 0.3561 & 0.1552 & 0.0657 & 0.4810 & 0.2445 & 0.1509 \\
 & \Algo{}\texttt{-oracle} & \textbf{0.3605} & \textbf{0.1645} & \textbf{0.0825} & \textbf{0.4853} & \textbf{0.2542} & \textbf{0.1680} \\
\midrule
\multirow{4}{*}{Random deletion} & \texttt{AOL} & 0.0256 & 0.0326 & 0.0286 & 0.0376 & 0.0524 & 0.0590 \\
 & \texttt{WISER} & 0.1674 & 0.1081 & 0.0715 & 0.2277 & 0.1584 & 0.1259 \\
 & \Algo{}\texttt{-plugin} & 0.3625 & 0.1578 & 0.0681 & 0.4766 & 0.2407 & 0.1389 \\
 & \Algo{}\texttt{-oracle} & \textbf{0.3693} & \textbf{0.1717} & \textbf{0.0850} & \textbf{0.4846} & \textbf{0.2542} & \textbf{0.1545} \\
\midrule
\multirow{4}{*}{Adversarial edits} & \texttt{AOL} & 0.0976 & 0.0931 & 0.0681 & 0.1131 & 0.1134 & 0.1003 \\
 & \texttt{WISER} & 0.3357 & \textbf{0.1712} & \textbf{0.1084} & 0.3648 & \textbf{0.1979} & \textbf{0.1467} \\
 & \Algo{}\texttt{-plugin} & 0.3681 & 0.1512 & 0.0546 & 0.4018 & 0.1803 & 0.1019 \\
 & \Algo{}\texttt{-oracle} & \textbf{0.3763} & 0.1637 & 0.0800 & \textbf{0.4095} & 0.1887 & 0.1149 \\
\midrule
\multirow{4}{*}{Block permutation} & \texttt{AOL} & 0.1048 & 0.0954 & 0.0641 & 0.1174 & 0.1112 & 0.0936 \\
 & \texttt{WISER} & 0.2688 & 0.1783 & \textbf{0.1163} & 0.2871 & 0.1994 & \textbf{0.1474} \\
 & \Algo{}\texttt{-plugin} & 0.4013 & 0.1754 & 0.0545 & 0.4289 & 0.2014 & 0.1003 \\
 & \Algo{}\texttt{-oracle} & \textbf{0.4122} & \textbf{0.1936} & 0.0970 & \textbf{0.4393} & \textbf{0.2152} & 0.1271 \\
\bottomrule
\end{tabular}}
\end{table}

\clearpage
\endgroup

%% file: 8-proof.tex
\section{Proof for Theoretical Guarantees}
\label{proof}

This section proves the theoretical guarantees in the order in which they are stated in the main text. For each result, we first present the main proof. Auxiliary lemmas are introduced when they are needed, while their proofs are deferred until after the corresponding theorem proof. This organization keeps the main argument transparent while retaining the technical details required for verification.

Throughout this supplementary material, the vocabulary and the dominant--core--light decomposition are as specified in Assumption~\ref{asmp:HCL+}. In particular, for each position $t$,
\[
\mathcal W_n=\{w_{t,n}^\star\}\cup C_{t,n}\cup L_{t,n},
\qquad
P_{t,w_{t,n}^\star}=1-\Delta_n,
\qquad
\Delta_n\asymp n^{-q},
\]
where $|L_{t,n}|\asymp n^\alpha$, $P_{t,w}\asymp n^{-(\alpha+q)}$ for $w\in L_{t,n}$, and $|C_{t,n}|\asymp n^r$ for some $r<\alpha$. Constants implicit in $O(\cdot)$, $o(\cdot)$, $\lesssim$, $\gtrsim$, and $\asymp$ do not depend on $t$ or $n$.

We first present the watermarked-component distribution of the pivotal statistic.  This is the same distribution referred to as the alternative law in the watermark-detection setting of~\citet{li2024statistical}; in the present mixed-source setting, however, it serves as one component of the mixture model.  The same finite-vocabulary calculation still applies here because $\mathcal W_n$ is finite for each fixed $n$.

\begin{lem}[Alternative law \citep{li2024statistical}]
\label{lem:alt-pivot-density}
Fix $t$ and $n$.  Conditional on the NTP distribution $\bP_t$ over the vocabulary $\mathcal W_n$, with $P_{t,w}>0$ and $\sum_{w\in\mathcal W_n}P_{t,w}=1$, the Gumbel--max pivotal statistic $Y_t\in(0,1)$ generated from the watermarked component has density or probability density function (PDF)
\begin{equation}
\label{eq:alt-pivot-density}
f_{1,\bP_t}(y)=\sum_{w\in\mathcal W_n} y^{1/P_{t,w}-1},
\qquad 0<y<1.
\end{equation}
Equivalently, its conditional CDF is
\begin{equation}
\label{eq:alt-pivot-cdf}
F_{1,\bP_t}(y)=\sum_{w\in\mathcal W_n} P_{t,w}y^{1/P_{t,w}},
\qquad 0\le y\le1.
\end{equation}
\end{lem}

\paragraph{Common notation.}
Recall that $\FM_{t-1}:=\sigma(\{\token_j,\xi_j,\bP_{j+1}\}_{j=1}^{t-1})$ contains the history used to form $\bP_t$, and $\mathcal H_t:=\sigma(\theta_1,\ldots,\theta_t)$ is the editing filtration.
We also define $\mathcal G_t:=\FM_{t-1}\vee\mathcal H_t$ by the joint generation-editing filtration and $\mathcal Y_t:=\sigma(Y_1,\ldots,Y_t)$ by the filtration generated by the observed pivotal statistics, and write
\[
\pi_t:=\mathbb P(\theta_t=1\mid\mathcal F_{t-1}).
\]
Assumption~\ref{asmp:main}(c) implies that, uniformly in $t$, $c\varepsilon_n\le \pi_t\le C\varepsilon_n$ almost surely.
Since the observed filtration $\mathcal Y_{t-1}$ is generally coarser than $\mathcal F_{t-1}$, we use the following conditional mixture quantities:
\[
\bar\pi_t:=\mathbb E[\pi_t\mid\mathcal Y_{t-1}],
\qquad
\bar f_{1,t}(y):=
\frac{\mathbb E[\pi_t f_{1,\bP_t}(y)\mid\mathcal Y_{t-1}]}{\bar\pi_t},
\]
and denote the corresponding watermarked-component law by $\bar\mu_{1,t}$.  
In short,  $\bar\mu_{1,t}$ is the effective watermarked-component law when conditioning on the observed history of pivotal statistics $\mathcal Y_{t-1}$.
Then, conditional on $\mathcal Y_{t-1}$, the one-step density of $Y_t$ is
\[
(1-\bar\pi_t)+\bar\pi_t\bar f_{1,t}(y),
\qquad 0<y<1,
\]
where $c\varepsilon_n\le \bar\pi_t\le C\varepsilon_n$ almost surely. 
We will prove the last equation in Lemma \ref{lem:conditional-pivot-mixture} soon.
We denote the null law by $\mu_0$, and write $\mathbb E_0$ for expectation under $Y\sim\mu_0$.  In the Gumbel--max case, $\mu_0$ is the uniform law on $(0,1)$.  Similarly, $\mu_{1,\bP_t}$ denotes the watermarked-component law with density given in~\eqref{eq:alt-pivot-density}, and $\mathbb E_{1,\bP_t}$ denotes expectation under $Y\sim\mu_{1,\bP_t}$.

\subsection{Proof of Theorem~\ref{thm:gumbel-general}: Detection boundary}

The proof follows the autoregressive Hellinger-affinity argument of~\citet{li2025robust}.  The tensorization step is standard; the new calculation is the one-step Hellinger scale under the growing dominant--core--light vocabulary.

\begin{lem}[Conditional one-step mixture]
\label{lem:conditional-pivot-mixture}
Under Assumption~\ref{asmp:main}, for any integrable function $h$,
\begin{equation}
\label{eq:conditional-pivot-mixture}
\mathbb E\big[h(Y_t)\mid \mathcal F_{t-1}\big]
=(1-\pi_t)\mathbb E_0h(Y)+\pi_t\mathbb E_{1,\bP_t}h(Y).
\end{equation}
Consequently, conditional on $\mathcal Y_{t-1}$, the one-step alternative density is
\begin{equation}
\label{eq:conditional-density-Ypast}
1-\bar\pi_t+\bar\pi_t\bar f_{1,t}(y).
\end{equation}
\end{lem}

\begin{lem}[Autoregressive Hellinger reduction]
\label{lem:H2}
Let $\rho_0$ and $\rho_1$ denote the joint densities of $Y_{1:n}$ under $H_0$ and $H_1$, respectively, and let $\mathcal P_n$ be the class of NTP distributions satisfying Assumption~\ref{asmp:HCL+}.
\begin{enumerate}[label=(\roman*),leftmargin=2.5em]
\item We have $\mathrm{TV}(\rho_0,\rho_1)\to0$, if
\begin{equation}
\label{eq:hellinger-merging}
\sum_{t=1}^n
\sup_{\substack{\bP\in\mathcal P_n\\ \gamma\in[c\varepsilon_n,C\varepsilon_n]}}
H^2\big(\mu_0,(1-\gamma)\mu_0+\gamma\mu_{1,\bP}\big)=o(1).
\end{equation}
\item We have $\mathrm{TV}(\rho_0,\rho_1)\to1$, if there exists a deterministic sequence $a_n>0$ with $na_n\to\infty$ such that, almost surely,
\begin{equation}
\label{eq:hellinger-separation-condition}
\min_{1\le t\le n}
H^2\big(\mu_0,(1-\bar\pi_t)\mu_0+\bar\pi_t\bar\mu_{1,t}\big)
\ge a_n.
\end{equation}
\end{enumerate}
\end{lem}

\begin{lem}[Dominant--core--light Hellinger estimates]
\label{lem:HCL-hellinger-estimates}
Assume Assumption~\ref{asmp:HCL+} is true. It follows that, uniformly over $t$ and all admissible $\bP_t$,
\begin{equation}
\label{eq:second-moment-scale}
\mathbb E_0\big(f_{1,\bP_t}(Y)-1\big)^2\asymp n^{\alpha-q}.
\end{equation}
Moreover, if $\gamma_n\asymp n^{-p}$ for $p>0$, and if $\gamma_n$ is bounded away from zero and one for $p=0$, then
\begin{equation}
\label{eq:hellinger-scale-uncond}
H^2\big(\mu_0,(1-\gamma_n)\mu_0+\gamma_n\mu_{1,\bP_t}\big)
\asymp
\begin{cases}
\gamma_n^2 n^{\alpha-q}, & p>\alpha,\\
\gamma_n\Delta_n, & 0\le p\le \alpha.
\end{cases}
\end{equation}
The above (lower) bounds also hold almost surely with $\mu_{1,\bP_t}$ replaced by $\bar\mu_{1,t}$ and $\gamma_n$ replaced by $\bar\pi_t$, that is,
\begin{equation}
\label{eq:second-moment-scale1}
\EB_0(\bar{f}_{1,t}(Y)-1)^2  \asymp n^{\alpha-q}
\quad \text{and} \quad
H^2\big(\mu_0,(1-\bar{\pi}_t)\mu_0+\bar{\pi}_t\bar{\mu}_{1,t}\big)
\gtrsim
\begin{cases}
\gamma_n^2 n^{\alpha-q}, & p>\alpha,\\
\gamma_n\Delta_n, & 0\le p\le \alpha.
\end{cases}
\end{equation}
\end{lem}

\begin{proof}[Proof of Theorem~\ref{thm:gumbel-general}]
We essentially use Lemma~\ref{lem:H2} together with Lemma~\ref{lem:HCL-hellinger-estimates}.  Since $\pi_t\asymp\varepsilon_n$ uniformly in $t$, and since $\varepsilon_n\asymp n^{-p}$ with the convention $n^{-p}=1$ when $p=0$, the one-step Hellinger scale is
\begin{equation}
\label{eq:one-step-H-order-v3}
H^2\big(\mu_0,(1-\pi_t)\mu_0+\pi_t\mu_{1,\bP_t}\big)
\asymp
\begin{cases}
\varepsilon_n^2 n^{\alpha-q}, & p>\alpha,\\
\varepsilon_n\Delta_n, & 0\le p\le\alpha.
\end{cases}
\end{equation}
For $p=0$, the second line means $H^2\asymp\Delta_n$.  The same two orders are valid as almost-sure lower bounds after conditioning on the observed past and replacing $\mu_{1,\bP_t}$ by $\bar\mu_{1,t}$.

We first prove the merging case.  If $\max\{p+q,2p+q-\alpha\}>1$, then necessarily $p>0$ under $p,q\in[0,1]$ and $\alpha\in[0,1)$.  By Lemma~\ref{lem:H2}(i), it suffices to show that the sum of the one-step Hellinger terms tends to zero.  In the regime $p>\alpha$, using \eqref{eq:one-step-H-order-v3}, we have
\[
\sum_{t=1}^n H^2\big(\mu_0,(1-\pi_t)\mu_0+\pi_t\mu_{1,\bP_t}\big)
\lesssim n\varepsilon_n^2n^{\alpha-q}
=n^{1+\alpha-q-2p+o(1)},
\]
which tends to zero when $2p+q-\alpha>1$.  In the regime $0<p\le\alpha$, similarly we have
\[
\sum_{t=1}^n H^2\big(\mu_0,(1-\pi_t)\mu_0+\pi_t\mu_{1,\bP_t}\big)
\lesssim n\varepsilon_n\Delta_n
=n^{1-p-q+o(1)},
\]
which tends to zero when $p+q>1$.  Therefore the joint null and alternative laws merge whenever $\max\{p+q,2p+q-\alpha\}>1$.

We next prove separation.  If $p>\alpha$, take $a_n\asymp\varepsilon_n^2n^{\alpha-q}$ in Lemma~\ref{lem:H2}(ii).  Then $na_n\to\infty$ precisely when $2p+q-\alpha<1$.  If $0\le p\le\alpha$, take $a_n\asymp\varepsilon_n\Delta_n$, with $\varepsilon_n\asymp1$ in the dense case $p=0$.  Then $na_n\to\infty$ precisely when $p+q<1$.  Hence, whenever
\[
\max\{p+q,2p+q-\alpha\}<1,
\]
Lemma~\ref{lem:H2}(ii) gives $\mathrm{TV}(\rho_0,\rho_1)\to1$.  The likelihood-ratio test has minimum total testing error $1-\mathrm{TV}(\rho_0,\rho_1)$, so its sum of Type~I and Type~II errors tends to zero.
\end{proof}

\paragraph{Auxiliary proofs for Theorem~\ref{thm:gumbel-general}.}

In the following, we present the omitted proofs for the lemmas used above.

\begin{proof}[Proof of Lemma~\ref{lem:conditional-pivot-mixture}]
Conditional on $\mathcal F_{t-1}$, the NTP distribution $\bP_t$ is fixed and $\theta_t$ has conditional success probability $\pi_t$.  If $\theta_t=1$, the observed token $w_t$ is generated by the Gumbel--max decoder and the pivotal statistic $Y_t$  has law $\mu_{1,\bP_t}$.  If $\theta_t=0$, the observed token $w_t$ is conditionally independent of the pseudorandomness $\zeta_t$ so that the pivotal statistic $Y_t$  has law $\mu_0$.  Taking the conditional expectation over $\theta_t$ gives~\eqref{eq:conditional-pivot-mixture}. 

Note that $\mathcal Y_{t-1} \subset \mathcal{F}_{t-1}$ because $Y_t$ is a measurable function of $(w_t, \zeta_t)$. 
The law of total expectation, or the tower rule, implies that we have 
$\EB[\EB[h(Y)|\FM_{t-1}]|\mathcal{Y}_{t-1}] = \EB[h(Y)|\mathcal{Y}_{t-1}]$.
When taking a further conditional expectation over $\mathcal Y_{t-1}$ on both sides of \eqref{eq:conditional-pivot-mixture}, we have
\[
\mathbb E\big[h(Y_t)\mid \mathcal Y_{t-1}\big]
=(1-\bar{\pi}_t)\mathbb E_0h(Y)+\bar{\pi}_t\mathbb E_{\bar{f}_{1,t}}[h(Y)].
\]
Since the above identity holds for all bounded measurable $h$, the conditional law of $Y_t$ over $\mathcal{Y}_{t-1}$ has the PDF given in~\eqref{eq:conditional-density-Ypast}.
\end{proof}

\begin{proof}[Proof of Lemma~\ref{lem:H2}]
For probability measures \(P\) and \(Q\) dominated by a common measure \(\nu\), with densities
\(p\) and \(q\), we use the convention
\[
H^2(P,Q)
:=1-\int \sqrt{pq}\,\rd\nu
=\frac12\int(\sqrt p-\sqrt q)^2\,\rd\nu .
\]
Equivalently, the Hellinger affinity is
\[
\mathcal A(P,Q):=\int \sqrt{pq}\,\rd\nu=1-H^2(P,Q).
\]
In our setting, the joint null law of $(Y_1,\ldots,Y_n)$ is $\rho_0=\mathrm{Unif}(0,1)^{\otimes n},$ the uniform law on the unit cube $[0,1]^n$. We denote by $\rho_1$ the joint law of $(Y_1,\ldots,Y_n)$ under the mixed-source model whose conditional one-step densities are given by Lemma~\ref{lem:conditional-pivot-mixture}. 
Let \(A_t:=\mathbb E_0\sqrt{\rho_1(Y_{1:t})}\) with \(A_0=1\).
By the definition above, \(A_n\) is the Hellinger affinity between the alternative law $\rho_1$ and null joint law $\rho_0$.

For simplicity, we write $\lambda_t(y):=1-\bar\pi_t+\bar\pi_t\bar f_{1,t}(y),$ where \(\bar\pi_t=\mathbb P(\theta_t=1\mid\mathcal Y_{t-1})\) and
\(\bar\pi_t\bar f_{1,t}(y)= \mathbb E[\pi_t f_{1,\bP_t}(y)\mid\mathcal Y_{t-1}]\). At a high level, \(\lambda_t\) is the actual conditional density of \(Y_t\) given \(\mathcal Y_{t-1}\) under the
alternative $H_1$.  By Lemma~\ref{lem:conditional-pivot-mixture}, the alternative
likelihood factors recursively as
\begin{equation}
\label{eq:conditional-likelihood-factorization}
\rho_1(Y_{1:t})=
\rho_1(Y_{1:t-1})
\lambda_t(Y_t),
\qquad t=1,\ldots,n.
\end{equation}
Using the last equation, we then simplify the expression of $A_n$.
Conditioning on \(\mathcal Y_{t-1}\) under \(H_0\), and using \(Y_t\sim\mu_0\), gives
\[
A_t
=
\mathbb E_0\!\left[
\sqrt{\rho_1(Y_{1:t-1})}
\int_0^1\sqrt{\lambda_t(y)}\,\rd y
\right]
=
\mathbb E_0\!\left[
\sqrt{\rho_1(Y_{1:t-1})}\{1-h_t\}
\right],
\]
where, for simplicity, we denote by
\[
h_t:=
H^2\big(\mu_0,(1-\bar\pi_t)\mu_0+\bar\pi_t\bar\mu_{1,t}\big).
\]

For part~(i), the density \(\lambda_t\) is not a fixed density of the form
\(1-\gamma+\gamma f_{1,\bP}\); rather, it is the conditional average
$\lambda_t(y)
=
\mathbb E\!\left[
1-\pi_t+\pi_t f_{1,\bP_t}(y)
\,\middle|\,\mathcal Y_{t-1}
\right].$
Since \(H^2(\mu_0,\cdot)\) is convex in its second argument, or equivalently since the Hellinger
affinity is concave in the second density, we have
\[
h_t
\le
\mathbb E\!\left[
H^2\big(\mu_0,(1-\pi_t)\mu_0+\pi_t\mu_{1,\bP_t}\big)
\,\middle|\,\mathcal Y_{t-1}
\right].
\]
Using \(c\varepsilon_n\le\pi_t\le C\varepsilon_n\) and \(\bP_t\in\mathcal P_n\), it almost surely follows that
\[
h_t
\le
\sup_{\substack{\bP\in\mathcal P_n\\ \gamma\in[c\varepsilon_n,C\varepsilon_n]}}
H^2\big(\mu_0,(1-\gamma)\mu_0+\gamma\mu_{1,\bP}\big) =: b_t.
\]
The affinity recursion thus gives \(A_t\ge(1-b_t)A_{t-1}\).  Iterating gives $A_n\ge\prod_{t=1}^n(1-b_t)$ and 
\[
H^2(\rho_0,\rho_1)
=1-A_n
\le
1-\prod_{t=1}^n(1-b_t)
\le
\sum_{t=1}^n b_t
=o(1).
\]
Thus \(\mathrm{TV}(\rho_0,\rho_1)\le \sqrt{2}\,H(\rho_0,\rho_1)\to0\).

For part~(ii), the assumption is on the actual conditional one-step law, namely
\(h_t\ge a_n\) almost surely for every \(t\).  Hence, every conditional one-step affinity (that is $\int_0^1\sqrt{\lambda_t(y)}\,\rd y = 1- h_t$) is at most \(1-a_n\), and the same recursion yields $A_t \le (1-a_n)A_{t-1}.$
Therefore, it follows that
\[
A_n\le(1-a_n)^n\le \exp(-na_n)\to0.
\]
Finally, for any two probability measures \(P,Q\), $\mathrm{TV}(P,Q)\ge 1-\int\sqrt{\rd P\,\rd Q}.$
Applying this to the two joint laws gives $\mathrm{TV}(\rho_0,\rho_1)\ge 1-A_n\to1.$
This proves the lemma.
\end{proof}

\begin{proof}[Proof of the second-moment estimate \eqref{eq:second-moment-scale} in Lemma~\ref{lem:HCL-hellinger-estimates}]
Fix \(t\) and suppress the \(t,n\) subscripts whenever there is no ambiguity. So we write
\(P_w:=P_{t,w}\), \(w^\star:=w^\star_{t,n}\), \(C:=C_{t,n}\), and \(L:=L_{t,n}\) for simplicity.
 For each \(w\in\mathcal W_n\), set $b_w(y):=y^{\frac1{P_w}-1}-P_w$ for short.
Then \(f_{1,\bP_t}-1=\sum_{w\in\mathcal W_n}b_w(y)\) and
\(\int_0^1 b_w(y)\rd y=0\).  For \(w_1,w_2\in\mathcal W_n\), we introduce the helper function
\[
I(w_1,w_2):=\int_0^1 b_{w_1}(y)b_{w_2}(y)\rd y .
\]
A direct integration gives that, if we set \(p_i=P_{w_i}\),
\begin{equation}
\label{eq:Iww-v3}
I(w_1,w_2)
=\frac{1}{a_{w_1}+a_{w_2}-1} -P_{w_1}P_{w_2}
= p_1p_2 \frac{(1-p_1)(1-p_2)}{p_1+p_2-p_1p_2}.
\end{equation}
In particular, \(I(w_1,w_2)\ge0\), and since
\(p_1+p_2-p_1p_2\ge p_1\vee p_2\), we also have the useful bound $I(w_1,w_2)\le \min\{P_{w_1},P_{w_2}\}.$
Therefore, the target $\mathbb E_0 [f_{1,\bP_t}(Y)-1]^2$ is reformulated as 
\begin{equation}
\label{eq:second-moment-sun-decomposition}
\mathbb E_0 [f_{1,\bP_t}(Y)-1]^2 = \sum_{w_1,w_2\in\mathcal W_n} I(w_1,w_2).
\end{equation}

\paragraph{Upper bound.}
For the upper bound, we split the ordered pairs $(w_1,w_2)$ in the sum in~\eqref{eq:second-moment-sun-decomposition} into three disjoint classes.
First, we consider pairs with at least one light token.  If, say, \(w_1\in L\), then \(P_{w_1}\asymp n^{-(\alpha+q)}\) by definition.  By the bound $I(w_1,w_2)\le \min\{P_{w_1},P_{w_2}\}$, we have
\[
I(w_1,w_2)\le P_{w_1}\lesssim n^{-(\alpha+q)}
\qquad
\text{for every }w_2\in\mathcal W_n .
\]
Note that there are at most \(2|L||\mathcal W_n|\lesssim n^{2\alpha}\) ordered pairs with at least one
light coordinate.  Hence, their total contribution is bounded by
\[
\sum_{\substack{(w_1,w_2):\\ w_1\in L\ \text{or}\ w_2\in L}}
I(w_1,w_2)
\lesssim
n^{2\alpha}n^{-(\alpha+q)}
=
n^{\alpha-q}.
\]

Second, we consider pairs with no light token but with at least one core token.  In other words, we consider core--core and dominant--core pairs.  Let $\Delta^{\mathrm{core}}_{t,n}:=\sum_{w\in C}P_w$ denote the total probability mass of the core tokens.
By Assumption~\ref{asmp:HCL+}, \(|C|\asymp n^r\), \(r<\alpha\), and
\(\Delta^{\mathrm{core}}_{t,n}\lesssim n^{-q}\).  For the core--core pairs,
using \(I(w_1,w_2)\le P_{w_1}\) again gives
\[
\sum_{w_1,w_2\in C} I(w_1,w_2)
\le
\sum_{w_1\in C}\sum_{w_2\in C}P_{w_1}
=
|C|\sum_{w_1\in C}P_{w_1}
\lesssim
n^r n^{-q}
=
n^{r-q}.
\]
For the dominant--core pairs, again using \(I(w^\star,w)\le P_w\) and
\(I(w,w^\star)\le P_w\),
\[
\sum_{w\in C}\{I(w^\star,w)+I(w,w^\star)\}
\le
2\sum_{w\in C}P_w
\lesssim
n^{-q}.
\]
Thus, the total contribution from all pairs with no light token but at least one core token is
\(O(n^{r-q})\), which is \(O(n^{\alpha-q})\) because \(r<\alpha\).

Finally, we consider the dominant--dominant pair.  Since \(P_{w^\star} = 1-\Delta_n\),
\[
I(w^\star,w^\star)
=
\frac{1}{2/(1-\Delta_n)-1}
-
(1-\Delta_n)^2
=
\frac{1-\Delta_n}{1+\Delta_n}
-
(1-\Delta_n)^2
=
\frac{(1-\Delta_n)\Delta_n^2}{1+\Delta_n}.
\]
Therefore \(I(w^\star,w^\star)\lesssim \Delta_n^2\asymp n^{-2q}\), and since
\(\alpha\ge0\), this is also \(O(n^{\alpha-q})\).

Combining the three classes gives
\[
\mathbb E_0 [f_{\bP_t}(Y)-1]^2
=
\sum_{w_1,w_2\in\mathcal W_n} I(w_1,w_2)
\lesssim
n^{\alpha-q}.
\]

\paragraph{Lower bound.}
For the lower bound, we restrict the sum in~\eqref{eq:second-moment-sun-decomposition} to light--light pairs only.  If \(w_1,w_2\in L\), then
\(P_{w_i}\asymp n^{-(\alpha+q)}\).  Using the positive representation of \(I(w_1,w_2)\),
\[
I(w_1,w_2)
=
P_{w_1}P_{w_2}
\frac{(1-P_{w_1})(1-P_{w_2})}{P_{w_1}+P_{w_2}-P_{w_1}P_{w_2}}
\asymp
n^{-(\alpha+q)}.
\]
Indeed, the numerator is of order \(n^{-2(\alpha+q)}\), while the denominator is of order
\(n^{-(\alpha+q)}\).  There are \(|L|^2\asymp n^{2\alpha}\) light--light ordered pairs.  Since all
\(I(w_1,w_2)\) are nonnegative,
\[
\mathbb E_0 [f_{\bP_t}(Y)-1]^2
\ge
\sum_{w_1,w_2\in L}I(w_1,w_2)
\gtrsim
n^{2\alpha}n^{-(\alpha+q)}
=
n^{\alpha-q}.
\]
Together with the upper bound, this proves that, uniformly over $t$ and all admissible $\bP_t$,
\[
\mathbb E_0\big(f_{1,\bP_t}(Y)-1\big)^2
=
\mathbb E_0 [f_{\bP_t}(Y)-1]^2
\asymp
n^{\alpha-q}.
\]
\end{proof}

\begin{proof}[Proof of the Hellinger estimates \eqref{eq:hellinger-scale-uncond} in Lemma~\ref{lem:HCL-hellinger-estimates}]
Fix \(t\), and write \(f_t:=f_{1,\bP_t}\) and \(g_t:=f_t-1\).  For a mixture weight
\(\gamma_n\asymp\varepsilon_n\), we define
\(H^2_{\gamma_n}:= H^2\big(\mu_0,(1-\gamma_n)\mu_0+\gamma_n\mu_{1,\bP_t}\big)\) for simplicity.  The density of the mixture-source data is \(1+\gamma_n g_t\).  Thus, with \(Y\sim\mu_0\) and \(Z:=\gamma_n g_t(Y)\), we have \(\mathbb E_0 Z=0\) and
\[
H^2_{\gamma_n} = 1-\mathbb E_0\sqrt{1+Z} = \mathbb E_0 h(Z),
\qquad
h(z):=1-\sqrt{1+z}+\frac z2 .
\]
We shall use the elementary facts that \(h(z)\asymp z^2\) uniformly for \(|z|\) small; if
\(z\ge -\bar\gamma\) for some fixed \(\bar\gamma<1\), then \(0\le h(z)\lesssim_{\bar\gamma}|z|\);
and \(h(z)\asymp z\) for \(z\ge1\). Here, $\lesssim_{\bar\gamma}$ means the omitted constant depends only on $\bar\gamma$.
As a result, we have $h(z) \lesssim |z|$ for all $z \ge -\gamma_0$ where $\gamma_0 \in (0, 1)$ by definition.

\paragraph{First suppose \(p>\alpha\).}  Since \(0\le y^{1/P_{t,w}-1}\le1\) on \((0,1)\), we have
\(0\le f_t(y)\le|\mathcal W_n|\) and hence \(|Z|\le \gamma_n|\mathcal W_n|\).  Because
\(\gamma_n|\mathcal W_n|\asymp n^{\alpha-p}=o(1)\), the perturbation is uniformly small.  Therefore
\(h(Z)\asymp Z^2\), and the second-moment estimate
\eqref{eq:second-moment-scale} gives
\begin{equation}
\label{eq:help1}
H^2_{\gamma_n}
\asymp
\mathbb E_0Z^2
=
\gamma_n^2\mathbb E_0\big(f_t(Y)-1\big)^2
\asymp
\gamma_n^2 n^{\alpha-q}.
\end{equation}
This proves the quadratic Hellinger scale.

\paragraph{Now suppose \(0\le p\le\alpha\).}  We first prove the upper bound.  Since \(Z\ge-\gamma_n\), and the mixture weights considered here are bounded above by a constant strictly smaller than one, the
bound \(h(Z)\lesssim |Z|\) gives
\[
H^2_{\gamma_n}
\lesssim
\mathbb E_0|Z|
=
2\mathbb E_0 (Z)_+
=
2\gamma_n\mathbb E_0\big(f_t(Y)-1\big)_+ .
\]
Here we used \(\mathbb E_0 Z=0\) and the notation $(x)_+ := \max\{x, 0\}$.  The dominant token gives
\(f_t(y)\ge y^{1/P_{t,w^\star_{t,n}}-1}=y^{\Delta_n/(1-\Delta_n)}\).  Hence, \((1-f_t(y))_+\le 1-y^{\Delta_n/(1-\Delta_n)}\), and since
\(\mathbb E_0(f_t-1)_+=\mathbb E_0(1-f_t)_+\), we obtain
\[
\mathbb E_0\big(f_t(Y)-1\big)_+
\le
\int_0^1\left(1-y^{\Delta_n/(1-\Delta_n)}\right)\rd y
=
\Delta_n .
\]
Therefore, we prove that the upper bound that \(H^2_{\gamma_n}\lesssim \gamma_n\Delta_n \).

It remains to prove the matching lower bound in the regime \(0\le p\le\alpha\).  Let
\(A_n:=\{Y\ge 1-c n^{-(\alpha+q)}\}\), where \(c>0\) is a sufficiently small fixed constant.  Then
\(\mathbb P_0(A_n)\asymp n^{-(\alpha+q)}\) due to $\mu_0 = \mathrm{Unif}(0, 1)$.  We then consider a light token \(w\in L_{t,n}\) which satisfies the bounds \(P_{t,w} \le C n^{-(\alpha+q)}\) for some universal constant $C > 0$.
On the event \(A_n\), there exists a universal constant $c_0 > 0$ such that for any sufficiently large $n$,
\[
Y^{1/P_{t, w}-1} \ge (1-c n^{-(\alpha+q)})^{n^{\alpha+q}/C} \ge c_0.
\]
Therefore, uniformly over any $t$, \(g_t(Y) = f_t(Y)-1\gtrsim |L_{t,n}|\asymp n^\alpha\) on \(A_n\), and consequently \(Z = \varepsilon_n g_t(Y)\gtrsim \gamma_n n^\alpha\) on this event.  If \(p<\alpha\), then \(\gamma_n n^\alpha\to\infty\), so \(h(Z)\gtrsim Z\) on \(A_n\), and
\begin{equation}
\label{eq:help2}
H^2_{\gamma_n}
\ge
\mathbb E_0\big[h(Z)\mathbf 1_{A_n}\big]
\gtrsim
\gamma_n n^\alpha\,\mathbb P_0(A_n)
\asymp
\gamma_n n^{-q}
\asymp
\gamma_n\Delta_n .
\end{equation}
If \(p=\alpha\), including \(p=\alpha=0\), then by a similar argument, \(Z \gtrsim \gamma_n n^\alpha \asymp 1\) is bounded below by a positive constant on \(A_n\).  Hence, \(h(Z)\gtrsim1\) on \(A_n\), and thus
\[
H^2_{\gamma_n}
\gtrsim
\mathbb P_0(A_n)
\asymp
n^{-(\alpha+q)}
\asymp
\gamma_n\Delta_n .
\]
Together with the upper bound, this proves
\(H^2_{\gamma_n}\asymp\gamma_n\Delta_n\) for \(0\le p\le\alpha\).
\end{proof}

\begin{proof}[Proof of the conditional counterparts result \eqref{eq:second-moment-scale1} in Lemma~\ref{lem:HCL-hellinger-estimates}]
Finally, we verify that similar bounds hold for the actual conditional one-step watermarked density $\bar f_{1,t}$.  The argument is essentially the same as above, with $f_{1,\bP_t}$ replaced by $\bar f_{1,t}$.  For completeness, we spell out the only differences below.

We first prove the second-moment estimate.
By Jensen's inequality and~\eqref{eq:second-moment-scale}, we have $\mathbb E_0\big(\bar f_{1,t}(Y)-1\big)^2 \le \mathbb E_0\big( f_{1,\bP_t}(Y)-1\big)^2 \asymp n^{\alpha - q}$. 
The inverse direction follows from Lemma \ref{lem:conditional-second-moment-inverse}.
\begin{lem}
\label{lem:conditional-second-moment-inverse}
Under Assumptions \ref{asmp:main}---\ref{asmp:HCL+}, it follows that 
\[
\mathbb E_0\big( \bar{f}_{1,t}(Y)-1\big)^2 \gtrsim n^{\alpha - q}.
\]
\end{lem}

We next prove the corresponding lower bound for the Hellinger estimate, whose argument is similar to that for~\eqref{eq:hellinger-scale-uncond}.
In the first regime \(0\le p\le \alpha\), the pointwise lower bound \(f_t(Y)-1 \gtrsim n^{\alpha}\) on the event \(A_n:={Y\ge 1-c n^{-(\alpha+q)}}\) holds almost surely and uniformly over all admissible \(\bP_t\). This lower bound is thus preserved under the weighted conditional average defining \(\bar f_{1,t}\). Repeating the lower-bound argument in~\eqref{eq:help1}, with \(\gamma_n\) replaced by \(\bar\pi_t\), gives
\[
H^2\bigl(\mu_0,(1-\bar\pi_t)\mu_0+\bar\pi_t\bar\mu_{1,t}\bigr)
\gtrsim
\bar\pi_t\Delta_n
\qquad\text{a.s.}
\]
In the second regime \(p>\alpha\), we again have
\(\bar\pi_t|\mathcal W_n|=o(1)\), so the Hellinger integrand can be analyzed similarly as in \eqref{eq:help1}:
\[
H^2\big(\mu_0,(1-\bar\pi_t)\mu_0+\bar\pi_t\bar\mu_{1,t}\big)
\asymp
\bar\pi_t^{\,2}\mathbb E_0\big(\bar f_{1,t}(Y)-1\big)^2 .
\]
Since \(\bar\pi_t\asymp\varepsilon_n\), the desired conditional lower bound follows from Lemma \ref{lem:conditional-second-moment-inverse}.
\end{proof}

\begin{proof}[Proof of Lemma~\ref{lem:conditional-second-moment-inverse}]
Recall that \(\bar\pi_t:=\mathbb P(\theta_t=1\mid\mathcal Y_{t-1})\) and \(\bar\pi_t\bar f_{1,t}(y)= \mathbb E[\pi_t f_{1,\bP_t}(y)\mid\mathcal Y_{t-1}]\).  
Fix $t$ and condition on $\mathcal{Y}_{t-1}$. Let \(\rho\) denote the conditional law of \((\pi_t,\bP_t)\) given \(\mathcal Y_{t-1}\), and let \((\pi_1,\bP_1)\) and \((\pi_2,\bP_2)\) be two independent draws from \(\rho\).  Since
\(\mathbb E_0\bar f_{1,t}(Y)=1\), we have
\[
\begin{aligned}
\mathbb E_0\bigl(\bar f_{1,t}(Y)-1\bigr)^2
&= \mathbb E_0\bar f_{1,t}(Y)^2-1 = \frac{
\mathbb E_{\rho\otimes\rho}
\left[
\pi_1\pi_2
\left\{
\mathbb E_0 f_{1,\bP_1}(Y)f_{1,\bP_2}(Y)-1
\right\}
\right]
}{
\bigl(\mathbb E_\rho\pi_1\bigr)^2
}.
\end{aligned}
\]
It remains to obtain a uniform lower bound for
\(\mathbb E_0 f_{1,\bP_1}(Y)f_{1,\bP_2}(Y)-1\) over all admissible
\(\bP_1,\bP_2\).
By the same calculation as in equation~(20) of~\citet{li2025robust},
\[
\begin{aligned}
\mathbb E_0 f_{1,\bP_1}(Y)f_{1,\bP_2}(Y)-1
&=
\sum_{w\in\mathcal W_n}\sum_{j\in\mathcal W_n}
\frac{
P_{1,w}P_{2,j}(1-P_{1,w})(1-P_{2,j})
}{
1-(1-P_{1,w})(1-P_{2,j})
}  \\
&\ge
\frac12
\sum_{w\in\mathcal W_n}\sum_{j\in\mathcal W_n}
(P_{1,w}\wedge P_{2,j})(1-P_{1,w})(1-P_{2,j}).
\end{aligned}
\]
For \(k=1,2\), let \(L_n(\bP_k)\) be the light set of \(\bP_k\).  Under
Assumption~\ref{asmp:HCL+}, we have $|L_n(\bP_k)|\asymp n^\alpha$ and $P_{k,w}\asymp n^{-(\alpha+q)}$ for any $w\in L_n(\bP_k).$ 
For all sufficiently large \(n\), the light probabilities are uniformly small, so
\((1-P_{k,w})\ge 1/2\) on \(L_n(\bP_k)\).  Restricting the preceding sum to
\(w\in L_n(\bP_1)\) and \(j\in L_n(\bP_2)\), we obtain
\[
\begin{aligned}
\mathbb E_0 f_{1,\bP_1}(Y)f_{1,\bP_2}(Y)-1
&\gtrsim
\sum_{w\in L_n(\bP_1)}
\sum_{j\in L_n(\bP_2)}
n^{-(\alpha+q)}  \asymp
|L_n(\bP_1)|\,|L_n(\bP_2)|\,n^{-(\alpha+q)}
\asymp n^{\alpha-q}.
\end{aligned}
\]
This lower bound is uniform over all admissible \(\bP_1,\bP_2\).  Since
\(\pi_t\asymp \varepsilon_n\) uniformly by Assumption~\ref{asmp:main}(c), the weighted average above preserves this lower bound:
\[
\mathbb E_0\bigl(\bar f_{1,t}(Y)-1\bigr)^2
\gtrsim
\frac{
\mathbb E_{\rho\otimes\rho}[\pi_1\pi_2]\, n^{\alpha-q}
}{
(\mathbb E_\rho\pi_1)^2
}
\asymp n^{\alpha-q}.
\]
The bound holds almost surely in \(\mathcal Y_{t-1}\), uniformly in \(t\).  This proves the lemma.
\end{proof}

\subsection{Proof of Theorem~\ref{thm:discovery-boundary}: Discovery boundary}

The proof is organized around the right-tail behavior of the pivotal statistic.  The tail bounds in Lemma~\ref{lem:tail-pivot-HCL} will be used repeatedly in the proofs of the classification and adaptivity results.

\begin{lem}[Pivot tails under the dominant--core--light model]
\label{lem:tail-pivot-HCL}
Under Assumptions~\ref{asmp:main} and~\ref{asmp:HCL+}, for every $u>0$,
\begin{equation}
\label{eq:tail-null-HCL}
\mathbb P\big(Y_t>1-n^{-u}\mid \theta_t=0\big)=n^{-u}.
\end{equation}
Moreover, uniformly in $t$,
\begin{subequations}
\label{eq:tail-signal}
\begin{empheq}[left={\mathbb P\big(Y_t>1-n^{-u}\mid \theta_t=1\big)\asymp\empheqlbrace}]{align}
n^{-u}+n^{-q}, &\qquad 0<u<\alpha+q, \label{eq:tail-signal-u-less}\\
n^{\alpha-u}, &\qquad u\ge \alpha+q . \label{eq:tail-signal-u-greater}
\end{empheq}
\end{subequations}
\end{lem}

\begin{lem}[Bayes tail rule]
\label{lem:bayes-tail-rule-v3}
Fix \(\lambda>0\).  A local decision rule is denoted by
\(\Bdelta=(\delta_1,\ldots,\delta_n)\in\Decision\), where each coordinate decision has the form
\(\delta_t=\phi_t(Y_t)\in\{0,1\}\).  Consider the coordinatewise classification loss
\[
L(\Bdelta) = \frac{1}{n} \sum_{t=1}^n \left(
\theta_t(1-\delta_t)+\lambda(1-\theta_t)\delta_t\right) .
\]
Then the Bayes rule $\Bdelta^{\star} = (\delta_1^\star, \ldots, \delta_n^\star)$ that minimizes the expected loss $\EB[L(\Bdelta)]$ is also coordinatewise. Specifically, for each \(t\), its coordinate decision is a right-tail threshold in \(Y_t\), that is, for some exponent \(u_{n,t}\in[0,\infty]\),
\[
\delta_t^\star=\mathbf 1\{Y_t>1-n^{-u_{n,t}}\},
\]
with the conventions \(n^{-\infty}=0\) and \(n^0=1\).
\end{lem}

\begin{lem}[Reduction to right-tail threshold rules]
\label{lem:threshold-reduction}
For every local rule \(\Bdelta\in \Decision\), there exists a coordinatewise right-tail threshold rule \(\tilde\Bdelta=(\tilde\delta_1,\ldots,\tilde\delta_n)\) of the form
\[
\tilde\delta_t=\mathbf 1\{Y_t>1-a_{n,t}\},
\qquad 0\le a_{n,t}\le1,
\]
such that \(\tilde\Bdelta\) is at least as good as \(\Bdelta\) in expected false and true discoveries:
\[
\mathrm{EFP}_{\tilde\Bdelta}\le \mathrm{EFP}_{\Bdelta},
\qquad
\mathrm{ETP}_{\tilde\Bdelta}\ge \mathrm{ETP}_{\Bdelta}.
\]
Equivalently, for $a_{n,t}>0$, one may define $a_{n,t}=n^{-u_{n,t}}$ and $\tilde\delta_t=\mathbf 1\{Y_t>1-n^{-u_{n,t}}\}$.  
Therefore, to prove impossibility for all local rules, it suffices to prove impossibility for coordinatewise right-tail threshold rules.
\end{lem}

\begin{proof}[Proof of Theorem~\ref{thm:discovery-boundary}] 
Recall that for any local rule \(\Bdelta=(\delta_1, \ldots, \delta_n)\),
\(\mathrm{EFP}_{\Bdelta}:=\mathbb E[\sum_{t=1}^n(1-\theta_t)\delta_t]\) and
\(\mathrm{ETP}_{\Bdelta}:=\mathbb E[\sum_{t=1}^n\theta_t\delta_t]\).  Equivalently,
\(\mathrm{EFP}_{\Bdelta}=\sum_{t=1}^n\mathbb P(\theta_t=0)\mathbb E[\delta_t\mid\theta_t=0]\) and
\(\mathrm{ETP}_{\Bdelta}=\sum_{t=1}^n\mathbb P(\theta_t=1)\mathbb E[\delta_t\mid\theta_t=1]\).

\paragraph{We first prove achievability.}  Suppose $p<\alpha$ and $p+q<1$.  Choose a constant $u\in(p+q,\alpha+q)$ and define the local rule $\delta^{(u)}_t:=\mathbf 1\{Y_t>1-n^{-u}\}$ with \(\Bdelta^{(u)} :=(\delta^{(u)}_1,\ldots,\delta^{(u)}_n)\).
By Lemma~\ref{lem:tail-pivot-HCL} and Assumption~\ref{asmp:main}(c),
\begin{equation}
\label{eq:discovery-ach-EFP}
\mathrm{EFP}_{\Bdelta^{(u)}}
=\sum_{t=1}^n\mathbb P(\theta_t=0)n^{-u}
\asymp n^{1-u},
\end{equation}
and, since $p+q < u<\alpha+q$,
\begin{equation}
\label{eq:discovery-ach-ETP}
\mathrm{ETP}_{\Bdelta^{(u)}}
=\sum_{t=1}^n\mathbb P(\theta_t=1)
\mathbb P(Y_t>1-n^{-u}\mid \theta_t=1)
\asymp n\varepsilon_n(n^{-u}+n^{-q})
\asymp n^{1-p-q}.
\end{equation}
Thus $\mathrm{ETP}_{\Bdelta^{(u)}}\to\infty$ because $p+q<1$, and
\[
\frac{\mathrm{ETP}_{\Bdelta^{(u)}}}{\mathrm{EFP}_{\Bdelta^{(u)}}}
\gtrsim n^{u-(p+q)}\to\infty
\]
because $u>p+q$.  Hence $\mathrm{mFDR}_{\Bdelta^{(u)}}\to0$ by definition.

It remains to show that a discovery is made with probability tending to one.  Let
\[
Z_t:=\delta^{(u)}_t\mathbf 1\{\theta_t=1\},
\qquad
m_n:=\mathbb E\sum_{t=1}^nZ_t=\mathrm{ETP}_{\Bdelta^{(u)}}.
\]
By Assumption~\ref{asmp:alphamixing}, $\{Z_t\}_{t=1}^n$ is bounded and geometrically strongly mixing.  Lemma~\ref{lem:mixing-rare-count} gives
\[
\frac{\mathrm{Var}(\sum_tZ_t)}{m_n^2}
\lesssim \frac{\log n}{m_n}\to0.
\]
Chebyshev's inequality implies $\mathbb P(\sum_tZ_t=0)\to0$.  Therefore, $\mathbb P(|S_{\Bdelta^{(u)}}|\ge1)\to1$, proving the possibility part.

\begin{lem}[Variance bound for rare discoveries under geometric mixing]
\label{lem:mixing-rare-count}
Let $\{Z_t\}_{t=1}^n$ be a triangular array with $0\le Z_t\le 1$.  Define its strong-mixing coefficients by
\[
\alpha_Z(k):=
\sup_{1\le t\le n-k}
\alpha\bigl(\sigma(Z_1,\ldots,Z_t),\sigma(Z_{t+k},Z_{t+k+1},\ldots,Z_n)\bigr),
\]
where $\alpha(\cdot,\cdot)$ denotes the usual strong-mixing coefficient between two sigma-fields, defined by $\alpha(\mathcal A,\mathcal B):=\sup_{A\in\mathcal A,B\in\mathcal B}
|\mathbb P(A\cap B)-\mathbb P(A)\mathbb P(B)|$.  Suppose that $\alpha_Z(k)\le C\rho^k$ for some constants $C>0$ and $\rho\in(0,1)$.  Let $m_n:=\sum_{t=1}^n\mathbb E Z_t$.  Then
\begin{equation}
\label{eq:mixing-rare-count-var}
\mathrm{Var}\Big(\sum_{t=1}^n Z_t\Big)
\lesssim m_n\log n.
\end{equation}
Consequently, if $m_n/\log n\to\infty$, then $\sum_{t=1}^n Z_t>0$ with probability tending to one.
\end{lem}

\paragraph{We next prove impossibility.}  By Lemma~\ref{lem:threshold-reduction}, it suffices to consider
threshold rules of the form \(\delta_t=\mathbf 1\{Y_t>1-a_{n,t}\}\), where \(0\le a_{n,t}\le1\).
Set \(A_n:=\sum_{t=1}^n a_{n,t}\).  Since \(Y_t\mid\{\theta_t=0\}\sim\mathrm{Unif}(0,1)\), the null
rejection probability at coordinate \(t\) is exactly \(a_{n,t}\).  Therefore
\[
\mathrm{EFP}_{\Bdelta}
=
\sum_{t=1}^n\mathbb P(\theta_t=0)a_{n,t}
\asymp
A_n,
\]
because \(\mathbb P(\theta_t=0)\) is bounded away from zero uniformly in \(t\).
We also need a uniform upper bound on the true-positive contribution.  Lemma~\ref{lem:tail-pivot-HCL}
implies that for every threshold level \(a\in[0,1]\),
\(\mathbb P(Y_t>1-a\mid\theta_t=1)\lesssim a+n^{-q}\).  Applying this with
\(a=a_{n,t}\), and using \(\mathbb P(\theta_t=1)\asymp\varepsilon_n\), gives
\begin{equation}
\label{eq:bound-ETP}
\mathrm{ETP}_{\Bdelta}
=
\sum_{t=1}^n
\mathbb P(\theta_t=1)
\mathbb P(Y_t>1-a_{n,t}\mid\theta_t=1)
\lesssim
\varepsilon_n A_n+n^{1-p-q}.
\end{equation}

Now, we are ready to prove the impossible result. We do this by considering two cases.
\begin{itemize}
    \item First, suppose \(p+q>1\) and, toward a contradiction, that discovery is achieved.  Since the definition of discovery implies that
\[
\mathrm{mFDR}_{\Bdelta}
=
\frac{\mathrm{EFP}_{\Bdelta}}
{\mathrm{EFP}_{\Bdelta}+\mathrm{ETP}_{\Bdelta}}
\to0,
\]
we must have \(\mathrm{EFP}_{\Bdelta}=o(\mathrm{ETP}_{\Bdelta})\); otherwise the ratio above could not
vanish.  Since \(\mathrm{EFP}_{\Bdelta}\asymp A_n\), this gives \(A_n=o(\mathrm{ETP}_{\Bdelta})\).
Substituting into the preceding upper bound \eqref{eq:bound-ETP} yields
\[
\mathrm{ETP}_{\Bdelta}
\lesssim
\varepsilon_n o(\mathrm{ETP}_{\Bdelta})+n^{1-p-q}
=
o(\mathrm{ETP}_{\Bdelta})+o(1),
\]
where the last term is \(o(1)\) because \(p+q>1\).  Hence \(\mathrm{ETP}_{\Bdelta}=o(1)\), and then
\(\mathrm{EFP}_{\Bdelta}=o(\mathrm{ETP}_{\Bdelta})=o(1)\) as well.  Consequently
\(\mathbb E|S_{\Bdelta}|=\mathrm{EFP}_{\Bdelta}+\mathrm{ETP}_{\Bdelta}=o(1)\).  Markov's inequality gives
\(\mathbb P(|S_{\Bdelta}|\ge1)\le\mathbb E|S_{\Bdelta}|\to0\), contradicting the discovery
requirement.

\item Second, we consider the case where $p\ge\alpha$. Lemma~\ref{lem:tail-pivot-HCL} implies the uniform domination
\begin{equation}
\label{eq:signal-tail-domination-alpha}
\mathbb P(Y_t>1-a\mid\theta_t=1)\lesssim n^\alpha a,
\qquad 0\le a\le1.
\end{equation}
Indeed, for $a=n^{-u}$, if $u<\alpha+q$, then $n^{-q}\le n^\alpha n^{-u}$, while if $u\ge\alpha+q$, the bound follows from~\eqref{eq:tail-signal-u-greater}; general $a$ follows by monotonicity.  Therefore,
\begin{equation}
\label{eq:help3}
\mathrm{ETP}_{\Bdelta}
\lesssim \varepsilon_n n^\alpha A_n
\asymp n^{\alpha-p}A_n
\asymp n^{\alpha-p}\mathrm{EFP}_{\Bdelta}.
\end{equation}
If \(\Bdelta\) achieves discovery, then \(\mathrm{mFDR}_{\Bdelta}\to 0\), which requires
\(\mathrm{EFP}_{\Bdelta}/\mathrm{ETP}_{\Bdelta}\to 0\) whenever \(\mathrm{ETP}_{\Bdelta}>0\).  However, the preceding display in \eqref{eq:help3} implies
\[
\frac{\mathrm{EFP}_{\Bdelta}}{\mathrm{ETP}_{\Bdelta}}
\gtrsim n^{p-\alpha},
\]
which is bounded away from zero when \(p\ge \alpha\) (even when $p=\alpha=0$).  Hence, \(\mathrm{mFDR}_{\Bdelta}\) cannot vanish.  Therefore, no local rule achieves discovery when \(p\ge \alpha\).
\end{itemize}

\end{proof}

\paragraph{Auxiliary proofs for Theorem~\ref{thm:discovery-boundary}.}
In the following, we present the omitted proofs for the lemmas used above.

\begin{proof}[Proof of Lemma~\ref{lem:tail-pivot-HCL}]
The proof has two steps.  First, we use the latent mixture mechanism to identify the tail law of the pivotal statistic conditional on the label $\theta_t$.  The null case is straightforward.  For the watermarked component, we then evaluate the tail probability by decomposing the NTP distribution into its dominant, light, and core parts.

We now start with the first step and write $x=n^{-u}$.  Since $\sigma(\theta_t,\bP_t)\subset\mathcal G_t$, conditional on $\mathcal G_t$ the values of $\theta_t$ and $\bP_t$ are fixed and should be considered as non-deterministic.  By Assumption~\ref{asmp:main}(b), if $\theta_t=0$, then $w_t$ is conditionally independent of $\zeta_t$, and the pivotal statistic $Y_t$ has the null law $\mu_0$.  If $\theta_t=1$, then $w_t=S(\bP_t,\zeta_t)$ is generated from the watermarked component.  Therefore, by the watermarked-component CDF in~\eqref{eq:alt-pivot-cdf},
\[
\mathbb P( Y_t>1-x \mid \mathcal G_t)
= (1-\theta_t)x + \theta_t \sum_{w\in\mathcal W_n} P_{t,w}\Bigl\{1-(1-x)^{1/P_{t,w}}\Bigr\}.
\]
Since $\sigma(\theta_t)\subset\mathcal G_t$, the tower property gives for
$i\in\{0,1\}$,
\[
\mathbb P(Y_t>1-x\mid\theta_t=i)=\mathbb E[\mathbb P(Y_t>1-x\mid\mathcal G_t)\mid\theta_t=i].
\]
 Hence, $\mathbb P(Y_t>1-n^{-u}\mid\theta_t=0)=n^{-u}$, which proves the null case.
For the watermarked component, the same identity gives
\[
\mathbb P(Y_t>1-x\mid\theta_t=1)
=
\mathbb E_{\bP_t}\!\left[
\sum_{w\in\mathcal W_n}
P_{t,w}\Bigl\{1-(1-x)^{1/P_{t,w}}\Bigr\}
\,\middle|\,\theta_t=1
\right].
\]
The estimates below are uniform over all admissible $\bP_t$, so it suffices to analyze the inner sum conditional on a particular $\bP_t$.  By Assumption~\ref{asmp:HCL+}, we split this sum into the dominant token $w_{t,n}^{\star}$, the light set $L_{t,n}$, and the core set $C_{t,n}$:
\begin{equation}
\label{eq:help4}
\mathbb P(Y_t>1-x\mid \theta_t=1,\bP_t)
=
\left[
\sum_{w=w_{t,n}^{\star}}
+
\sum_{w\in L_{t,n}}
+
\sum_{w\in C_{t,n}}
\right]
P_{t,w}\Bigl\{1-(1-x)^{1/P_{t,w}}\Bigr\}.
\end{equation}
To analyze each summand, we repeatedly use the elementary inequality in Lemma \ref{lem:help-inequality}.
\begin{lem}
\label{lem:help-inequality}
There exist constants \(c_1,c_2>0\) such that, for all \(x\in[0,1/2]\) and \(P\in(0,1]\),
\[
c_1(x\wedge P)
\le
P\bigl\{1-(1-x)^{1/P}\bigr\}
\le
c_2(x\wedge P).
\]
\end{lem}
\begin{proof}[Proof of Lemma \ref{lem:help-inequality}]
Let \(r=x/P\).  Since \(x\in[0,1/2]\), we have \(-2x\le \log(1-x)\le -x\), and hence \(1-e^{-r}\le 1-(1-x)^{1/P}\le 1-e^{-2r}\). 
We assert that \(1-e^{-r}\asymp 1\wedge r\) for \(r\ge0\).
This is because for \(0\le r\le 1\), \(e^{-1}r\le 1-e^{-r}\le r\), while for \(r\ge 1\), \(1-e^{-1}\le 1-e^{-r}\le 1\).
Hence, we get \(1-(1-x)^{1/P}\asymp 1\wedge x/P\).  Multiplying by \(P\) on both sides gives the desired bound.
\end{proof}

Recall that $x:=n^{-u}$ for simplicity.
From Lemma \ref{lem:help-inequality}, if \(x/P\to\infty\), then
\(1-(1-x)^{1/P} \asymp 1\) is bounded above and below by positive constants.  If \(x/P=O(1)\), then
\(1-(1-x)^{1/P}\asymp x/P\). We then start to analyze each term in \eqref{eq:help4} by considering two cases.

\begin{itemize}
    \item First, we suppose \(0<u<\alpha+q\).  For the dominant token $w_{t,n}^{\star}$, \(P_{t,w^\star_{t,n}}=1-\Delta_n\asymp1\), so \(x/P_{t,w^\star_{t,n}}\asymp x\to0\).  Hence,
\(1-(1-x)^{1/P_{t,w^\star_{t,n}}}\asymp x\), and the dominant token in the contribution to \eqref{eq:help4} is \(\asymp x=n^{-u}\).
For each light token, \(P_{t,w}\asymp n^{-(\alpha+q)}\), so
\(x/P_{t,w}\asymp n^{\alpha+q-u}\to\infty\).  Therefore,
\(1-(1-x)^{1/P_{t,w}}\asymp1\), and each light token contributes \(\asymp P_{t,w}\asymp n^{-(\alpha+q)}\). Since \(|L_{t,n}|\asymp n^\alpha\), the total light contribution to the sum in~\eqref{eq:help4} is \(\asymp n^{-q}\).  Finally, the core contribution is at most its total mass, \(\sum_{w\in C_{t,n}}P_{t,w}=\Delta_n^{\mathrm{core}}\lesssim n^{-q}\), which is smaller than the contribution of light tokens.  Combining these three
bounds, the signal tail is \(\asymp n^{-u}+n^{-q}\) when \(0<u<\alpha+q\).

\item Second, we suppose \(u\ge\alpha+q\).  The dominant contribution is again \(\asymp x=n^{-u}\), since the dominant probability $P_{t,w^\star_{t,n}}$ is of constant order.  For each light token $w \in L_{t,n}^{\star}$, \(x/P_{t,w}=O(1)\), so
\(1-(1-x)^{1/P_{t,w}}\asymp x/P_{t,w}\).  Thus each light token contributes \(\asymp x\) to the sum in~\eqref{eq:help4}, and the
total light contribution is \(|L_{t,n}|x\asymp n^\alpha n^{-u}=n^{\alpha-u}\). For the core set $C_{t,n}$, we use the upper bound \(1-(1-x)^{1/P_{t,w}}\lesssim x/P_{t,w}\), which gives \(P_{t,w}\{1-(1-x)^{1/P_{t,w}}\}\lesssim x\) for every core token.  Hence, the core contribution is at most \(|C_{t,n}|x\asymp n^{r-u}=o(n^{\alpha-u})\), since \(r<\alpha\).  Therefore, the light contribution dominates, and the signal tail is \(\asymp n^{\alpha-u}\) when \(u\ge\alpha+q\).

\end{itemize}

The bounds above are uniform over all admissible \(\bP_t\), so they remain valid after removing the
conditioning on \(\bP_t\).  This proves the lemma.
\end{proof}


\begin{proof}[Proof of Lemma~\ref{lem:bayes-tail-rule-v3}]
For a fixed coordinate $t$, conditional on $Y_t$, the conditional risk is
\[
\mathbb E[\theta_t(1-\delta_t)+\lambda(1-\theta_t)\delta_t\mid Y_t]
=\pi_t(Y_t)(1-\delta_t)+\lambda(1-\pi_t(Y_t))\delta_t,
\]
where $\pi_t(Y_t):=\mathbb P(\theta_t=1\mid Y_t)$.  Choosing $\delta_t=1$ is better than choosing
$\delta_t=0$ exactly when $\pi_t(Y_t)\ge\lambda/(1+\lambda)$.

It remains to show that the event
\(\{\pi_t(Y_t)\ge \lambda/(1+\lambda)\}\) is a right-tail event.  Let
\(f_{1,t}\) denote the density of \(Y_t\) conditional on \(\theta_t=1\), after averaging over the possible NTP distributions at position \(t\):
\(f_{1,t}(y):=\mathbb E[f_{1,\bP_t}(y)\mid \theta_t=1]\).  Under the null, the density is
\(f_0(y)\equiv1\).  Therefore, by Bayes' rule, for \(y\in(0,1)\),
\[
\pi_t(y)
=
\mathbb P(\theta_t=1\mid Y_t=y)
=
\frac{\mathbb P(\theta_t=1) f_{1,t}(y)}
{\mathbb P(\theta_t=0) f_0(y)+\mathbb P(\theta_t=1) f_{1,t}(y)}
=
\frac{\mathbb P(\theta_t=1) f_{1,t}(y)}
{\mathbb P(\theta_t=0)+\mathbb P(\theta_t=1) f_{1,t}(y)} .
\]
The map \(x\mapsto \mathbb P(\theta_t=1)x/
\{\mathbb P(\theta_t=0)+\mathbb P(\theta_t=1)x\}\) is increasing in \(x\).  Hence, it is enough to check that \(f_{1,t}(y)\) is nondecreasing in \(y\).  This follows from Lemma~\ref{lem:alt-pivot-density}: for each admissible \(\bP_t\),
\(f_{1,\bP_t}(y)=\sum_{w\in\mathcal W_n} y^{1/P_{t,w}-1}\), and every summand is nondecreasing because \(1/P_{t,w}-1\ge0\).  Averaging over \(\bP_t\) preserves monotonicity, so \(f_{1,t}\), and therefore \(\eta_t\), is nondecreasing in \(y\).

Therefore, the Bayes rejection region is an upper interval in \(Y_t\), say
\(\{Y_t>y_{n,t}^\star\}\), up to irrelevant boundary ties.  Writing
\(1-y_{n,t}^\star=n^{-u_{n,t}}\), with \(u_{n,t}\in[0,\infty]\), gives
\(\delta_t^\star=\mathbf 1\{Y_t>1-n^{-u_{n,t}}\}\).
\end{proof}

\begin{proof}[Proof of Lemma~\ref{lem:threshold-reduction}]
Fix a coordinate $t$ and write the original local rule as $\delta_t=\phi_t(Y_t)$.  Let
$a_{n,t}:=\mathbb P(\delta_t=1\mid \theta_t=0)$.  Since $Y_t\mid\{\theta_t=0\}\sim \mathrm{Unif}(0,1)$, the right-tail rule
$\tilde\delta_t:=\mathbf 1\{Y_t>1-a_{n,t}\}$ has the same null rejection probability as $\delta_t$.

It remains to compare the rejection probabilities under the signal label.  Let $g_t$ denote the density of $Y_t$ conditional on $\theta_t=1$, after averaging over the unobserved NTP distribution and other randomness.  By the watermarked-component law,
$g_t(y)$ is a mixture of densities of the form $f_{1,\bP}(y)=\sum_{w\in\mathcal W_n} y^{1/P_w-1}$.  Each such density is nondecreasing in $y$, and hence $g_t$ is also nondecreasing in $y$. 

Now let $B_t:=\{y:\phi_t(y)=1\}$.  Since the null density is uniform, the null rejection probability of $\delta_t$ is the Lebesgue measure of $B_t$, namely $|B_t|=a_{n,t}$.  Among all measurable sets of Lebesgue measure $a_{n,t}$, the integral of the nondecreasing density $g_t$ is maximized by the upper-tail set $(1-a_{n,t},1)$.  Therefore,
\[
\mathbb P(\tilde\delta_t=1\mid\theta_t=1)
=
\int_{1-a_{n,t}}^1 g_t(y)\,\rd y
\ge
\int_{B_t} g_t(y)\,\rd y
=
\mathbb P(\delta_t=1\mid\theta_t=1).
\]
At the same time, by construction,
$\mathbb P(\tilde\delta_t=1\mid\theta_t=0)=\mathbb P(\delta_t=1\mid\theta_t=0)$.
Multiplying these two inequalities by $\mathbb P(\theta_t=0)$ and $\mathbb P(\theta_t=1)$, respectively, and summing over $t$, we obtain
\[
\mathrm{EFP}_{\tilde\Bdelta}\le \mathrm{EFP}_{\Bdelta},
\qquad
\mathrm{ETP}_{\tilde\Bdelta}\ge \mathrm{ETP}_{\Bdelta}.
\]
Thus, every local rule is dominated by a coordinatewise right-tail threshold rule in the stated sense.  Finally, if $a_{n,t}>0$, we may write $a_{n,t}=n^{-u_{n,t}}$; the cases $a_{n,t}=0$ and $a_{n,t}=1$ correspond to the conventions $u_{n,t}=\infty$ and $u_{n,t}=0$. 
\end{proof}

\begin{proof}[Proof of Lemma~\ref{lem:mixing-rare-count}]
Since $0\le Z_t\le1$, $\sum_t\mathrm{Var}(Z_t)\le m_n$.  For $k\ge1$, by definition,
\[
|\mathrm{Cov}(Z_t,Z_{t+k})|
\le4\alpha_Z(k)
\le4C\rho^k,
\]
and also $|\mathrm{Cov}(Z_t,Z_{t+k})|\le\mathbb E(Z_tZ_{t+k})\le\mathbb E Z_t$.  Hence
\[
|\mathrm{Cov}(Z_t,Z_{t+k})|
\le \min\{\mathbb E Z_t,4C\rho^k\}.
\]
Summing this bound over $t$ and $k$ gives
\[
\mathrm{Var}\Big(\sum_{t=1}^n Z_t\Big)
\le m_n+2\sum_{k=1}^{n-1}\min\{m_n,4Cn\rho^k\}.
\]
Let \(K_n:=\min\{k:4Cn\rho^k\le m_n\}\).  Since \(\rho\in(0,1)\), we have
\(K_n=O(\log n)\).  For \(k<K_n\), the summand is at most \(m_n\), so these lags contribute
at most \(m_nK_n=O(m_n\log n)\).  For \(k\ge K_n\), the summand is at most \(4Cn\rho^k\), and
the geometric tail is bounded by
\[
\sum_{k\ge K_n}4Cn\rho^k
\lesssim n\rho^{K_n}
\lesssim m_n,
\]
where the last inequality follows from the definition of \(K_n\).  Therefore
\[
\sum_{k=1}^{n-1}\min\{m_n,4Cn\rho^k\}=O(m_n\log n),
\]
which proves~\eqref{eq:mixing-rare-count-var}.  If \(m_n/\log n\to\infty\), Chebyshev's inequality gives
\[
\mathbb P\Big(\sum_t Z_t=0\Big)
\le
\mathbb P\Big(\Big|\sum_tZ_t-m_n\Big|\ge m_n\Big)
\le
\frac{\mathrm{Var}(\sum_tZ_t)}{m_n^2}
\lesssim
\frac{\log n}{m_n}
\to0 .
\]
\end{proof}

\subsection{Proof of Theorem~\ref{thm:class}: Impossibility of classification}

\begin{proof}[Proof of Theorem~\ref{thm:class}]
  Suppose, for contradiction, that a local rule \(\Bdelta=(\delta_1,\ldots,\delta_n)\) achieves classification.  The following lemma records what this implies for its false and true discoveries.

\begin{lem}[Consequences of successful classification]
\label{lem:classification-count-implications}
If a sequence of local rules \(\Bdelta=(\delta_1,\ldots,\delta_n)\) achieves classification, then
\[
\mathrm{ETP}_{\Bdelta}=(1+o(1))\mathbb E|I|,
\qquad
\mathrm{EFP}_{\Bdelta}=o(\mathbb E|I|),
\]
where \(I=\{t:\theta_t=1\}\).  Moreover, \(\mathbb E|I|\asymp n\varepsilon_n\), with
\(n\varepsilon_n\asymp n\) in the dense case \(p=0\).
\end{lem}

By Lemma~\ref{lem:classification-count-implications}, this means that
\(\mathrm{ETP}_{\Bdelta}=(1+o(1))\mathbb E|I|\) and
\(\mathrm{EFP}_{\Bdelta}=o(\mathbb E|I|)\).  By Lemma~\ref{lem:threshold-reduction}, it is enough to
consider right-tail threshold rules of the form
\(\delta_t=\mathbf 1\{Y_t>1-a_{n,t}\}\), where \(0\le a_{n,t}\le1\).  Set \(A_n:=\sum_{t=1}^n a_{n,t}\).  Since \(Y_t\mid\{\theta_t=0\}\sim\mathrm{Unif}(0,1)\), the null
rejection probability at coordinate \(t\) is \(\PB(\delta_t = 1 \mid \theta_t = 0) = a_{n,t}\), and hence by definition,
\[
\mathrm{EFP}_{\Bdelta} = \sum_{t=1}^n\PB(\theta_t=0) \cdot \PB(\delta_t = 1 \mid \theta_t = 0)  \asymp A_n.
\]

We then derive a contradiction.
First assume \(q>0\).  We use the signal-tail upper bound from Lemma~\ref{lem:tail-pivot-HCL}:
for every threshold level \(a\in[0,1]\),
\(\mathbb P(Y_t>1-a\mid\theta_t=1)\lesssim a+n^{-q}\).  Applying this with \(a=a_{n,t}\) gives
\[
\mathrm{ETP}_{\Bdelta}=
\sum_{t=1}^n\PB(\theta_t=1) \cdot \PB(\delta_t = 1 \mid \theta_t = 1)
\lesssim
\varepsilon_n\sum_{t=1}^n(a_{n,t}+n^{-q})
=
\varepsilon_n A_n+n\varepsilon_n n^{-q}.
\]
On the other hand, classification requires
\(\mathrm{ETP}_{\Bdelta}=(1+o(1))\mathbb E|I|\asymp n\varepsilon_n\).  Dividing the preceding upper bound by \(\varepsilon_n\), we obtain \(A_n+n^{1-q}\gtrsim n\).  Since \(q>0\), we have \(n^{1-q}=o(n)\), and therefore \(A_n\gtrsim n\).  It follows that \(\mathrm{EFP}_{\Bdelta}\asymp A_n\gtrsim n\).  This contradicts Lemma~\ref{lem:classification-count-implications}, which requires \(\mathrm{EFP}_{\Bdelta}=o(\mathbb E|I|)\).  Indeed, when \(p>0\), \(\mathbb E|I|\asymp n\varepsilon_n=o(n)\); when \(p=0\), \(\mathbb E|I|\asymp n\), so the same condition still requires \(\mathrm{EFP}_{\Bdelta}=o(n)\).  In both cases, \(\mathrm{EFP}_{\Bdelta}\gtrsim n\) is impossible.

It remains to consider \(q=0\).  In this case, Assumption~\ref{asmp:HCL+} gives a constant \(\eta_\star\in(0,1)\) such that \(\max_{w}P_{t, w} = P_{t,w^\star_{t,n}}\ge \eta_\star\) for all large \(n\) and all \(t\).  Hence, for a right-tail threshold with null rejection probability \(a\), the proof of Lemma~\ref{lem:tail-pivot-HCL} implies that \(\mathbb P(Y_t>1-a\mid \theta_t=1)\le 1-\eta_\star+Ca\) for some constant \(C>0\).  Applying this bound with \(a=a_{n,t}\) and summing over \(t\), we obtain \(\mathrm{ETP}_{\Bdelta}\le (1-\eta_\star)\mathbb E|I|+C\varepsilon_n A_n\).
On the other hand, classification requires \(\mathrm{EFP}_{\Bdelta}\asymp A_n=o(\mathbb E|I|)\), and hence \(C\varepsilon_n A_n=o(\mathbb E|I|)\).  Therefore the preceding upper bound gives \(\mathrm{ETP}_{\Bdelta}\le (1-\eta_\star+o(1))\mathbb E|I|\), which contradicts the necessary condition \(\mathrm{ETP}_{\Bdelta}=(1+o(1))\mathbb E|I|\) from Lemma~\ref{lem:classification-count-implications}.  This proves that classification is impossible also when \(q=0\), and hence for all \(p,q\in[0,1]\) and \(\alpha\in[0,1)\).
\end{proof}

\begin{proof}[Proof of Lemma~\ref{lem:classification-count-implications}]
By Assumption~\ref{asmp:main}(c),
\(\mathbb E|I|=\sum_{t=1}^n\mathbb P(\theta_t=1)\asymp n\varepsilon_n\).  If
\(\mathrm{MDR}_{\Bdelta}\to0\), then
\[
\mathrm{EFN}_{\Bdelta}
=
o(\mathrm{EFN}_{\Bdelta}+\mathrm{ETP}_{\Bdelta})
=
o(\mathbb E|I|),
\]
because \(\mathrm{EFN}_{\Bdelta}+\mathrm{ETP}_{\Bdelta}= \mathbb E|I|\).  Therefore
\[
\mathrm{ETP}_{\Bdelta}
=
\mathbb E|I|-\mathrm{EFN}_{\Bdelta}
=
(1+o(1))\mathbb E|I|.
\]
If also \(\mathrm{mFDR}_{\Bdelta}\to0\), then
\[
\frac{\mathrm{EFP}_{\Bdelta}}
{\mathrm{EFP}_{\Bdelta}+\mathrm{ETP}_{\Bdelta}}
\to0.
\]
Since \(\mathrm{ETP}_{\Bdelta}\asymp\mathbb E|I|\), this implies
\(\mathrm{EFP}_{\Bdelta}=o(\mathbb E|I|)\).  This proves the lemma.
\end{proof}

\subsection{Proof of Theorem~\ref{thm:adaptivityTrGoF}: Adaptive optimality for global detection}
\label{sec:th3.45proof}

Let $p_t:=1-Y_t$ be the p-value.  Under the null, $p_t\sim\mathrm{Unif}(0,1)$.  For $r\in(0,1)$, define the empirical lower tail
\[
\mathbb F_n(r):=\frac1n\sum_{t=1}^n\mathbf 1\{p_t\le r\}
=\frac1n\sum_{t=1}^n\mathbf 1\{Y_t\ge1-r\}.
\]
For any \(r\in(0,1)\), we also write $\mathbb E_1[\mathbb F_n(r)]
= \frac1n \sum_{t=1}^n \mathbb P_1(p_t\le r)$ for its expectation under the alternative hypothesis $H_1$.
As a reminder, for fixed $s\in[-1,2]$, the \texttt{Tr-GoF} statistic is
\[
S_n^+(s):=\sup_{r\in[p_n^+,1)}K_s^+(\mathbb F_n(r),r),
\]
where $K_s^+$ is the upper-tail one-sided Jager--Wellner divergence (defined below) and $p_n^+ = \sup\{p_{(i)}: p_{(i)} \le c_n^+ \}$ is the stability truncation used by the procedure with $0 \le c_n^+ \le \frac1n$.  The test rejects $H_0$ if $nS_n^+(s)$ is larger than a threshold.

\begin{defn}[One-sided Jager--Wellner divergence]
$K_s^+(\cdot, \cdot)$ is defined in the following way.
\begin{equation}
\label{eq:Ks+}
K_s^+(u, v) = \begin{cases}
K_s(u, v), &~~\text{if}~~0 < v < u < 1,\\
0, &\new{~~\text{otherwise}},
\end{cases}
\end{equation}
where $K_s(u, v)$ represents the $\phi_s$-divergence between $\mathrm{Ber}(u)$ and $\mathrm{Ber}(v)$:\footnote{$\mathrm{Ber}(u)$ denotes a Bernoulli distribution with parameter (or head probability) $u$.}
\begin{align*}
 \label{eq:Ks}
 \begin{split}
 K_s(u, v) &= D_{\phi_s}(\mathrm{Ber}(u) \| \mathrm{Ber}(v))
 = v \phi_s\left( \frac{u}{v} \right) + (1-v) \phi_s\left( \frac{1-u}{1-v} \right).
 \end{split}
\end{align*}
Here, the scalar function $\phi_s(x)$, indexed by $s \in \RB$, is convex in $x$ and is defined by \citep{jager2007goodness}:
\begin{equation}
\label{eq:phi}
\phi_s(x) = \begin{cases}
x \log x -x +1, & ~~\text{if}~~s = 1, \\
\frac{1-s+sx-x^s}{s(1-s)}, & ~~\text{if}~~s \neq 0, 1, \\
-\log x +x -1,& ~~\text{if}~~s = 0.
\end{cases}
\end{equation}
\end{defn}

\begin{proof}[Proof of Theorem~\ref{thm:adaptivityTrGoF}]
The Type~I error follows directly from Lemma~\ref{lem:trgof-null-comparison}, because under
\(H_0\), \(nS_n^+(s)\) is at most of order \(\log\log n\) with probability tending to one.

\begin{lem}[Null calibration of $K_s^+$]
\label{lem:trgof-null-comparison}
Under $H_0$ in Definition \ref{def:detection}, for every fixed $\delta>0$ and $s\in[-1,2]$,
\begin{equation*}
\mathbb P_0\big(nS_n^+(s)>(1+\delta)\log\log n\big)\to0.
\end{equation*}
\end{lem}
  It remains to prove that the Type~II error vanishes.  We do this by finding, in each detectable regime, a
deterministic tail point \(r_n\) inside the range (that is
\(r_n\in[p_n^+,1)\)) such that the empirical upper-tail excess
\(\mathbb F_n(r_n)-r_n\) is much larger than the stochastic fluctuation
\(\mathbb F_n(r_n)-\mathbb E_1\mathbb F_n(r_n)\) under the alternative $H_1$. To that end, we will use the following lemma.

\begin{lem}[Tail empirical concentration]
\label{lem:tail-count-concentration}
Under Assumption~\ref{asmp:alphamixing}, if $r_n=n^{-u}$ with any fixed $u<1$ that satisfies $\mathbb E_1 [\mathbb F_n(r_n)] \lesssim r_n$, then under the alternative $H_1$ in Definition \ref{def:detection},
\begin{equation}
\label{eq:tail-count-concentration}
\mathbb F_n(r_n)-\mathbb E_1 [\mathbb F_n(r_n)]
=O_{\mathbb P}\left(\sqrt{\frac{r_n}{n}}\log n\right).
\end{equation}
The same bound holds uniformly over any polylogarithmic grid of such $r_n$'s.
\end{lem}

\paragraph{Case 1: $0 < p < \alpha$.}
First, we consider the case where \(p<\alpha\).  The detection condition reduces to \(p+q<1\).  Choose a small fixed
\(\eta>0\) such that \(p+q+\eta<1\), and set \(r_n=n^{-u}\) with \(u=p+q-\eta\).  Then
\(u\in(0,\alpha+q)\).  By Lemma~\ref{lem:tail-pivot-HCL}, the signal tail probability, namely
\(\mathbb P(Y_t>1-r_n\mid\theta_t=1)\), is of order \(n^{-u}+n^{-q}\), while the null tail probability
is \(r_n=n^{-u}\).  Therefore,
\[
\mathbb E_1\mathbb F_n(r_n)-r_n
\asymp 
\varepsilon_n n^{-q}
\asymp
n^{-(p+q)}.
\]
Because \(r_n=n^{-(p+q-\eta)}\gg n^{-(p+q)}\), we then have $\mathbb E_1\mathbb F_n(r_n) = r_n+O(n^{-(p+q)}) \lesssim r_n$.
The concentration in Lemma \ref{lem:tail-count-concentration} shows that with probability tending to one,
\begin{equation}
\label{eq:help5}
\mathbb F_n(r_n)-\mathbb E_1\mathbb F_n(r_n) =O_{\mathbb P}\left(\sqrt{\frac{r_n}{n}}\log n\right) =o_{\mathbb P}(r_n).
\end{equation}
This stochastic fluctuation is negligible compared with the mean excess \(n^{-(p+q)}\), since
\[
\frac{\sqrt{r_n/n}\log n}{n^{-(p+q)}}
=
n^{(p+q)-(1+u)/2}\log n
=
n^{(p+q-1+\eta)/2}\log n
\to0.
\]
Thus, with probability tending to one, \(\mathbb F_n(r_n)-r_n \asymp n^{-(p+q)}\).  
Moreover, by the last equation in \eqref{eq:help5}, 
\(\mathbb F_n(r_n)-\mathbb E_1\mathbb F_n(r_n)=o_{\mathbb P}(r_n)\).  Hence, $\frac{\mathbb F_n(r_n)}{r_n}=O_{\mathbb P}(1)$ as a result of $\mathbb E_1\mathbb F_n(r_n) \lesssim r_n$.
On the other hand, the standardized tail deviation implies that under $H_1$,
\[
\frac{\sqrt n\{\mathbb F_n(r_n)-r_n\}}{\sqrt{r_n}}
\gtrsim
n^{1/2-(p+q)+u/2}
=
n^{(1-p-q-\eta)/2}
\to\infty.
\]
Finally, \(r_n\) is included in the scan $[p_{n}^+, 1)$. Indeed, the truncation condition gives
\(p_n^+\le c_n^+ \le \frac1n \lesssim\varepsilon_n\Delta_n\asymp n^{-(p+q)}\), while
\(r_n=n^{-(p+q-\eta)}\gg n^{-(p+q)}\).

\paragraph{Case 2: $p\ge\alpha$.}
Next suppose \(p\ge\alpha\).  The detection condition is \(2p+q-\alpha<1\) and we will take \(r_n=n^{-(\alpha+q)}\).  This is the transition point in Lemma~\ref{lem:tail-pivot-HCL}.  Similarly, we still have
\[
\mathbb E_1\mathbb F_n(r_n)-r_n
\asymp
\varepsilon_n n^{-q}
\asymp
n^{-(p+q)}.
\]
 Since \(r_n=n^{-(\alpha+q)}\) and \(p\ge\alpha\), we also have \(\mathbb E_1\mathbb F_n(r_n)=r_n+O(n^{-(p+q)})\lesssim r_n\).
From Lemma \ref{lem:tail-count-concentration}, the stochastic fluctuation  is \(\mathbb F_n(r_n)-\mathbb E_1 \mathbb F_n(r_n) = O_{\mathbb P}(\sqrt{r_n/n}\log n)\), which is negligible relative to \(n^{-(p+q)}\), because
\[
\frac{\sqrt{r_n/n}\log n}{n^{-(p+q)}}
=
n^{p+q-(1+\alpha+q)/2}\log n
=
n^{(2p+q-\alpha-1)/2}\log n
\to0.
\]
On the other hand, we also have \(\mathbb F_n(r_n)-\mathbb E_1\mathbb F_n(r_n)=o_{\mathbb P}(r_n)\) because $n r_n = n^{1-\alpha-q} \to \infty$.
As a result, we have $\frac{\mathbb F_n(r_n)}{r_n}=O_{\mathbb P}(1).$
Finally, \(\mathbb F_n(r_n)-r_n\gtrsim n^{-(p+q)}\) with probability tending to one, and
\[
\frac{\sqrt n\{\mathbb F_n(r_n)-r_n\}}{\sqrt{r_n}}
\gtrsim
n^{1/2-(p+q)+(\alpha+q)/2}
=
n^{(1+\alpha-q-2p)/2}
\to\infty.
\]
We assert that \(r_n=n^{-(\alpha+q)}\) here is also in the scan $[p_{n}^+, 1)$. Indeed, the detection condition \(2p+q-\alpha<1\) and \(p\ge\alpha\) imply \(\alpha+q<1\), and hence \(r_n=n^{-(\alpha+q)}\gg n^{-1}\ge p_n^+\).

We are now ready to conclude the proof.  In every regime considered above, we have shown that $\frac{\mathbb F_n(r_n)}{r_n} \lesssim 1$ and found a tail point
\(r_n\in[p_n^+,1)\) such that
\[
\frac{\sqrt n\{\mathbb F_n(r_n)-r_n\}}{\sqrt{r_n}}
\to\infty
\qquad\text{in probability}.
\]
Since the deviation is positive with probability tending to one, the upper-tail divergence satisfies
\[
nK_2^+(\mathbb F_n(r_n),r_n)
=
\frac{n\{\mathbb F_n(r_n)-r_n\}^2}{2r_n(1-r_n)}
\to\infty
\qquad\text{in probability}.
\]

\begin{lem}[Local comparison lemma {\cite[Lemma~A.11]{li2025robust}}]
\label{lem:Ks-comparison}
On any moderate-deviation range where $v/u > c$ for some $c \in (0, 1)$, $K_s^+(u,v)$ is bounded below by a positive constant times $K_2^+(u,v)$.
More specifically, $K_2^+(u, v) = \frac{(u-v)^2}{2v(1-v)}$ and 
\begin{enumerate}
\item For any \(s\le 2\) with \(s\neq 1\), it follows that
\[
K_s^+(u,v)
\ge
K_2^+(u,v)
\left[
1-(1-v)\left\{1-\left(\frac{v}{u}\right)^{2-s}\right\}
\right].
\]
\item For \(s=1\), it follows that
\[
K_1^+(u,v)
\ge
K_2^+(u,v)\cdot \frac{v}{u}.
\]
\end{enumerate}
\end{lem}

We are now ready to conclude the proof.  In both regimes, we have shown that \(\mathbb F_n(r_n)/r_n=O_{\mathbb P}(1)\) and that \(\sqrt n\{\mathbb F_n(r_n)-r_n\}/\sqrt{r_n}\to\infty\) in probability.  Since \(\mathbb F_n(r_n)-r_n>0\) with probability tending to one, the upper-tail divergence satisfies \(nK_2^+(\mathbb F_n(r_n),r_n)=n\{\mathbb F_n(r_n)-r_n\}^2/\{2r_n(1-r_n)\}\to\infty\) in probability.

Moreover, the preceding bound \(\mathbb F_n(r_n)/r_n=O_{\mathbb P}(1)\) implies that, with probability tending to one, \(\mathbb F_n(r_n)/r_n\le M\) for some constant \(M>1\).  On this event, with \(u=\mathbb F_n(r_n)\) and \(v=r_n\), we have \(0<v<u<1\) and \(v/u\ge M^{-1}\). Therefore, Lemma~\ref{lem:Ks-comparison} gives \(K_s^+(\mathbb F_n(r_n),r_n)\gtrsim K_2^+(\mathbb F_n(r_n),r_n)\) with probability tending to one.  Since \(S_n^+(s)\) is the supremum over all scanned tail points and \(r_n\in[p_n^+,1)\), we obtain \[
nS_n^+(s)
\ge
nK_s^+(\mathbb F_n(r_n),r_n)
\gtrsim
nK_2^+(\mathbb F_n(r_n),r_n)
\to\infty
\]
in probability.  Thus \(nS_n^+(s)\) exceeds the null critical order \(\log\log n\) with probability tending to one, and the Type~II error vanishes.

\end{proof}

\paragraph{Auxiliary proofs for Theorem~\ref{thm:adaptivityTrGoF}.}
In the following, we present the omitted proof for Lemma~\ref{lem:trgof-null-comparison} used above.

\begin{proof}[Proof of Lemma~\ref{lem:trgof-null-comparison}]
The proof follows from~\citet{li2025robust}; we include the details for completeness.  Under \(H_0\), the \(p\)-values are i.i.d. \(\mathrm{Unif}(0,1)\).  Moreover, \(S_n^+(s)\) is bounded above by the corresponding untruncated Jager--Wellner statistic \(S_n(s)\), that is, \(S_n^+(s)\le S_n(s)\) almost surely for all \(s\).  This inequality makes sense because \(S_n(s)\) is defined in the same way as \(S_n^+(s)\), except that it uses \(K_s(\cdot,\cdot)\) instead of the one-sided truncated divergence \(K_s^+(\cdot,\cdot)\).  By Theorem~3.1 of~\citet{jager2007goodness},
\[
\mathbb P_0(nS_n(s)\le (1+\delta)\log\log n)\to 1
\]
for every fixed \(\delta>0\) and \(s\in[-1,2]\).  As a result, the desired null bound follows.
\end{proof}

\begin{proof}[Proof of Lemma~\ref{lem:tail-count-concentration}]
Fix \(r_n=n^{-u}\) with \(u<1\), and write
\(X_t:=\mathbf 1\{p_t\le r_n\}\).  Then
\(\mathbb F_n(r_n)=n^{-1}\sum_{t=1}^nX_t\).  The variables \(X_t\) are bounded by one and are geometrically strongly mixing by Assumption~\ref{asmp:alphamixing}.  Moreover, \(\mathbb E X_t=\mathbb P(p_t\le r_n)\lesssim r_n+n^{-q}\) from Lemma \ref{lem:tail-pivot-HCL}, and on the tail points used in the proof, this is \(O(r_n)\); in particular, the variance scale is at most of order \(r_n\).

A Bernstein inequality for bounded geometrically mixing sequences with bounded expectation $\mathbb E_1 [\mathbb F_n(r_n)] \lesssim r_n$ (e.g., \citep[Theorem~2]{MerlevedePeligradRio2009Bernstein}) gives, for a constant \(c>0\),
\[
\mathbb P\left(\left|\mathbb F_n(r_n)-\mathbb E\mathbb F_n(r_n)\right|>x\right)
\le
2\exp\left[-c\frac{n x^2}{r_n+x\log^2 n}\right].
\]
Now take \(x=C\sqrt{r_n/n}\log n\).  Since \(r_n=n^{-u}\) with \(u<1\), we have
\(n r_n\to\infty\), and
\[
\frac{x\log^2 n}{r_n}
=
C\frac{\log^3 n}{\sqrt{n r_n}}
\to0.
\]
Hence, the denominator $r_n+x\log^2 n$ is \(r_n(1+o(1))\), and the exponent in the failure probability is at least a constant multiple of
\(n x^2/r_n=C^2\log^2 n\).  Choosing \(C\) large enough gives a probability \(o(n^{-A})\) for any
fixed \(A>0\).  This proves with probability at least $1-o(n^{-A})$,
\[
\mathbb F_n(r_n)-\mathbb E\mathbb F_n(r_n)
=
O_{\mathbb P}\left(\sqrt{\frac{r_n}{n}}\log n\right).
\]

For the uniform version over a polylogarithmic grid, apply the same bound to each grid point.  Since
the grid has only \((\log n)^{O(1)}\) points, the union bound preserves the same order, after
increasing the constant \(C\) if necessary.
\end{proof}

\subsection{Proof of Theorem~\ref{thm:adaptivity}: Adaptive optimality for discovery}
\label{sec:th3.5proof}

As a reminder, we use the following notations:
\[
\tau(u):=1-n^{-u},
\qquad
\widehat S_n(u):=\frac1n\sum_{t=1}^n\mathbf 1\{Y_t>\tau(u)\},
\qquad
S_n(u):=\mathbb E\widehat S_n(u).
\]
Let \(\mathcal U_n\) be the grid used by \Algo{}, with mesh \(\Delta_{u,n}\) and cardinality
\(M_n\asymp(\log n)^a\).  Define the oracle and empirical tail false-discovery proxies by
\[
T_n^{\mathrm{orc}}(u):=\frac{(1-\varepsilon_n)n^{-u}}{S_n(u)},
\qquad
\widehat T_n(u):=\frac{(1-\widehat\varepsilon_n)n^{-u}}{\widehat S_n(u)\vee n^{-1}}.
\]
The proof has three steps.  First, we control the empirical tail counts
\(\widehat S_n(u)\) uniformly over the grid \(\mathcal U_n\).  Second, combining this concentration
with the accuracy of \(\widehat\varepsilon_n\), we show that the empirical tail false-discovery proxy
\(\widehat T_n(u)\) uniformly approximates the oracle \(T_n^{\mathrm{orc}}(u)\).  Third, we locate
the crossing point of \(T_n^{\mathrm{orc}}(u)\) at level \(\lambda_n\), and transfer this localization to
the threshold selected by \Algo{}.

\begin{lem}[Uniform tail-count concentration on the \(u\)-grid]
\label{lem:tail-relative-u}
Under Assumptions~\ref{asmp:main}--\ref{asmp:alphamixing}, if \(u_{\max}<1\), then
\[
\mathbb P\left(
\max_{u\in\mathcal U_n}
\frac{|\widehat S_n(u)-S_n(u)|}{S_n(u)}>\frac{c}{(\log n)^2}
\right)=o(n^{-1})
\]
for a sufficiently small universal constant \(c>0\).
\end{lem}

\begin{lem}[Uniform tail proxy stability]
\label{lem:tail-LFDR-stability-u}
Suppose the assumptions of Theorem~\ref{thm:adaptivity} hold.  
Let $\eta_n = C_2/(\log n)^2$ for a fixed constant $C_2 > 0$.
Then
\[
\mathbb P\left(
\max_{u\in\mathcal U_n}|\widehat T_n(u)-T_n^{\mathrm{orc}}(u)|>\eta_n/2
\right)=o(n^{-1}).
\]
\end{lem}

\begin{lem}[Oracle calibration]
\label{lem:oracle-jump-calibration}
Suppose Assumptions~\ref{asmp:main}--\ref{asmp:HCL+} hold.  Assume \(p<\alpha\) and
\[
0<u_{\min}<p+q<u_{\max}<1,
\qquad
p+q<\alpha+q,
\]
and let \(\mathcal U_n\subset[u_{\min},u_{\max}]\) be the grid used by \emph{\Algo{}}, with mesh
\(\Delta_{u,n}=(u_{\max}-u_{\min})/(M_n-1)\to0\).  Then, uniformly for \(u\) in any fixed
neighborhood of \(p+q\) contained in \((u_{\min},u_{\max})\cap(0,\alpha+q)\),
\[
S_n(u):=\mathbb P(Y_t>\tau(u))\asymp n^{-u}+n^{-(p+q)},
\qquad
T_n^{\mathrm{orc}}(u)=\frac{(1-\varepsilon_n)n^{-u}}{S_n(u)}
\asymp \frac{n^{-u}}{n^{-u}+n^{-(p+q)}}.
\]
Define \(\lambda_n=C_1/\log n\) and the oracle exponent $u_n^\circ$ by
\[
u_n^\circ:=\inf\{u\in\mathcal U_n:T_n^{\mathrm{orc}}(u)\le\lambda_n\}.
\]
Then, \(u_n^\circ\) is well-defined for all large \(n\) and
\[
u_n^\circ
=
p+q+\frac{\log\log n}{\log n}
+O\!\left(\frac1{\log n}+\Delta_{u,n}\right).
\]
The same expansion holds if \(\lambda_n\) is replaced by any level in
\([\lambda_n/2,2\lambda_n]\).
\end{lem}

\begin{proof}[Proof of Theorem~\ref{thm:adaptivity}]
Let
\[
\widehat u_n:=\inf\{u\in\mathcal U_n:\widehat T_n(u)\le\lambda_n-\eta_n\}
\]
be the exponent selected by \Algo{}, and let \(\Bdelta_{\Algo}\) be the resulting decision rule.  The
first step is to show that this data-driven exponent is close to the oracle exponent.  On the event
in Lemma~\ref{lem:tail-LFDR-stability-u}, we have
\[
\max_{u\in\mathcal U_n}|\widehat T_n(u)-T_n^{\mathrm{orc}}(u)|\le \eta_n/2.
\]
Since \(\eta_n=o(\lambda_n)\), for all large \(n\) we have
\(\lambda_n-3\eta_n/2\in[\lambda_n/2,2\lambda_n]\) and
\(\lambda_n-\eta_n/2\in[\lambda_n/2,2\lambda_n]\).  Thus the empirical selection rule is sandwiched
between oracle rules at levels in \([\lambda_n/2,2\lambda_n]\): if the empirical proxy is below
\(\lambda_n-\eta_n\), then the oracle proxy is below \(\lambda_n-\eta_n/2\), and if the oracle proxy
is below \(\lambda_n-3\eta_n/2\), then the empirical proxy is below \(\lambda_n-\eta_n\).  Applying
Lemma~\ref{lem:oracle-jump-calibration} to these two oracle levels gives
\[
\widehat u_n
=
p+q+\frac{\log\log n}{\log n}
+O_{\mathbb P}\left(\frac1{\log n}+\Delta_{u,n}\right).
\]
In particular, since \(p<\alpha\), this implies
\(\widehat u_n\in(p+q,\alpha+q)\) with probability tending to one.

Let \(r_n:=C(1/\log n+\Delta_{u,n})\) with \(C\) sufficiently large, and define the good event
\[
\mathcal E_n:=\left\{
u_n^{-} \le \widehat u_n\le u_n^+
\right\}
\quad
\text{with}
\quad
u_n^{\pm}:=p+q+\frac{\log\log n}{\log n} \pm r_n
\]
By the preceding display, $\mathbb P(\mathcal E_n^c)=o(n^{-1})$.  We introduce $\mathcal E_n$ because $\widehat u_n$ is data-dependent, so the tail bounds in Lemma \ref{lem:tail-pivot-HCL} should not be applied directly at the random exponent $\widehat u_n$.  On $\mathcal E_n$, however, the data-driven threshold is sandwiched between two deterministic thresholds with exponents $u_n^-$ and $u_n^+$.  
Both $u_n^-$ and $u_n^+$ are deterministic such that Lemma \ref{lem:tail-pivot-HCL} can still apply.
The contribution from $\mathcal E_n^c$ is $o(1)$, since $\mathbb P(\mathcal E_n^c)=o(n^{-1})$ and all counts are bounded by $n$.


 We now bound the expected false and true positives on \(\mathcal E_n\).  
\paragraph{Analysis of {EFP}.} Note that on \(\mathcal E_n\), the rule with exponent $u_n^-$ gives an upper bound for false positives. As a result,
\[
\begin{aligned}
\mathrm{EFP}_{\Bdelta_{\Algo}} 
& = \mathbb E\left[
\sum_{t=1}^n
(1-\theta_t)\mathbf 1\{Y_t>1-n^{-\widehat u_n}\} \right] \\
&\le \mathbb E\left[
\sum_{t=1}^n (1-\theta_t)\mathbf 1\{Y_t>1-n^{-u_n^-}\} \right] + n\mathbb P(\mathcal E_n^c) \\
&=
\sum_{t=1}^n
\mathbb P(\theta_t=0)\,
\mathbb P\!\left(Y_t>1-n^{-u_n^-}\mid \theta_t=0\right)
+
n\mathbb P(\mathcal E_n^c).
\end{aligned}
\]
Since \(Y_t\mid\{\theta_t=0\}\sim\mathrm{Unif}(0,1)\), the null tail probability from Lemma \ref{lem:tail-pivot-HCL} is \(n^{- u_n^-}\).  Therefore,
$ \mathrm{EFP}_{\Bdelta_{\Algo}} \lesssim n\,n^{- u_n^-} = n^{1-u_n^-}. $
Because $ u_n^- = p+q+\frac{\log\log n}{\log n} +o(1)$, we have
\[
n^{1- u_n^-}
=
n^{1-p-q}
\cdot
n^{-\log\log n/\log n}
\cdot
n^{o(1)}
=
\frac{n^{1-p-q+o(1)}}{\log n}.
\]
The above argument, together with $\mathbb P(\mathcal E_n^c)  = o(\frac1n)$, implies that 
\begin{equation}
\label{eq:our-EFP}
\mathrm{EFP}_{\Bdelta_{\Algo}}
\lesssim n^{1- u_n^-} + o(1)
\lesssim
\frac{n^{1-p-q+o(1)}}{\log n} + o(1).
\end{equation}

\paragraph{Analysis of {ETP}.}
Similarly, note that on $\mathcal{E}_n$, the rule with exponent $u_n^+$ gives a lower bound for true positives and nonempty discovery. Here, we have
\[
\begin{aligned}
\mathrm{ETP}_{\Bdelta_{\Algo}} 
& = \mathbb E\left[
\sum_{t=1}^n \theta_t\mathbf 1\{Y_t>1-n^{-\widehat u_n}\} \right] 
\ge \mathbb E\left[
\sum_{t=1}^n \theta_t\mathbf 1\{Y_t>1-n^{-u_n^+}\} \1\{\mathcal{E}_n\} \right]  \\
&\ge \mathbb E\left[
\sum_{t=1}^n \theta_t\mathbf 1\{Y_t>1-n^{-u_n^+}\} \right] 
- n \PB(\mathcal{E}_n^c) \\
&= \sum_{t=1}^n \mathbb P(\theta_t=1)\,\mathbb P\!\left(Y_t>1-n^{-u_n^+}\mid \theta_t=1\right)-n \PB(\mathcal{E}_n^c).
\end{aligned}
\]
Lemma~\ref{lem:tail-pivot-HCL} gives $\mathbb P\big(Y_t>1-n^{- u_n^+}\mid\theta_t=1\big) \asymp n^{- u_n^+}+n^{-q}.$
Since \(u_n^+>p+q\ge q\), we have \(n^{-u_n^+}=o(n^{-q})\), so the \(n^{-q}\) term
dominates.  Using \(\mathbb P(\theta_t=1)\asymp\varepsilon_n\) and $\mathbb P(\mathcal E_n^c)  = o(\frac1n)$, we obtain
\begin{equation}
\label{eq:our-ETP}
\mathrm{ETP}_{\Bdelta_{\Algo}}
\gtrsim n\varepsilon_n n^{-q} + o(1)
\asymp n^{1-p-q} + o(1).
\end{equation}
Combining the above two results in \eqref{eq:our-EFP} and \eqref{eq:our-ETP}, we have \(\mathrm{ETP}_{\Bdelta_{\Algo}}\to\infty\) and $\mathrm{ETP}_{\Bdelta_{\Algo}}/\mathrm{EFP}_{\Bdelta_{\Algo}}\to\infty$. Hence \(\mathrm{mFDR}_{\Bdelta_{\Algo}}\to0\).

\paragraph{Non-trivial discovery set.}
It remains to prove that \(\Algo{}\) makes at least one discovery with probability tending to one.
On \(\mathcal E_n\), \(\widehat u_n\le u_n^+\), where
\(u_n^+:=p+q+\log\log n/\log n+r_n\).  Since a smaller exponent corresponds to a lower threshold
and hence a larger selected set, the discovery set of \(\Algo{}\) contains the discovery set of the
deterministic comparison rule \(\mathbf 1\{Y_t>1-n^{-u_n^+}\}\) on \(\mathcal E_n\).  Moreover,
\(u_n^+<\alpha+q\) for all large \(n\), and Lemma~\ref{lem:tail-pivot-HCL} gives this comparison
rule expected true positives of order
\[
n\varepsilon_n(n^{-u_n^+}+n^{-q})\asymp n^{1-p-q}.
\]
Since \(p+q<1\), this expectation diverges.  By Lemma~\ref{lem:mixing-rare-count}, the comparison
rule makes at least one true discovery with probability tending to one.  Therefore
\(\mathbb P(|S_{\Bdelta_{\Algo}}|\ge1)\to1\).  Together with
\(\mathrm{mFDR}_{\Bdelta_{\Algo}}\to0\), this proves that \(\Algo{}\) achieves discovery.

\end{proof}

\paragraph{Auxiliary proofs for Theorem~\ref{thm:adaptivity}.}
In the following, we present the omitted proofs for the lemmas used above.

\begin{proof}[Proof of Lemma~\ref{lem:tail-relative-u}]
Fix a grid point \(u\in\mathcal U_n\), and write
\(X_t(u):=\mathbf 1\{Y_t>\tau(u)\}\).  Then
\(\widehat S_n(u)=n^{-1}\sum_{t=1}^nX_t(u)\) and \(S_n(u)=n^{-1}\sum_{t=1}^n\mathbb E X_t(u)\).
The variables \(X_t(u)\) are bounded by one and are geometrically strongly mixing by
Assumption~\ref{asmp:alphamixing}, as a result of the fact that $\mathcal{Y}_t \subset \mathcal{F}_t$ for any $t \ge 1$. Also, since the null tail is \(n^{-u}\) and
\(\mathbb P(\theta_t=0)\) is bounded away from zero, we have
\(S_n(u)\ge c_0n^{-u}\ge c_0n^{-u_{\max}}\).  Thus \(nS_n(u)\) grows polynomially because
\(u_{\max}<1\).

Therefore, we apply a Bernstein-type inequality for bounded, geometrically strongly mixing sequences; see, for example, \citep[Theorem~2]{MerlevedePeligradRio2009Bernstein}. This inequality applies in our setting under the uniform mixing coefficients established above. Concretely, there exist constants
\(c_1,c_2>0\),
\[
\mathbb P\left(\left|\sum_{t=1}^n\{X_t(u)-\mathbb E X_t(u)\}\right|>y\right)
\le
2\exp\left\{-c_1\frac{y^2}{nS_n(u)+y\log^2 n}\right\}.
\]
Taking \(y=xnS_n(u)\) gives that for any $u \in \mathcal{U}_n$,
\[
\mathbb P\left(|\widehat S_n(u)-S_n(u)|>xS_n(u)\right)
\le
2\exp\left\{-c_1\frac{nS_n(u)x^2}{1+x\log^2 n}\right\}.
\]
Now set \(x=c/(\log n)^2\).  Since \(nS_n(u)\ge c_0n^{1-u_{\max}}\), the exponent is at least a
constant multiple of \(n^{1-u_{\max}}/(\log n)^4\), which is larger than any multiple of \(\log n\).
Therefore the probability above is \(o(n^{-2})\) uniformly over \(u\in\mathcal U_n\) for sufficiently large $n$.  Since \(|\mathcal U_n|=M_n\asymp(\log n)^a\), a union bound over the grid gives the desired \(o(n^{-1})\) probability.
\end{proof}

\begin{proof}[Proof of Lemma~\ref{lem:tail-LFDR-stability-u}]
Let \(\mathcal E_S := \left\{ \max_{u\in\mathcal U_n}
\frac{|\widehat S_n(u)-S_n(u)|}{S_n(u)}\le \frac{c}{(\log n)^2}  \right\} \) be the event in Lemma~\ref{lem:tail-relative-u}, and let
\(\mathcal E_\varepsilon:=\{|\widehat\varepsilon_n-\varepsilon_n|\le c\eta_n\}\).  By the accuracy assumption and Lemma~\ref{lem:tail-relative-u}, we have
\[
\mathbb P(\mathcal E_S^c\cup\mathcal E_\varepsilon^c)=o(n^{-1})
\]
We work on the event \(\mathcal E_S\cap\mathcal E_\varepsilon\).  On that event, uniformly over
\(u\in\mathcal U_n\), write \(\widehat S_n(u)=S_n(u)(1+e_u)\), where
\(|e_u|\le c/(\log n)^2\).  For \(c\) sufficiently small (which we will further minimize later), \(1+e_u\ge1/2\).  Also,
\(S_n(u)\ge c_0n^{-u}\ge c_0n^{-u_{\max}}\), because \(u\le u_{\max}\).  Since
\(u_{\max}<1\), we have \(n^{-u_{\max}}\gg n^{-1}\).  Hence, for all large \(n\), $\widehat S_n(u)\ge \frac{c_0}{2}n^{-u_{\max}}>n^{-1},$ so the truncation in the denominator is inactive:
\(\widehat S_n(u)\vee n^{-1}=\widehat S_n(u)\).

For such \(u\), we can write
\[
\widehat T_n(u)
=
\frac{(1-\widehat\varepsilon_n)n^{-u}}{S_n(u)(1+e_u)},
\qquad
T_n^{\mathrm{orc}}(u)
=
\frac{(1-\varepsilon_n)n^{-u}}{S_n(u)}.
\]
Subtracting the two expressions gives
\[
|\widehat T_n(u)-T_n^{\mathrm{orc}}(u)|
\le
\frac{n^{-u}}{S_n(u)}
\left|
\frac{1-\widehat\varepsilon_n}{1+e_u}
-
(1-\varepsilon_n)
\right|.
\]
Since \(S_n(u)\ge c_0n^{-u}\), the prefactor \(n^{-u}/S_n(u)\) is uniformly bounded.  For the
remaining term, add and subtract \(1-\widehat\varepsilon_n\):
\[
\left|
\frac{1-\widehat\varepsilon_n}{1+e_u}
-
(1-\varepsilon_n)
\right|
\le
|\widehat\varepsilon_n-\varepsilon_n|
+
(1-\widehat\varepsilon_n)\left|\frac{1}{1+e_u}-1\right|.
\]
Because \(|e_u|\le c/(\log n)^2\le1/2\), we have
\(\left|(1+e_u)^{-1}-1\right|\le 2|e_u|\).  Therefore, uniformly over \(u\in\mathcal U_n\),
\[
|\widehat T_n(u)-T_n^{\mathrm{orc}}(u)|
\lesssim
|\widehat\varepsilon_n-\varepsilon_n|+|e_u|
\le
c\eta_n+\frac{c}{(\log n)^2}.
\]
Since \(\eta_n=C_2/(\log n)^2\), the right-hand side is at most \(\eta_n/2\) by taking this fixed constant \(c>0\) sufficiently small.  Thus
\[
\max_{u\in\mathcal U_n}
|\widehat T_n(u)-T_n^{\mathrm{orc}}(u)|
\le \eta_n/2
\]
on \(\mathcal E_S\cap\mathcal E_\varepsilon\).  Since the complement of this event has probability \(o(n^{-1})\), the lemma follows.
\end{proof}


\begin{proof}[Proof of Lemma~\ref{lem:oracle-jump-calibration}]
Fix \(u\) in a fixed neighborhood of \(p+q\) contained in \((0,\alpha+q)\).  By the definition of \(S_n(u)\), $S_n(u) = \frac1n\sum_{t=1}^n \mathbb P(Y_t>\tau(u)).$
For each \(t\), decompose the tail probability according to \(\theta_t\):
\[
\mathbb P(Y_t>\tau(u))
=
\mathbb P(\theta_t=0)\mathbb P(Y_t>\tau(u)\mid\theta_t=0)
+
\mathbb P(\theta_t=1)\mathbb P(Y_t>\tau(u)\mid\theta_t=1).
\]
The null tail $\mathbb P(\theta_t=0)\mathbb P(Y_t>\tau(u)\mid\theta_t=0)$ is \(n^{-u}\).  Since \(u<\alpha+q\), Lemma~\ref{lem:tail-pivot-HCL} gives the signal
tail order \( \mathbb P(Y_t>\tau(u)\mid\theta_t=1) \asymp n^{-u}+n^{-q}\).  Assumption~\ref{asmp:main}(c) gives
\(\mathbb P(\theta_t=1)\asymp\varepsilon_n\) uniformly in \(t\), and thus
\[
S_n(u)
\asymp
n^{-u}+\varepsilon_n(n^{-u}+n^{-q})
\asymp
n^{-u}+n^{-(p+q)}.
\]
Note that the term \(\varepsilon_n n^{-u}\) is absorbed into \(n^{-u}\), since \(\varepsilon_n\le1\).
Therefore, we have
\[
T_n^{\mathrm{orc}}(u)
=
\frac{(1-\varepsilon_n)n^{-u}}{S_n(u)}
\asymp
\frac{n^{-u}}{n^{-u}+n^{-(p+q)}}.
\]

We now determine the oracle exponent \(u_n^\circ\), namely the smallest value of \(u\) for which \(T_n^{\mathrm{orc}}(u)\le \lambda_n\).  If \(u<p+q\), then \(n^{-u}\gg n^{-(p+q)}\), so
\(T_n^{\mathrm{orc}}(u)\asymp1\), which is much larger than
\(\lambda_n=C_1/\log n\).  Hence, the oracle crossing cannot occur below \(p+q\).  If
\(u>p+q\), then \(n^{-u}\ll n^{-(p+q)}\), and the preceding display becomes
\(T_n^{\mathrm{orc}}(u)\asymp n^{-(u-(p+q))}\).  Therefore, the crossing at level
\(\lambda_n\) is determined by the boundary
\[
n^{-(u-(p+q))}\asymp \frac{1}{\log n}
\quad\Longleftrightarrow\quad
(u-(p+q))\log n\approx \log\log n.
\]
Taking logarithms gives
\[
u
=
p+q+\frac{\log\log n}{\log n}+O\!\left(\frac1{\log n}\right),
\]
where the \(O(1/\log n)\) term absorbs constants such as \(C_1\) and the implicit constants in
\(\asymp\).

This gives the continuous crossing location.  Since \(\Algo{}\) scans the grid \(\mathcal U_n\), the
first grid point at which \(T_n^{\mathrm{orc}}(u)\le\lambda_n\) can differ from the continuous
crossing location by at most one grid mesh, namely \(\Delta_{u,n}\).  Thus
\[
u_n^\circ
=
p+q+\frac{\log\log n}{\log n}
+O\!\left(\frac1{\log n}+\Delta_{u,n}\right).
\]

If \(\lambda_n\) is replaced by a level \(L_n\in[\lambda_n/2,2\lambda_n]\), then
\(L_n=C_L/\log n\) for some constant \(C_L\in[C_1/2,2C_1]\).  Solving
\(n^{-(u-(p+q))}\asymp L_n\) similarly gives $u = p+q+\frac{\log\log n}{\log n} + O\!\left(\frac1{\log n}\right)$, because the constant \(C_L\) only contributes an \(O(1/\log n)\) term after taking logarithms.
Thus, the same expansion holds.
\end{proof}

\subsection{Proof of Theorem~\ref{thm:25optimal}: Near-optimal number of discoveries}

\begin{proof}[Proof of Theorem~\ref{thm:25optimal}]
Let \(\Bdelta_{\Algo}\) be the rule returned by \Algo{}.  From the proof of
Theorem~\ref{thm:adaptivity}, namely \eqref{eq:our-EFP} and \eqref{eq:our-ETP}, we have
\[
\mathrm{EFP}_{\Bdelta_{\Algo}}
\lesssim
\lambda_n n^{1-p-q} + o(1),
\qquad
\mathrm{ETP}_{\Bdelta_{\Algo}}
\gtrsim
n^{1-p-q} + o(1),
\]
where we use the facts that \(\lambda_n=C_1/\log n\).
As a result, for a constant \(c_1>0\),
\[
\mathrm{mFDR}_{\Bdelta_{\Algo}} = \frac{\mathrm{EFP}_{\Bdelta_{\Algo}}}
{\mathrm{EFP}_{\Bdelta_{\Algo}}+\mathrm{ETP}_{\Bdelta_{\Algo}}} \le c_1\lambda_n+o(1).
\]

It remains to upper bound the number of true discoveries achievable by any homogeneous local rule
satisfying the same false-discovery constraint.  By the monotone likelihood-ratio rearrangement in
Lemma~\ref{lem:threshold-reduction}, it is enough to consider right-tail threshold rules.  Let
\(\Bdelta\in \mathcal{D}_n^{\mathrm{hom}}(\lambda_n)\), and let \(A_n\) be its total null rejection mass.  Then
\(\mathrm{EFP}_{\Bdelta}\asymp A_n\) as we show many times in the previous proof.  Lemma~\ref{lem:tail-pivot-HCL} gives the same uniform signal-tail upper bound as \eqref{eq:bound-ETP}:
\begin{equation}
\label{eq:previous-display}
\mathrm{ETP}_{\Bdelta}
\lesssim
\varepsilon_n(A_n+n^{1-q}).
\end{equation}
Since \(\Bdelta\in \mathcal{D}_n^{\mathrm{hom}}(\lambda_n)\), we have
\(\mathrm{mFDR}_{\Bdelta}
=\mathrm{EFP}_{\Bdelta}/(\mathrm{EFP}_{\Bdelta}+\mathrm{ETP}_{\Bdelta})
\le\lambda_n\) by definition.  Equivalently,
\((1-\lambda_n)\mathrm{EFP}_{\Bdelta}\le\lambda_n\mathrm{ETP}_{\Bdelta}\).  Since
\(\lambda_n\to0\), for all large \(n\) we have \(1-\lambda_n\ge1/2\), and therefore
\(\mathrm{EFP}_{\Bdelta}\le2\lambda_n\mathrm{ETP}_{\Bdelta}\).  Because
\(A_n\asymp\mathrm{EFP}_{\Bdelta}\), this gives
\(A_n\lesssim\lambda_n\mathrm{ETP}_{\Bdelta}\).

Substituting this \(A_n\lesssim\lambda_n\mathrm{ETP}_{\Bdelta}\) into the previous display in \eqref{eq:previous-display} yields
\[
\mathrm{ETP}_{\Bdelta}
\precsim \varepsilon_n\lambda_n\mathrm{ETP}_{\Bdelta} + n^{1-p-q}.
\]
Since \(\varepsilon_n\lambda_n\le\lambda_n\to0\), the first term on the right-hand side can be summarized as $o(1) \cdot \mathrm{ETP}_{\Bdelta}$.
Arranging the last inequality, we have that for every \(\Bdelta\in \mathcal{D}_n^{\mathrm{hom}}(\lambda_n)\) satisfies
\[
\mathrm{ETP}_{\Bdelta}
\lesssim
n^{1-p-q}.
\]
Taking the supremum over the benchmark class gives
\[
\sup_{\Bdelta\in \mathcal{D}_n^{\mathrm{hom}}(\lambda_n)}
\mathrm{ETP}_{\Bdelta}
\lesssim
n^{1-p-q}.
\]

On the other hand, the lower bound from Theorem~\ref{thm:adaptivity} gives
\(\mathrm{ETP}_{\Bdelta_{\Algo}}\gtrsim n^{1-p-q}+o(1)\).  Combining the upper bound for the benchmark
class with this lower bound for \(\Algo{}\), we obtain a constant \(c_2>0\) such that
\[
\mathrm{ETP}_{\Bdelta_{\Algo}}
\ge
c_2
\sup_{\Bdelta\in \mathcal{D}_n^{\mathrm{hom}}(\lambda_n)}
\mathrm{ETP}_{\Bdelta}
+o(1).
\]
Together with the false-discovery bound for \(\Bdelta_{\Algo}\), this proves the theorem.
\end{proof}